\documentclass[letterpaper, journal, twoside]{IEEEtran}

\usepackage[dvipsnames]{xcolor}
\usepackage{cite}

\usepackage{graphics} 
\usepackage{epsfig} 
\usepackage{amsmath} 
\newcommand\numberthis{\addtocounter{equation}{1}\tag{\theequation}}
\usepackage{amssymb} 
\usepackage{mathtools}
\usepackage{xspace}
\usepackage{algpseudocode}
\usepackage{centernot}
\usepackage{algorithm}
\usepackage{epstopdf}
\usepackage{tcolorbox}
\usepackage{varwidth}

\usepackage{verbatim}

\usepackage{amsthm}

\usepackage{cancel}

\usepackage{enumitem}

\usepackage{tikz} 
\usepgflibrary{arrows}
\usepackage{eso-pic}
\usepackage[hypertexnames=false]{hyperref}
\usepackage{support_caption}
\usepackage[labelformat=simple]{subcaption}

\usepackage[normalem]{ulem}
\usepackage{enumitem}
\usepackage[capitalise]{cleveref}
\crefname{assumption}{Assumption}{Assumptions}

\DeclareMathOperator*{\argmax}{arg\,max}
\DeclareMathOperator*{\argmin}{arg\,min}
\newcommand{\manish}[1]{\textcolor{red}{Manish: #1}}

\newcommand{\red}[1]{\textcolor{black}{#1}}
\newcommand{\green}[1]{{\color{Black} #1}}
\newcommand{\brown}[1]{{\color{Black} #1}}
\newcommand{\blue}[1]{{\color{Black} #1}}
\newcommand{\yellow}[1]{{\color{Black} #1}}
 
\newcommand{\revdisp}[1]{}

\newcommand{\noisevar}{\varsigma}

\newcommand{\coninput}{a}
\newcommand{\prob}{p}
\newcommand{\tsample}{t'}

\newcommand{\inputSpace}{\mathcal{A}}
\newcommand{\pwcinput}{\mathcal{A}_{\textsc{\tiny{PC}}\xspace}}
\newcommand{\horizon}{H}
\newcommand{\horizonmid}{H^\prime}
\newcommand{\timeX}{T_c}

\newcommand{\xpos}{x_p}
\newcommand{\ypos}{y_p}
\newcommand{\dxpos}{\dot{x}_p}
\newcommand{\dypos}{\dot{y}_p}

\newcommand{\dyn}{f}
\newcommand{\identity}{\mathbb{I}}

\newcommand{\mpc}{\textsc{\small{MPC}}\xspace}
\newcommand{\sempc}{\textsc{\small{SageMPC}}\xspace}
\newcommand{\sempcL}{\textsc{\small{SageMPC-L}}_\constrain\xspace}
\newcommand{\dist}{\textsc{\small{DIST}}\xspace}

\newcommand{\nlp}{\textsc{\small{NLP}}\xspace}

\newcommand{\AvR}{\textsc{\small{AvR}}\xspace}
\newcommand{\seocp}{\textsc{\small{SageOC}}\xspace}

\newcommand{\SQP}{\textsc{\small{SQP}}\xspace}

\newcommand{\expansion}[1][]{\mathcal{G}_{#1}}
\newcommand{\ball}[1]{\mathcal{B}_{#1}}

\newcommand{\GP}{\textsc{\small{GP}}\xspace}
\newcommand{\GPs}{\textsc{\small{GP}}s\xspace}
\newcommand{\BO}{\textsc{\small{BO}}\xspace}
\newcommand{\gpucb}{\textsc{\small{GP-UCB}}\xspace}

\newcommand{\Nat}{\mathbb{N}}
\newcommand{\calP}{\mathcal P}

\newcommand{\A}{\mathcal A}

\newcommand{\X}{\mathcal X}

\renewcommand{\S}{\mathcal S}

\newcommand{\E}{\mathbb E}

\newcommand{\N}{\mathcal N}

\newcommand{\R}{\mathbb{R}}

\newtheorem{lemma}{Lemma}
\newtheorem{assumption}{Assumption}

\newtheorem{definition}{Definition}
\newtheorem{theorem*}{Theorem}
\newtheorem{corollary}{Corollary}
\newtheorem{proposition}{Proposition}
\newtheorem{theorem}{Theorem}
\newtheorem{remark}{Remark}
\newtheorem{objective}{Objective}

\newcommand{\LocAgents}{X}
\newcommand{\state}{x}

\newcommand{\condLocAgent}[3][]{x_{#2|#3}^{#1}}
\newcommand{\condInput}[3][]{a_{#2|#3}^{#1}}

\newcommand{\Domain}{\mathcal{X}}
\newcommand{\PtInDomain}{x}

\newcommand{\constrain}{q}
\newcommand{\utility}{\rho}
\newcommand{\noise}{\eta}

\newcommand{\LipConst}{L_\constrain}

\newcommand{\Rcontoper}[2][]{\mathcal{R}_{T}({{#1}},{#2})}
\newcommand{\Rpathoper}[2][]{\mathcal{P}({{#1}},{#2})}

\newcommand{\ubutility}[1][]{u^{\utility}_{#1}}

\newcommand{\ubconst}[1][]{u_{#1}}
\newcommand{\lbconst}[1][]{l_{#1}}
\newcommand{\muconst}[1][]{\mu_{#1}}
\newcommand{\sigconst}[1][]{\sigma_{#1}}

\newcommand{\epsconst}{\epsilon}
\newcommand{\epsutility}{\epsilon_{\utility}}
\newcommand{\noiseconst}{\noisevar^{-2}}

\newcommand{\betaconst}[1][]{\beta_{#1}}

\newcommand{\gammaconst}[1]{\gamma_{#1}}

\newcommand{\LpessiSet}[2][]{\mathcal{S}_{#2}^{ {p\!}_{L}} #1}
\newcommand{\pessiSet}[2][]{\mathcal{S}_{#2}^{ p #1}}

\newcommand{\optiSet}[2][]{\mathcal{S}_{#2}^{ o, \epsconst #1}}
\newcommand{\LconstSet}[2][]{\mathcal{S}_{#2}^{{\constrain\!}_L #1}}
\newcommand{\constSet}[2][]{\mathcal{S}_{#2}^{ \constrain #1}}

\newcommand{\sumMaxwidth}[2][]{w^{#1}_{#2}}

\newcommand{\kernelfunc}{k}

\newcommand{\safeInit}[1]{\mathbb{\LocAgents}_{#1}}
\newcommand{\n}{n}
\newcommand{\nfin}{{n^\prime}}
\newcommand{\Bq}{B_q}

\newcommand{\closure}[1]{\overline{#1}}
\newcommand{\interior}[1]{{#1}^\circ}
\newcommand{\gostateOpti}[1][]{\state^{g,o}_{#1}}
\newcommand{\gostatePessi}[1][]{\state^{g,p}_{#1}}
\newcommand{\sdim}{p}

\input{v_line_patch}
\begin{document}
\bstctlcite{IEEEexample:BSTcontrol}

\title{Safe Guaranteed Exploration for Non-linear Systems}
\author{Authors}

\author{Manish Prajapat, Johannes K\"ohler, Matteo Turchetta, Andreas Krause$^\dagger$, Melanie N. Zeilinger$^\dagger$ \thanks{\looseness -1 $\dagger$ Joint supervision. All authors are from ETH Zurich. [manishp, jkoehle, matteotu, krausea, mzeilinger]@ethz.ch. Manish Prajapat is supported by ETH AI center, Johannes K\"ohler and Matteo Turchetta by the Swiss National Science Foundation under NCCR Automation, grant agreement 51NF40 180545.}}

\maketitle
\AddToShipoutPicture*{%
  \put(50,0){%
    \parbox[b][\paperheight]{\textwidth}{%
      \vfill
      \centering
      \footnotesize
      \copyright~2025 IEEE. All rights reserved, including rights for text and data mining and training of artificial intelligence and similar technologies. Personal use is permitted, but republication/redistribution requires IEEE permission. See \href{https://www.ieee.org/publications/rights/index.html}{https://www.ieee.org/publications/rights/index.html} for more information.
      \vspace{10pt} 
    }
  }
}
\AddToShipoutPicture*{%
  \put(50,\dimexpr\paperheight-30pt\relax){%
    \parbox{\textwidth}{%
      \centering
      \footnotesize
      This article has been accepted for publication in IEEE Transactions on Automatic Control. Citation information: DOI 10.1109/TAC.2025.3541577
    }
  }
}

\thispagestyle{empty}
\begin{abstract}
    \looseness 0 Safely exploring environments with \emph{a-priori} unknown constraints is a fundamental challenge that restricts the autonomy of robots. While safety is paramount, guarantees on sufficient exploration are also crucial for ensuring autonomous task completion. To address these challenges, we propose a novel {\em safe guaranteed exploration} framework using optimal control, which achieves first-of-its-kind results: guaranteed exploration for non-linear systems with finite time sample complexity bounds, while being provably safe with arbitrarily high probability.
    The framework is general and applicable to many real-world scenarios with complex non-linear dynamics and unknown domains.
    We improve the efficiency of this general framework by proposing an algorithm, \sempc, SAfe Guaranteed Exploration using Model Predictive Control. \sempc leverages three key techniques: 
    \textit{i)} exploiting a Lipschitz bound, \textit{ii)} goal-directed exploration, and \textit{iii)} receding horizon style re-planning, all while maintaining the desired sample complexity, safety and exploration guarantees of the framework. Lastly, we demonstrate safe efficient exploration in challenging unknown environments using \sempc with a car model.
\end{abstract}

\vspace{-0.5em}
\section{Introduction}
\looseness -1 \emph{Motivation:} 
A core challenge limiting the complete autonomy of robots is safe exploration. If robots are to be deployed in the wild, they must learn through interaction with their environment, while guaranteeing safe operation.
Safe exploration is essential in real-world tasks such as mapping or navigating unknown environments with a priori unknown hazards 
\cite{kulkarni2022autonomous}, search and rescue \cite{delmerico2019current,tomic2012toward}, industrial inspection \cite{hutter2018towards}, surveillance \cite{grocholsky2006cooperative}, to name a few. 
Most robotic systems of interest, e.g., walking or driving robots or flying quadrotors, have highly non-linear dynamics, which makes the safe exploration process much harder both theoretically and practically.

\looseness -1 Safety can often be addressed by taking highly conservative actions (such as remaining stationary), which however may not be sufficient for completing the task. Latter may require {\em exploration} -- gathering observations to learn about which states may be safely reachable to complete the task.
The key question is thus how one can effectively gather information about a-priori unknown constraints to 
{\em guarantee exploration}, all while maintaining the robot's safety throughout the learning process. 
The guaranteed exploration can either represent maximum domain exploration or sufficient exploration for task completion (e.g., reach a point). 
To this end, our goal is to develop a control framework that guarantees exploration while ensuring safety at all times in \textit{a-priori} unknown environments for general non-linear systems.





%
\begin{figure}
\setlength{\abovecaptionskip}{5.5pt}
    \hspace{-0.90em}\scalebox{1}{\tikzset{every picture/.style={line width=0.75pt}} 

\begin{tikzpicture}[x=0.75pt,y=0.75pt,yscale=-1,xscale=1]

\draw  [line width=0.75]  (61.94,57.15) .. controls (61.94,52.98) and (65.32,49.6) .. (69.5,49.6) -- (203.68,49.6) .. controls (207.85,49.6) and (211.24,52.98) .. (211.24,57.15) -- (211.24,85.2) .. controls (211.24,89.37) and (207.85,92.76) .. (203.68,92.76) -- (69.5,92.76) .. controls (65.32,92.76) and (61.94,89.37) .. (61.94,85.2) -- cycle ;
\draw  [line width=0.75]  (235.03,57.48) .. controls (235.03,53.35) and (238.38,50) .. (242.51,50) -- (377.52,50) .. controls (381.65,50) and (385,53.35) .. (385,57.48) -- (385,85.27) .. controls (385,89.41) and (381.65,92.76) .. (377.52,92.76) -- (242.51,92.76) .. controls (238.38,92.76) and (235.03,89.41) .. (235.03,85.27) -- cycle ;
\draw  [draw opacity=0][fill={rgb, 255:red, 155; green, 155; blue, 155 }  ,fill opacity=0.18 ][line width=0.75]  (64,129.96) .. controls (64,121.84) and (70.59,115.25) .. (78.71,115.25) -- (366.49,115.25) .. controls (374.61,115.25) and (381.2,121.84) .. (381.2,129.96) -- (381.2,174.09) .. controls (381.2,182.21) and (374.61,188.8) .. (366.49,188.8) -- (78.71,188.8) .. controls (70.59,188.8) and (64,182.21) .. (64,174.09) -- cycle ;
\draw [draw opacity=0] [fill={rgb, 255:red, 155; green, 155; blue, 155 }  ,fill opacity=0.18 ][line width=0.75]  (79,283.87) .. controls (79,278.22) and (83.58,273.64) .. (89.23,273.64) -- (358.77,273.64) .. controls (364.42,273.64) and (369,278.22) .. (369,283.87) -- (369,314.57) .. controls (369,320.22) and (364.42,324.8) .. (358.77,324.8) -- (89.23,324.8) .. controls (83.58,324.8) and (79,320.22) .. (79,314.57) -- cycle ;
\draw [line width=0.75]    (137.9,93) -- (137.9,111.8) ;
\draw [shift={(137.8,114.8)}, rotate = 270.26] [fill={rgb, 255:red, 0; green, 0; blue, 0 }  ][line width=0.08]  [draw opacity=0] (8.93,-4.29) -- (0,0) -- (8.93,4.29) -- cycle    ;
\draw [line width=0.75]    (313,93) -- (313,111.8) ;
\draw [shift={(313,114.8)}, rotate = 269.93] [fill={rgb, 255:red, 0; green, 0; blue, 0 }  ][line width=0.08]  [draw opacity=0] (8.93,-4.29) -- (0,0) -- (8.93,4.29) -- cycle    ;
\draw [line width=0.75]    (223,188.8) -- (223,270.6) ;
\draw [shift={(223.15,274.2)}, rotate = 270.17] [fill={rgb, 255:red, 0; green, 0; blue, 0 }  ][line width=0.08]  [draw opacity=0] (8.93,-4.29) -- (0,0) -- (8.93,4.29) -- cycle    ;
\draw  [line width=0.75]  (251,217.44) .. controls (251,212.78) and (254.78,209) .. (259.44,209) -- (380.56,209) .. controls (385.22,209) and (389,212.78) .. (389,217.44) -- (389,242.76) .. controls (389,247.42) and (385.22,251.2) .. (380.56,251.2) -- (259.44,251.2) .. controls (254.78,251.2) and (251,247.42) .. (251,242.76) -- cycle ;
\draw  [line width=0.75]  (58.25,217.47) .. controls (58.25,212.79) and (62.04,209) .. (66.72,209) -- (186.28,209) .. controls (190.96,209) and (194.75,212.79) .. (194.75,217.47) -- (194.75,242.88) .. controls (194.75,247.56) and (190.96,251.35) .. (186.28,251.35) -- (66.72,251.35) .. controls (62.04,251.35) and (58.25,247.56) .. (58.25,242.88) -- cycle ;
\draw [line width=0.75]    (122,252) -- (122,270.6) ;
\draw [shift={(122.2,273.6)}, rotate = 270.26] [fill={rgb, 255:red, 0; green, 0; blue, 0 }  ][line width=0.08]  [draw opacity=0] (8.93,-4.29) -- (0,0) -- (8.93,4.29) -- cycle    ;
\draw [line width=0.75]    (323,252) -- (323,270.6) ;
\draw [shift={(323,273.2)}, rotate = 269.74] [fill={rgb, 255:red, 0; green, 0; blue, 0 }  ][line width=0.08]  [draw opacity=0] (8.93,-4.29) -- (0,0) -- (8.93,4.29) -- cycle    ;
\draw (137.55,71.23) node  [font=\normalsize] [align=left] {\begin{minipage}[lt]{101.66pt}\setlength\topsep{0pt}
\begin{center}
Safe exploration with Gaussian process
\end{center}
\end{minipage}};

\draw (310.47,71.23) node  [font=\normalsize] [align=left] {\begin{minipage}[lt]{102.74pt}\setlength\topsep{0pt}
\begin{center}
Non-linear predictive control
\end{center}
\end{minipage}};

\draw (223.49,153) node  [font=\normalsize] [align=left] {\begin{minipage}[lt]{256.66pt}\setlength\topsep{0pt}
\begin{center}
\textbf{Safe guaranteed exploration framework} (\cref{sec:SE_theory})\\[0.5em]
\begin{varwidth}{\textwidth}
\begin{itemize}
    \item Sample complexity (\cref{thm:sample_complexity})
    \item Safety (\cref{coro:sample_complexity})
    \item Guaranteed exploration (\cref{thm:SE}) 
\end{itemize}
\end{varwidth}
\end{center}
\end{minipage}};

\draw (225,300) node  [font=\normalsize] [align=left] {\begin{minipage}[lt]{250.78pt}\setlength\topsep{0pt}
\begin{center}
{\Large S}{AGE}{\Large MPC}\\[0.3em]
Efficient algorithm using re-planning (\cref{sec:MPC})\\
\end{center}

\end{minipage}};
\draw (319.75,230) node  [font=\normalsize] [align=left] {\begin{minipage}[lt]{99.66pt}\setlength\topsep{0pt}
\begin{center}
Goal-directed exploration \!(\cref{sec:goalSE})
\end{center}

\end{minipage}};
\draw (127,230) node  [font=\normalsize] [align=left] {\begin{minipage}[lt]{99.66pt}\setlength\topsep{0pt}
\begin{center}
Exploiting Lipschitz bound (\cref{sec:SE_expanders})
\end{center}

\end{minipage}};

\end{tikzpicture}} 
\caption{ \looseness -1 We introduce a novel safe guaranteed exploration framework by leveraging tools from optimal control and Gaussian processes. Utilizing the framework, we propose \textsc{SageMPC}\! --\! an algorithm that incorporates Lipschitz bound, goal-directed exploration and re-planning with new information to enhance efficiency; all while keeping the theoretical guarantees of the framework.}
    \label{fig:flowchart} \vspace{-1em}
\end{figure} 

\looseness -1 \emph{Related work:} Given the importance of safety, many control methods have been developed to ensure safe operation for non-linear systems under uncertainty, e.g., using model predictive control (\mpc) \cite{rawlings2017model}, safety filters \cite{wabersich2021predictive,wabersich2023data}, control barrier functions \cite{ames2016control}, and Hamilton-Jacobi reachability \cite{bansal2017hamilton}. While most of the existing work considers known constraints, \cite{tordesillas2021faster,saccani2022multitrajectory,soloperto2023safe,batkovic2022safe} ensure safety in uncertain environments. To incentivise exploration, \cite{tordesillas2021faster,saccani2022multitrajectory,soloperto2023safe} plan optimistically while maintaining a backup plan to ensure safety. Although the strategy is sufficient to guarantee safety, it does not provide a principled approach to guarantee domain exploration. 

\looseness -1 A systematic way to model an uncertain environment, i.e., uncertain constraints, is through a Gaussian process (\GP) \cite{gp-Rasmussen} which is a popular tool to model unknown functions e.g., in Bayesian optimization (\BO) \cite{beta-srinivas,beta_chowdhury17a}, \mpc \cite{hewing2019cautious}, or experiment design \cite{prajapat2023submodular}.
For stateless problems (no dynamics), safe guaranteed exploration is studied in constrained BO with uncertain safety constraints modelled with \GPs~\cite{safe-bo-sui15,sui2018stagewise,berkenkamp2023bayesian}.
In \cite{turchetta2016safemdp}, these approaches are extended to safe domain exploration with a dynamical system. This approach was generalized to a goal-directed setting in \cite{turchetta2019safe}
 and further to coverage control in a multi-agent setting~\cite{prajapat2022near}.
Interestingly, \cite{turchetta2016safemdp,turchetta2019safe,prajapat2022near} guarantee exploration in finite time while ensuring safety. However, they consider discretized domains and simplistic dynamics (e.g., moving left or right) captured using underlying graph-based transitions. Applying these methods to the continuous domain (with appropriate discretization) is computationally intense, which imposes a bottleneck for practical application. 

\looseness -1 \emph{Contributions:} 
We introduce a novel framework of \textit{safe guaranteed exploration} using optimal control for non-linear systems. In the framework, we recursively determine a sufficiently informative sampling location in order to explore, while maintaining the safety of the non-linear dynamical system throughout the transient operation.
\begin{tcolorbox}[colframe=white!, top=2pt,left=2pt,right=2pt,bottom=2pt]
\centering
\emph{The framework guarantees exploration in finite time while being provably safe at all times for non-linear systems.}
\end{tcolorbox}
\looseness -1 In particular, to ensure the safety of dynamical systems, we steer the system to those safely \emph{reachable} locations from where it can safely \emph{return} back to a known safe invariant set. To ensure exploration in finite time, we provide a sample complexity result that bounds the maximum number of informative measurements that can be obtained. We next show that this number of measurements is sufficient to guarantee domain exploration for non-linear systems. 
The framework only assumes regularity of the constraint function (\cref{assump:q_RKHS,assump:sublinear}), a terminal set (\cref{assump: safe-init}), and regularity of reachable sets (\cref{assump: regularconnected,assump: equal_boundary}). Thus, our framework is fairly general and applicable to a wide range of real-world scenarios with complicated non-linear dynamics and unknown environments. 

\looseness 0 Our sample complexity analysis improves prior work addressing discrete domains \cite{safe-bo-sui15,sui2018stagewise,berkenkamp2023bayesian,turchetta2016safemdp,turchetta2019safe,prajapat2022near} by obtaining bounds that are independent of the discretization size. 
We establish guarantees on exploration using a different set of tools, i.e., optimal control problems and boundaries of continuous domains, which makes our algorithm and guarantees fundamentally different to existing work \cite{safe-bo-sui15,sui2018stagewise,berkenkamp2023bayesian,turchetta2016safemdp,turchetta2019safe,prajapat2022near}. Although our motivation stems from the non-linear dynamics of robots, our algorithm is also useful in other \BO-type settings, where parameters cannot be arbitrarily chosen but follow a dynamics, e.g., for tuning of free electron laser \cite{kirschner2019adaptive}.

\looseness -1 Building on this framework, we propose an efficient algorithm \sempc, {\em SAfe Guaranteed Exploration using Model Predictive Control} by integrating the following techniques:
\begin{enumerate}[label=\textit{\roman*})]
\item Exploiting Lipschitz bound of the constraints to enlarge the known safe region;
\item \looseness -1 Performing a goal-directed safe exploration, i.e., only exploring the region required for optimizing a function (e.g., reaching a goal) instead of exploring the full domain;
\item Re-planning in a receding horizon style as soon as new information is obtained.
\end{enumerate}
These techniques significantly enhance the efficiency of \sempc while keeping the desired theoretical guarantees of the framework, i.e., guaranteed exploration under finite time while being provably safe for non-linear systems.

\looseness -1 Lastly, we compare our provable algorithm using the non-linear unicycle model in simulations and further demonstrate it with a car model in a-priori unknown environments with complicated obstacle constraints as shown in the 
\href{https://github.com/manish-pra/sagempc?tab=readme-ov-file#simulation-videos}{video}.

\looseness -1 \emph{Notation:}  
For a set $S\subseteq\mathbb{R}^\sdim$ we denote the closure by $\closure{S} \coloneqq \{ x\in \R^\sdim | \forall r  > 0, \ball{r}(x) \cap S \neq \emptyset \}$, where $\ball{r}(x)=\{z\in\mathbb{R}^p|~\|x-z\|\leq r\}$ denotes a ball of radius $r$ around $x$.
The boundary of a set $S\subseteq\R^\sdim$ is denoted by $\partial S \coloneqq \{ x\in \R^\sdim | \forall r  > 0, \ball{r}(x) \cap S \neq \emptyset , \ball{r}(x) \backslash S \neq \emptyset\}$.
The interior of a set $S \subseteq \R^\sdim$ is denoted by $\interior{S} \coloneqq \{ x\in \R^\sdim | \forall r  > 0, \ball{r}(x) \cap S \neq \emptyset , \ball{r}(x) \backslash S = \emptyset\}$. 
A set $S\subseteq \R^\sdim$ is a path connected if $\forall x_0, x_1 \in S, \exists~\textit{a continuous}~\gamma:[0,1] \to S: \gamma(0)=\state_0, \gamma(1)=\state_1$.
$|A|$ denotes the cardinality of the set A. $\Nat$ denotes natural numbers. 
A sequence $\{\eta_n\}_{n=1}^\infty$ is conditionally $\noisevar$-sub-Gaussian for a fixed $\noisevar\geq0$ if $\forall n\geq 0, \forall \lambda \in \R, \E[e^{\lambda \eta_n}|\{\eta_i\}_{i=1}^{\n-1},\{\state_i\}_{i=1}^{\n}] \leq \mathrm{exp}(\frac{\lambda^2 \noisevar^2}{2})$ \cite{beta_chowdhury17a}. 
\vspace{-0.35em}
\section{Problem Statement}
\looseness -1 We consider a continuous-time non-linear system
\begin{align}
    \dot{\state}(t) = \dyn(\state(t), \coninput(t)),\;\; \state(0) = \state_s,~ t \geq 0, \label{eq: dyn}
\end{align}
\looseness -1 with state \mbox{$\state(t) \in \R^\sdim$}, control input \mbox{$\coninput(t) \in \A \subseteq \R^m$} and initial state \mbox{$\state_s \in \X$}. The dynamic model $f$ is known, and additionally, we assume $f$ Lipschitz continuous, and the input signal 
$\coninput(\cdot)$ is piecewise continuous, denoted by $\coninput(\cdot) \in \pwcinput$.
This implies the existence of a unique solution $\state(\cdot)$ to the system \eqref{eq: dyn} \cite[Theorem 3.2]{khalil2015nonlinear}.
The system is subject to known state and input constraints \mbox{$\state(t) \in \X, \coninput(t) \in \inputSpace, \forall t \geq 0$}.
Additionally, the system should satisfy the safety constraints,
\begin{align}
\state(t) \in \constSet[]{} \coloneqq \{\state \in \X|~\constrain(\state)\geq 0\}, \forall t\geq 0. \label{eq:def_constraint}
\end{align}
\looseness -1 Here, the function $\constrain: \R^\sdim \to \R$ characterizing the safety constraint is \emph{a-priori unknown} and needs to be learnt. 
The system observes $y$ which is a noisy evaluation of the constraint at state $\state$, given by $y = \constrain(\state) + \eta$, where $\eta$ is zero-mean conditionally $\noisevar$-sub Gaussian noise. 
This includes the special case of zero-mean noise bounded in $[-\noisevar,\noisevar]$ and i.i.d. Gaussian noise with variance bounded by $\noisevar^2$.
For example, the constraint function $\constrain(\state)$ can be the distance to the nearest obstacle, which is not known, but for which noisy observations of $\constrain(\state)$ can be obtained using LIDAR or ultrasound sensors.
In order to ensure the safety of a dynamical system, we shall not only consider safe reachability to the measurement location but also ensure safe returnability to a safe set.


\begin{figure}\vspace{-0.3em}
\hspace{-1.50mm}
    \begin{subfigure}[t]{0.5\columnwidth}
  	\centering
  	\scalebox{0.55}{\tikzset{every picture/.style={line width=0.75pt}} 

\begin{tikzpicture}[x=0.75pt,y=0.75pt,yscale=-1,xscale=1]

\draw  [fill={rgb, 255:red, 177; green, 177; blue, 255 }  ,fill opacity=1 ] (23,813.2) .. controls (23,792.1) and (40.1,775) .. (61.2,775) -- (248.3,775) .. controls (269.4,775) and (286.5,792.1) .. (286.5,813.2) -- (286.5,927.8) .. controls (286.5,948.9) and (269.4,966) .. (248.3,966) -- (61.2,966) .. controls (40.1,966) and (23,948.9) .. (23,927.8) -- cycle ;
\draw  [fill={rgb, 255:red, 108; green, 215; blue, 108 }  ,fill opacity=0.6 ] (115.71,777) .. controls (137.71,777) and (159.54,777.8) .. (194.86,778.43) .. controls (230.17,779.05) and (264.25,779) .. (272,786.5) .. controls (279.75,794) and (285.5,819) .. (283,850.5) .. controls (280.5,882) and (289.5,901.5) .. (282.5,908) .. controls (275.5,914.5) and (265.5,917) .. (251.5,931) .. controls (237.5,945) and (240.5,950.5) .. (241.5,959.5) .. controls (242.5,968.5) and (187.5,964) .. (160.5,963.5) .. controls (133.5,963) and (64,968.5) .. (39.5,955) .. controls (15,941.5) and (29.66,867.31) .. (25.5,827.5) .. controls (21.34,787.69) and (44.47,803.2) .. (45.36,787.43) .. controls (46.24,771.66) and (93.71,777) .. (115.71,777) -- cycle ;
\draw  [fill={rgb, 255:red, 177; green, 177; blue, 255 }  ,fill opacity=1 ] (110.13,921.32) .. controls (95.04,929.89) and (81.34,917.3) .. (59.35,910.11) .. controls (37.36,902.93) and (110.88,905.63) .. (70.35,879.73) .. controls (29.82,853.82) and (33.5,854.5) .. (62.5,829) .. controls (91.5,803.5) and (75.5,789) .. (72.5,784.5) .. controls (69.5,780) and (67.29,777.14) .. (98.86,779) .. controls (130.43,780.86) and (80.14,833.25) .. (97,852.64) .. controls (113.87,872.03) and (143.5,790.5) .. (145,820.5) .. controls (146.5,850.5) and (125.22,912.74) .. (110.13,921.32) -- cycle ;
\draw  [fill={rgb, 255:red, 177; green, 177; blue, 255 }  ,fill opacity=1 ] (233,786.5) .. controls (255.02,786.92) and (237.95,825.73) .. (251.65,837.9) .. controls (254.12,840.09) and (257.6,841.42) .. (262.5,841.5) .. controls (294.5,842) and (228.5,884.5) .. (225,859.5) .. controls (221.5,834.5) and (207,786) .. (233,786.5) -- cycle ;
\draw [line width=1.5]    (37.11,788.61) -- (38.82,790.06) -- (40.23,791.25) ;
\draw [shift={(43.29,793.82)}, rotate = 220.1] [fill={rgb, 255:red, 0; green, 0; blue, 0 }  ][line width=0.08]  [draw opacity=0] (4.64,-2.23) -- (0,0) -- (4.64,2.23) -- cycle    ;
\draw [shift={(34.05,786.03)}, rotate = 40.2] [fill={rgb, 255:red, 0; green, 0; blue, 0 }  ][line width=0.08]  [draw opacity=0] (4.64,-2.23) -- (0,0) -- (4.64,2.23) -- cycle    ;
\draw  [fill={rgb, 255:red, 108; green, 215; blue, 108 }  ,fill opacity=0.4 ] (67.5,935.25) .. controls (79.5,932.25) and (111.25,932.25) .. (113,939.25) .. controls (114.75,946.25) and (88.75,962) .. (73.5,959.75) .. controls (58.25,957.5) and (55.5,938.25) .. (67.5,935.25) -- cycle ;
\draw  [fill={rgb, 255:red, 0; green, 0; blue, 0 }  ,fill opacity=1 ] (112.56,912.87) .. controls (94.19,928.44) and (107.04,912.59) .. (74.94,911.05) .. controls (42.84,909.52) and (116.65,901.83) .. (74.49,878.67) .. controls (32.33,855.5) and (36.72,856.63) .. (70.71,823.59) .. controls (104.7,790.54) and (53.14,776.65) .. (92.17,779.17) .. controls (131.21,781.68) and (69.37,829.46) .. (93.99,855.18) .. controls (118.62,880.9) and (142.55,798.26) .. (142.5,823.5) .. controls (142.45,848.74) and (130.93,897.3) .. (112.56,912.87) -- cycle ;
\draw  [color={rgb, 255:red, 0; green, 0; blue, 0 }  ,draw opacity=1 ][fill={rgb, 255:red, 0; green, 0; blue, 0 }  ,fill opacity=1 ] (286,910.5) .. controls (285.5,923.5) and (288.67,930.99) .. (282.1,944.7) .. controls (271.24,965.13) and (254.64,965.11) .. (244.7,965.9) .. controls (240.95,954.9) and (241,942.88) .. (252.5,934) .. controls (264,925.13) and (274.5,913.5) .. (286,910.5) -- cycle ;
\draw  [color={rgb, 255:red, 214; green, 39; blue, 40 }  ,draw opacity=1 ][fill={rgb, 255:red, 214; green, 39; blue, 40 }  ,fill opacity=1 ] (70.29,942) -- (71.66,945.27) -- (75.79,945.36) -- (72.5,947.46) -- (73.69,950.79) -- (70.29,948.81) -- (66.88,950.79) -- (68.07,947.46) -- (64.78,945.36) -- (68.92,945.27) -- cycle ;
\draw  [fill={rgb, 255:red, 0; green, 0; blue, 0 }  ,fill opacity=1 ] (233.5,789) .. controls (257.5,790) and (228,846.5) .. (260,843.5) .. controls (292,840.5) and (229.5,880.5) .. (226,855.5) .. controls (222.5,830.5) and (209.5,788) .. (233.5,789) -- cycle ;

\draw (45,787) node [anchor=north west][inner sep=0.75pt]  [font=\LARGE] [align=left] {$\displaystyle \epsilon $};

\end{tikzpicture}}
    \caption{Maximum safe domain exploration}
    \label{fig:obj_max_se}
    \end{subfigure}
    \hspace{-3.00mm}
~
    \begin{subfigure}[t]{0.5\columnwidth}
  	\centering
  	\scalebox{0.55}{\tikzset{every picture/.style={line width=0.75pt}} 

\begin{tikzpicture}[x=0.75pt,y=0.75pt,yscale=-1,xscale=1]

\draw  [fill={rgb, 255:red, 177; green, 177; blue, 255 }  ,fill opacity=1 ] (310.5,681.7) .. controls (310.5,660.6) and (327.6,643.5) .. (348.7,643.5) -- (535.8,643.5) .. controls (556.9,643.5) and (574,660.6) .. (574,681.7) -- (574,796.3) .. controls (574,817.4) and (556.9,834.5) .. (535.8,834.5) -- (348.7,834.5) .. controls (327.6,834.5) and (310.5,817.4) .. (310.5,796.3) -- cycle ;
\draw  [fill={rgb, 255:red, 108; green, 215; blue, 108 }  ,fill opacity=0.6 ] (347.25,799.75) .. controls (368.5,796) and (389.5,795.5) .. (400,784.5) .. controls (410.5,773.5) and (416,769.5) .. (424,741.5) .. controls (432,713.5) and (440,688.5) .. (429,684) .. controls (418,679.5) and (408.5,720.5) .. (396,724) .. controls (383.5,727.5) and (382.2,723) .. (378.5,714) .. controls (374.8,705) and (389.6,671.4) .. (410.4,657) .. controls (431.2,642.6) and (464,654.6) .. (468.4,695.8) .. controls (472.8,737) and (447.8,798.9) .. (424.8,815.4) .. controls (401.8,831.9) and (388.8,833.4) .. (361.25,832.25) .. controls (333.7,831.1) and (326,803.5) .. (347.25,799.75) -- cycle ;
\draw  [fill={rgb, 255:red, 108; green, 215; blue, 108 }  ,fill opacity=0.4 ] (355,803.75) .. controls (367,800.75) and (398.75,800.75) .. (400.5,807.75) .. controls (402.25,814.75) and (376.25,830.5) .. (361,828.25) .. controls (345.75,826) and (343,806.75) .. (355,803.75) -- cycle ;
\draw  [fill={rgb, 255:red, 0; green, 0; blue, 0 }  ,fill opacity=1 ] (400.06,781.37) .. controls (381.69,796.94) and (394.54,781.09) .. (362.44,779.55) .. controls (330.34,778.02) and (404.15,770.33) .. (361.99,747.17) .. controls (319.83,724) and (324.22,725.13) .. (358.21,692.09) .. controls (392.2,659.04) and (340.64,645.15) .. (379.67,647.67) .. controls (418.71,650.18) and (356.87,697.96) .. (381.49,723.68) .. controls (406.12,749.4) and (430.05,666.76) .. (430,692) .. controls (429.95,717.24) and (418.43,765.8) .. (400.06,781.37) -- cycle ;
\draw  [fill={rgb, 255:red, 0; green, 0; blue, 0 }  ,fill opacity=1 ] (521,657.5) .. controls (545,658.5) and (515.5,715) .. (547.5,712) .. controls (579.5,709) and (517,749) .. (513.5,724) .. controls (510,699) and (497,656.5) .. (521,657.5) -- cycle ;
\draw  [color={rgb, 255:red, 0; green, 0; blue, 0 }  ,draw opacity=1 ][fill={rgb, 255:red, 0; green, 0; blue, 0 }  ,fill opacity=1 ] (573.5,779) .. controls (573,792) and (576.17,799.49) .. (569.6,813.2) .. controls (558.74,833.63) and (542.14,833.61) .. (532.2,834.4) .. controls (528.45,823.4) and (528.5,811.38) .. (540,802.5) .. controls (551.5,793.63) and (562,782) .. (573.5,779) -- cycle ;
\draw  [color={rgb, 255:red, 214; green, 39; blue, 40 }  ,draw opacity=1 ][fill={rgb, 255:red, 214; green, 39; blue, 40 }  ,fill opacity=1 ] (357.79,810.5) -- (359.16,813.77) -- (363.29,813.86) -- (360,815.96) -- (361.19,819.29) -- (357.79,817.31) -- (354.38,819.29) -- (355.57,815.96) -- (352.28,813.86) -- (356.42,813.77) -- cycle ;
\draw  [color={rgb, 255:red, 44; green, 160; blue, 44 }  ,draw opacity=1 ][fill={rgb, 255:red, 44; green, 160; blue, 44 }  ,fill opacity=1 ] (391.29,704.5) -- (392.66,707.77) -- (396.79,707.86) -- (393.5,709.96) -- (394.69,713.29) -- (391.29,711.31) -- (387.88,713.29) -- (389.07,709.96) -- (385.78,707.86) -- (389.92,707.77) -- cycle ;
\draw [color={rgb, 255:red, 74; green, 144; blue, 226 }  ,draw opacity=1 ]   (357.79,815.36) .. controls (406,825.4) and (436.2,773.2) .. (447.2,732.2) .. controls (458.2,691.2) and (436.8,626.2) .. (391.29,709.36) ;

\end{tikzpicture}}
    \caption{Goal-directed safe exploration}
    \label{fig:obj_goal_se}
    \end{subfigure}\vspace{-0.5em}
\caption{\looseness -1 Illustration of safe exploration for different objectives in (a) and (b). The domain (blue region) is a-priori unknown with constraints represented by the black region. The robot starts at the red star with a known invariant safe set around it (small green set). The robot in (a) needs to explore the entire domain safely up to small tolerance $\epsilon>0$ and in (b) only needs to explore the area required for guaranteeing convergence to a goal (green star).} \vspace{-1.5em}\label{fig:obj_demonstration}
\end{figure}
\looseness -1 \textit{Objective:} In this setting, our objective is to guarantee exploration in finite time while always ensuring safety \eqref{eq:def_constraint}. 
We consider two exploration scenarios: first, maximum safe domain exploration and second, goal-directed safe exploration (see \cref{fig:obj_demonstration}). 
The task of maximum exploration manifests itself as learning about the a-priori unknown constraint $\constrain$ up to the limit of what can be learned safely, see also \cref{eqn:objective} for a formal definition. 
For a goal-directed safe exploration, we aim to minimise a known loss function $\utility:\R^\sdim \to \R$ while subject to safety constraint \eqref{eq:def_constraint}, see also \cref{eqn:goal_objective} for a formal definition.
Here, we only need to explore the region required for guaranteeing loss minimization instead of maximum domain exploration.
In both exploration scenarios, we want to achieve the corresponding objectives in finite time, up to a user-defined (arbitrarily small) tolerance $\epsilon > 0$ and guarantee safety with a user-defined arbitrarily high probability $1-\prob, (\prob >0)$.

\looseness -1 \textit{Solution approach:} To tackle this challenge, we introduce our core algorithm, \sempc in \cref{sec:MPC}. 
Throughout the paper, we gradually build up this algorithm as summarized in \Cref{fig:flowchart}.
We first introduce its foundation in \cref{sec:SE_theory}, the safe guaranteed exploration framework, which provides the key technique to ensure exploration in finite time. Then, we introduce two techniques to improve the efficiency: \textit{(i)} utilizing a known Lipschitz bound (\cref{sec:SE_expanders}), \textit{(ii)} considering a goal-directed setting (\cref{sec:goalSE}). Lastly, \sempc introduces re-planning in a receding horizon style and additionally incorporates the improvements of $(i)$ and $(ii)$. 
\vspace{-0.35em}
\section{Background on Gaussian process}\label{sec:backgroundGP}
\looseness -1 During safe exploration, to collect a measurement at a state $\state$ safely, we require a mechanism to certify that $\state$ is safe before visiting it. For this, we require some regularity assumptions on the constraint function $\constrain$ such that knowing about $\constrain$ at a certain point also provides information about the neighbouring region.
\begin{assumption}[Regularity]
   \looseness -1 The constraint function $\constrain$ lies in a \textit{Reproducing Kernel Hilbert Space} (RKHS) \cite{kernel-Schlkopf} associated with a positive definite kernel \mbox{$\kernelfunc(\cdot,\cdot)$} such that \mbox{$\kernelfunc(\cdot,\cdot)$} is continuous and $\constrain$ has bounded RKHS norm, \mbox{$\|\constrain\|_{\kernelfunc} \leq \Bq< \infty$}. \label{assump:q_RKHS}
\end{assumption}
\looseness -1 
\cref{assump:q_RKHS} is common \cite{safe-bo-sui15,sui2018stagewise,berkenkamp2023bayesian,turchetta2016safemdp,turchetta2019safe,prajapat2022near} to prescribe regularity conditions on the unknown function, i.e., the kernel captures how close the function values are if two locations are close. \cref{assump:q_RKHS} enables us to use Gaussian Processes (GPs) \cite{gp-Rasmussen} to construct high-probability confidence intervals for $\constrain$. We specify the GP prior over $\constrain$ through a mean function $\mu: \R^\sdim \to \R$
and a kernel function $\kernelfunc: \R^\sdim \times \R^\sdim \to \R$ that captures the covariance between different locations. 
If we have access to $\n$ measurements, $Y_{X_\n} = \{ q(\PtInDomain_i)+\noise_i \}_{i=1}^\n$ measured at $X_\n = \{ \PtInDomain_i \}^{\n}_{i=1}$ perturbed by zero-mean conditionally $\noisevar$-sub-Gaussian noise $\noise_i$, 
we can compute the posterior mean and covariance over the constraint at unseen locations $\PtInDomain$, $\PtInDomain^{\prime}$ as:
\begin{align*}
    \muconst[\n](\PtInDomain) &= k^{\top}_\n(\PtInDomain) (K_\n + \noisevar^2 \identity)^{-1}Y_{X_\n}, \\
    k_\n(\PtInDomain,\PtInDomain^{\prime}) &= k(\PtInDomain,\PtInDomain^{\prime}) - k^{\top}_\n(\PtInDomain)(K_\n + \noisevar^2 \identity)^{-1}k_\n (\PtInDomain^{\prime}), \numberthis \label{eqn:posterior_update}\\
    \sigconst[\n](\PtInDomain) &= \sqrt{k_\n(\PtInDomain,\PtInDomain)},
\end{align*}
\looseness -1 where $k_\n(\PtInDomain) = [k(\PtInDomain_1,\PtInDomain), . . . , k(\PtInDomain_{\n}, \PtInDomain)]^{\top}, K_\n$ is the positive definite kernel matrix $[k(\PtInDomain,\PtInDomain^{\prime})]_{\PtInDomain,\PtInDomain^{\prime} \in X_\n}$ and $\identity \in \R^{\n \times \n}$ denotes the identity matrix.
Note that system \eqref{eq: dyn} operates in continuous time, while the \GP only uses measurements indexed by $\n \in \mathbb{N}$ at discrete-time $t_\n \in \R$.

\looseness -1 Let $H(Y_{X_\n})$ be the Shannon entropy for noisy samples $Y_{X_\n}$ collected at the set $X_\n$ \cite{beta-srinivas}. The mutual information between $\constrain_{X_\n}$ evaluated at $X_\n$ and the samples $Y_{X_\n}$ is given by, 
\begin{align}
    I(Y_{X_\n};\constrain_{X_\n}) = H(Y_{X_\n}) - H(Y_{X_\n}| \constrain_{X_\n}) \label{eqn: mutual-info-definition}
\end{align}
which quantifies the reduction in uncertainty about $\constrain$ after revealing $Y_{X_\n}$ (see \cref{apx:mutual_info}). Using this, we define the maximum information capacity $\gammaconst{n} \coloneqq \sup_{X \subseteq \X: |X| \leq n} I(Y_X;q_X)$, which upper bounds the information that can be obtained with a finite number of measurements $\n$ collected at any set $X$.  
The following lemma establishes high probability confidence intervals containing the \emph{a-priori} unknown $\constrain$ function using the posterior mean $\muconst[\n]$ and variance $\sigconst[\n]$ derived in \cref{eqn:posterior_update}.
\begin{lemma}
{\cite[Theorem 2]{beta_chowdhury17a}}
\label{thm: beta} 
\looseness -1 Let \cref{assump:q_RKHS} hold. If $\sqrt{\betaconst[\n]} = \Bq + 4\noisevar\sqrt{\gamma_{\n} + 1+ \log(1/\prob)}$, then $\forall \state \in \Domain$ and $\n \geq 1, |\constrain(\state) - \muconst[\n-1](\state)| \leq \sqrt{\betaconst[\n]} \sigconst[\n-1](\state)$ holds with probability at least $1-\prob$.
\end{lemma}

\looseness -1 Given $\sqrt{\betaconst[\n]}$ as in \cref{thm: beta}, we recursively build monotonic confidence intervals for the constraint after $\n \geq 1$ samples,
\begin{align*}
    \lbconst[\n](\state) &\coloneqq \max \{ \lbconst[\n-1](\state), \muconst[\n-1](\state) - \sqrt{\betaconst[\n]} \sigconst[\n-1](\state) \},\\
    \ubconst[\n](\state) &\coloneqq \min \{ \ubconst[\n-1](\state), \muconst[\n-1](\state) + \sqrt{\betaconst[\n]} \sigconst[\n-1](\state) \},
\end{align*}
with \mbox{$\lbconst[0](\state)\!\coloneqq\! \muconst[0](\state)\! -\! \sqrt{\betaconst[1]} \sigconst[0](\state)$}, \mbox{$\ubconst[0](\state)\!\coloneqq\! \muconst[0](\state)\! +\! \sqrt{\betaconst[1]} \sigconst[0](\state)$}. The functions $\lbconst[\n](\cdot)$ and $\ubconst[\n](\cdot)$ are continuous due to continuity of the kernel, see \cref{assump:q_RKHS} \cite{korezlioglu1968reproducing}.
We define the confidence width $\sumMaxwidth[]{\n}(\state)\coloneqq \ubconst[\n](\state) - \lbconst[\n](\state)$. With the definition of lower and upper bounds, it directly follows that $\forall \n \in \Nat, \forall x \in \X,$ 
\begin{align}
    \lbconst[\n](x) \leq \lbconst[\n+1](x), \ubconst[\n](x) \geq \ubconst[\n+1](x) ,  \sumMaxwidth[]{\n}(x) \geq \sumMaxwidth[]{\n+1}(x).
\end{align}

\looseness -1 
Throughout the paper, we assume  $\sqrt{\betaconst[\n]}$ is defined as in \cref{thm: beta}. 
Due to the construction of our confidence sets, an immediate corollary of \cref{thm: beta} results as follows.
\begin{corollary}\label{coro:hp_bounds}
   \looseness -1 Let \cref{assump:q_RKHS} hold. If $\sqrt{\betaconst[\n]} = \Bq + 4\noisevar\sqrt{\gamma_{\n} + 1+ \log(1/\prob)}$ then the following holds:
   \begin{align*}
       \mathrm{Pr}\Big\{ \constrain(\state) \in \left[\lbconst[\n](\state),\ubconst[\n](\state)\right]~~~\forall \state \in \Domain, \n \geq 1 \Big\} \geq 1 - \prob.
   \end{align*}\vspace{-1.5em}
\end{corollary}
Hence, if we ensure $\lbconst[\n](\state(t))\geq 0, \forall t \geq 0$, then \cref{coro:hp_bounds} guarantees satisfaction of  constraints~\eqref{eq:def_constraint} jointly for all $t \geq0, \n\geq 1$ with probability at least $1-p$. Note that since we consider a noise distribution with unbounded tails, we must exclude rare extreme events in probability, thus our results provide high probability safety guarantees.
\section{Safe exploration - theoretical analysis}
\label{sec:SE_theory}
\begin{figure}
\setlength{\abovecaptionskip}{5pt}
    \centering
    \includegraphics[width=0.9\columnwidth]{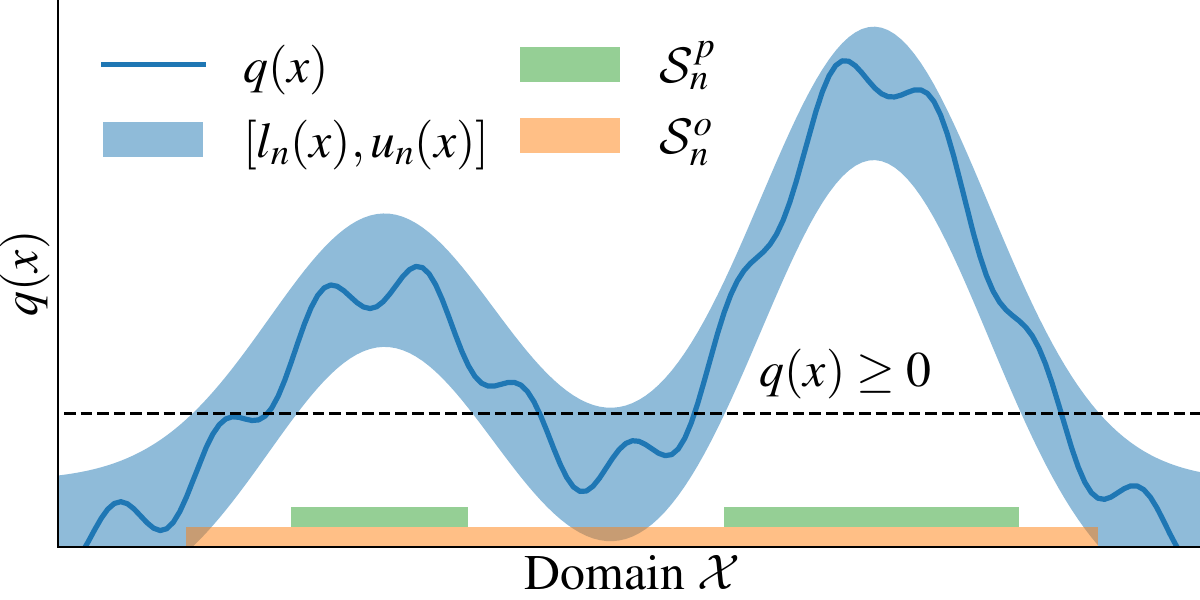}
    \caption{\looseness -1 With regularity assumptions on $\constrain$, we define high probability upper and lower confidence bounds on the a-priori unknown function $\constrain$. These are utilized to obtain pessimistic $\pessiSet[]{\n}$ and optimistic $\mathcal{S}^o_\n$ estimates of constrained sets.}
    \label{fig:1D-safe-set}\vspace{-1.5em}
\end{figure}
\looseness -1 This section gradually introduces our framework, safe guaranteed exploration using optimal control. 
We first present the sample complexity result for the continuous state-space $\X \subseteq \R^\sdim$ without considering dynamics, i.e., how many samples are required to accurately learn the constraint function safely. 
Then, we extend this result to non-linear dynamical systems in \cref{sec:sageoc} and show guaranteed maximum domain exploration while being safe for all times in \cref{sec:guaranteeingSE}.
\vspace{-0.6em}
\subsection{Sample complexity for continuous domain} \label{sec:see:sampling}
\looseness -1 Using the confidence bounds of the \GP, we define the constraint sets which will be crucial for safe sampling.

\looseness -1 \emph{Constraint sets}:
The pessimistic and optimistic estimate of the true constraint set are defined analogously to \eqref{eq:def_constraint} as
\begin{equation}
\begin{aligned}
\pessiSet[]{\n} \coloneqq& \{ \state \in \Domain| \lbconst[\n](\state) \geq 0 \}, \\
  \mathcal{S}^{o}_{\n} \coloneqq& \{\state \in \Domain| \ubconst[\n](\state) \geq 0 \},  \label{eqn:constrained_set}
\end{aligned}
\end{equation}
\looseness -1 utilizing the lower bound and the upper bound of the true function, respectively. \cref{fig:1D-safe-set} graphically illustrates the pessimistic and optimistic approximation of the true constraint set. 
With high probability, the pessimistic set is an inner approximation of the true constraint set, and hence any point in it is safe with high probability.
Similarly, the optimistic set represents a high probability outer approximation of the true constraint set, and any point outside of it is unsafe with high probability.



\looseness -1 \emph{Exploration:} We collect a measurement at a certain location, which reduces the uncertainty of the \GP in nearby regions due to the regularity of the \GP. This increases the lower confidence bound which may enlarge the pessimistic set, i.e., with every constraint sample we infer the safety of nearby regions.

\looseness -1 Thus, the agent iteratively learns about the constraint function $\constrain$ by sampling; however with noisy samples. In general, we cannot learn the continuous function $\constrain: \R^\sdim \to \R$ exactly with finite samples using only the posed regularity condition (\cref{assump:q_RKHS}).
Hence we focus on learning $\constrain$ up to a user-defined tolerance $\epsilon > 0$.
We represent an $\epsconst\operatorname{-}$approximation of the optimistic constraint set as $\optiSet[]{\n} \coloneqq \{\state \in \Domain| \ubconst[\n](\state) -\epsconst \geq 0 \}$.

\looseness -1 \emph{Sampling strategy:} In order to learn about the constraint function up to $\epsconst$ precision, we devise the sampling strategy to sample at $\state_{n+1}$ after the collection of $\n$ samples as
\begin{align}
\mathrm{Find}~\state \in \pessiSet[]{\n}: \ubconst[\n](\state) - \lbconst[\n](\state) \geq \epsconst. \label{eqn:BOsampling_strategy}
\end{align}

\looseness -1 Thus, we recursively sample in the pessimistic constraint set, which ensures safety and where the uncertainty is at least $\epsilon$, which ensures sufficient information gain for exploration. A natural implementation is to sample states with the largest confidence width, i.e., at $\argmax_{\state \in \pessiSet[]{\n}} \sumMaxwidth[]{\n}(x)$.
After each measurement, we perform a posterior update \eqref{eqn:posterior_update} and again identify the next sampling location as per \cref{eqn:BOsampling_strategy}. 
In order to show finite time termination of this simple sampling strategy, we consider the following regularity assumption.
\begin{assumption} \label{assump:sublinear}
    $\!\betaconst[\n]\!\gamma_{\n}$ grows sublinearly in $\n$, i.e., $\betaconst[\n]\!\gamma_{\n} \!< \!\mathcal{O}(\n)$. 
\end{assumption}
\looseness -1 Such an assumption is common in most prior works \cite{beta-srinivas,beta_chowdhury17a,safe-bo-sui15,sui2018stagewise,berkenkamp2023bayesian,prajapat2022near,turchetta2019safe,turchetta2016safemdp} to establish sample complexity or sublinear regret results and is not restrictive. 
Indeed, due to bounded $\Bq$ (\cref{assump:q_RKHS}), $\betaconst[\n]\gamma_{\n}$ grows sublinear in ${\n}$ for compact domains $\X$ and commonly used kernels, e.g., linear, squared exponential, Mat\'ern, etc., with sufficient eigen decay \!\cite{gammaT-vakili21a,beta-srinivas}. \cref{assump:sublinear} imposes regularity on the kernel function and restricts the class of kernels such that we can achieve finite-time convergence.
\begin{theorem} [Sample complexity] \looseness -1 Let \cref{assump:q_RKHS,assump:sublinear} hold. Let $\n^{\star}$ be the largest integer satisfying $\frac{\n^{\star}}{\beta_{\n^{\star}} \gamma_{\n^{\star}}} \leq \frac{C}{\epsconst^2}$, with $C = 8/ \log (1 + \noiseconst)$. The \GP update \eqref{eqn:posterior_update} with sampling scheme \eqref{eqn:BOsampling_strategy} yields $\state_\n$ 
satisfying $\constrain(\state_\n)\geq 0, \forall \n \geq 1$ with probability at least $1-\prob$ and $\exists \nfin \leq n^\star : \forall x \in \pessiSet[]{\nfin},~\sumMaxwidth[]{\nfin}(x)<\epsconst$.
\label{thm:sample_complexity}
\end{theorem}\vspace{-0.9em}
\begin{proof} 
\looseness -1 The confidence width $\sumMaxwidth[]{\n-1}(\state)= \ubconst[\n-1](\state) - \lbconst[\n-1](\state) \leq 2\sqrt{\betaconst[\n]} \sigconst[\n-1](\state)$ can be bounded as:
\vspace{-0.4em}
\begin{align*} 
\!\!    \sumMaxwidth[2]{\n-1}(\state) \!&\leq \!4\betaconst[\n] \sigconst[\n-1]^2(\state) 
   \! \leq \! 4\betaconst[\n] \noisevar^2 C_2 \! \log(1\!+\!\noiseconst \sigconst[\n-1]^2(\state))   \\
    &= 1/2 C_1 \betaconst[\n] \log(1+\noiseconst \sigconst[\n-1]^2(\state)). 
    \vspace{-0.4em}
\end{align*} \vspace{-0.15em}
\looseness -1 Following the analysis from \cite{beta-srinivas}, let $s \coloneqq \noiseconst \sigconst[\n-1]^2(\state) \leq \noiseconst\kernelfunc(\state,\state) \leq \noiseconst$. The second inequality uses the fact that $ \forall s \in [0, \noiseconst] , s \leq C_2 \log(1+s)$ where $C_2 = \noiseconst/\log(1+\noiseconst) \geq 1$. The last equation holds with $C_1 = 8/\log(1+\noiseconst)$.

We next bound it with the mutual information $I(Y_{X_\nfin};\constrain_{X_\nfin})$ acquired by $\nfin$ samples: \vspace{-0.4em}
\begin{align*}
    \sum_{\n=1}^{\nfin} \! \sumMaxwidth[2]{\n-1}(\state_\n)  \! &\leq  \sum_{\n=1}^{\nfin} C_1 \betaconst[\n] \frac{1}{2} \log(1+\noiseconst \sigconst[\n-1]^2(\state_\n)) \\[-0.5em] 
    &\leq C_1 \betaconst[\nfin] \frac{1}{2} \sum_{\n=1}^{\nfin}  \log(1+\noiseconst \sigconst[\n-1]^2(\state_\n))  \numberthis \label{eqn:non-dec-beta}\\[-0.1em]
    &= C_1 \betaconst[\nfin] I(Y_{X_\nfin};\constrain_{X_\nfin}). \label{eqn:mutual_info_bound}\numberthis 
    \vspace{-0.4em}
\end{align*}\vspace{-0.2em}
\looseness -1 \cref{eqn:non-dec-beta} follows since $\betaconst[\n]$ is non-decreasing (\!\!\cite{beta-srinivas}), and the last equality follows from the definition of mutual information $I(Y_{X_\nfin};\constrain_{X_\nfin})$, see \cref{apx:mutual_info}.
Suppose we sample as per criteria \eqref{eqn:BOsampling_strategy} for $\nfin$ iterations before terminating, then it holds that \vspace{-1em}
\begin{align}
    \!\!\nfin \epsconst^2 \!\leq \!\sum_{\n=1}^{\nfin} \! \sumMaxwidth[2]{\n-1}(\state_\n) \!\leq \! C_1 \betaconst[\nfin] &I(Y_{X_\nfin};\constrain_{X_\nfin}) \!\leq \!C_1 \betaconst[\nfin] \gamma_{\nfin} \label{eqn:mono-gamma}\\[-0.9em]
  \implies  \frac{\nfin}{\betaconst[\nfin] \gamma_{\nfin}} &\leq \frac{C_1 }{\epsconst^2}. \label{eqn:sample_complexity}
\end{align}\vspace{-0.2em}
\looseness -1 \cref{eqn:mono-gamma} follows from \cref{eqn:mutual_info_bound} and the definition of $\gamma_{\nfin}$ \cite{beta-srinivas}. 
Under \cref{assump:sublinear}, since $\betaconst[\nfin] \gamma_{\nfin}$ grows sublinear, there exists a largest integer $\n^\star$ that satisfies \eqref{eqn:sample_complexity} and thus $\nfin\leq\n^\star$.
Note that as per \eqref{eqn:BOsampling_strategy}, $\forall \n \geq 0,\state_\n \in \pessiSet[]{\n-1}$ which by definition \eqref{eqn:constrained_set} implies 
$\lbconst[\n-1](\state_\n) \geq 0$ and hence by  \cref{coro:hp_bounds} implies
$\constrain(\state_\n) \geq 0$ with probability at least $1-\prob$.
\end{proof}


\looseness -1 Thus, with $\nfin \leq \n^\star$ samples, the uncertainty will decrease below $\epsilon$, the sampling rule \eqref{eqn:BOsampling_strategy} will no longer recommend additional samples, and we are guaranteed to terminate with $\nfin$ finite samples.
Moreover, this readily implies that we know about the true constraint function up to $\epsconst$ precision within the pessimistic set i.e., $\forall x\in \pessiSet[]{\nfin}$, we have, $\constrain(\state) \in [\lbconst[\nfin](\state), \lbconst[\nfin](\state)+\epsconst]$.
In contrast to earlier works \cite{safe-bo-sui15,sui2018stagewise,berkenkamp2023bayesian,turchetta2016safemdp,turchetta2019safe,prajapat2022near}, which focused on a discrete domain, our sample complexity analysis removes explicit dependence from the discretization size. 
Notably, this result holds for any sampling rule ensuring $\sumMaxwidth[]{\n}(\state)\geq\epsconst$, while existing approaches~\cite{safe-bo-sui15,sui2018stagewise,berkenkamp2023bayesian,turchetta2016safemdp,turchetta2019safe,prajapat2022near} require sampling at the maximal uncertainity location, $\max_x \sumMaxwidth[]{\n}(\state)$, which is unnecessarily restrictive for dynamical systems.
\vspace{-0.7em}
\subsection{Safe exploration for non-linear dynamics} 
\label{sec:sageoc}
\looseness -1 In this section, we extend the sampling rule \eqref{eqn:BOsampling_strategy} to consider a dynamical system. The system cannot jump to any desired location and needs to navigate with a safe path. However, without making any assumption on the initial safe set we do not even know where to start the process. For this, we make the following assumption on the known terminal set, 
\begin{assumption}[Terminal set] \label{assump: safe-init} The agent starts exploration at the initial state $\state_s \in \safeInit{0}$. The terminal sets, $\safeInit{\n}$ $\n \geq 0$ satisfy:
\begin{itemize}
\item Pessimistically safe: $\forall n \geq 0, \safeInit{\n} \subseteq\pessiSet[]{\n}$.
\item Monotonicity: The terminal sets $\safeInit{\n}$ are nested (non-decreasing) w.r.t $\n$, i.e., $\safeInit{\n} \subseteq \safeInit{\n+1}$.
\item \looseness -1 Control invariance: There exists a continuous feedback $\kappa_\n:\safeInit{\n} \to \inputSpace$ such that the set $\safeInit{\n}$ is positively invariant \cite[Definition 3.2]{blanchini1999set} for $a=\kappa_{\n}(\state)$.
\item Controllable around the terminal set: For any two points $\state_0, \state_1 \in \safeInit{\n}$, there exists an input signal $\coninput(\cdot) \in \pwcinput$ that steers $\state_0$ to $\state_1$ safely in time $t_1 \leq \timeX$, i.e., $ \exists \coninput(\cdot),  \dot{\state}(t) = \dyn(\state(t),\coninput(t)), \state(0) = \state_0, \state(t_1) = \state_1, \state(t) \in \pessiSet[]{\n}, t \in [0, t_1], t_1 \leq \timeX$.


\end{itemize}
\end{assumption}
\begin{figure}
\setlength{\abovecaptionskip}{1pt}
    \centering
    \scalebox{0.4}{\tikzset{every picture/.style={line width=0.75pt}} 

\begin{tikzpicture}[x=0.75pt,y=0.75pt,yscale=-1,xscale=1]

\draw  [fill={rgb, 255:red, 177; green, 177; blue, 255 }  ,fill opacity=1 ] (67.47,55.26) .. controls (96.37,40.53) and (387.46,24.72) .. (430.37,68.95) .. controls (473.27,113.18) and (480.22,238.16) .. (459.53,261.56) .. controls (438.83,284.95) and (211.76,248.67) .. (211.72,271.36) .. controls (211.68,294.05) and (65.68,273.53) .. (41.03,258.56) .. controls (16.37,243.59) and (6.61,154.59) .. (23.76,146.02) .. controls (40.91,137.45) and (20.55,129.31) .. (24.53,105.12) .. controls (28.52,80.93) and (38.56,69.99) .. (67.47,55.26) -- cycle ;
\draw  [fill={rgb, 255:red, 245; green, 166; blue, 35 }  ,fill opacity=0.72 ] (399.24,208.72) .. controls (410.79,197.72) and (432,202.33) .. (444.08,203.82) .. controls (456.15,205.32) and (454,225.42) .. (451.46,236.54) .. controls (448.93,247.66) and (427.33,257.58) .. (416.34,258.04) .. controls (405.35,258.51) and (387.82,246.24) .. (386.48,238.29) .. controls (385.14,230.34) and (387.7,219.72) .. (399.24,208.72) -- cycle ;
\draw  [fill={rgb, 255:red, 245; green, 166; blue, 35 }  ,fill opacity=0.72 ] (81.55,80.69) .. controls (122.95,41.91) and (240.14,43.96) .. (283.4,49.14) .. controls (326.66,54.32) and (393.1,51.59) .. (372.02,79.63) .. controls (350.93,107.67) and (425.69,124.03) .. (416.73,148.93) .. controls (407.78,173.83) and (301.22,248.1) .. (285.04,253.58) .. controls (268.85,259.07) and (104.03,257.23) .. (78.53,222.28) .. controls (53.03,187.32) and (40.15,119.47) .. (81.55,80.69) -- cycle ;
\draw  [fill={rgb, 255:red, 108; green, 215; blue, 108 }  ,fill opacity=0.6 ] (81.55,80.69) .. controls (122.95,41.91) and (233.15,44.25) .. (283.4,49.14) .. controls (333.65,54.02) and (393.1,51.59) .. (372.02,79.63) .. controls (350.93,107.67) and (425.69,124.03) .. (416.73,148.93) .. controls (407.78,173.83) and (301.22,248.1) .. (285.04,253.58) .. controls (268.85,259.07) and (104.03,257.23) .. (78.53,222.28) .. controls (53.03,187.32) and (40.15,119.47) .. (81.55,80.69) -- cycle ;
\draw  [line width=3]  (293.9,141.14) .. controls (293.87,140.64) and (294.26,140.22) .. (294.76,140.2) .. controls (295.25,140.18) and (295.67,140.57) .. (295.69,141.06) .. controls (295.71,141.56) and (295.33,141.98) .. (294.83,142) .. controls (294.34,142.02) and (293.92,141.63) .. (293.9,141.14) -- cycle ;
\draw  [line width=2.25]  (168.79,160.21) .. controls (168.75,159.38) and (168.24,158.73) .. (167.63,158.76) .. controls (167.03,158.78) and (166.56,159.47) .. (166.6,160.3) .. controls (166.63,161.13) and (167.15,161.78) .. (167.76,161.75) .. controls (168.36,161.73) and (168.82,161.04) .. (168.79,160.21) -- cycle ;
\draw  [fill={rgb, 255:red, 10; green, 146; blue, 23 }  ,fill opacity=0.4 ] (267.94,204.26) .. controls (302.12,207.82) and (348.88,153.8) .. (331.96,132.49) .. controls (315.05,111.19) and (393.14,52.59) .. (372.02,79.63) .. controls (350.89,106.67) and (424.67,123.57) .. (416.73,148.93) .. controls (408.8,174.29) and (295.98,252.12) .. (285.04,253.58) .. controls (274.09,255.05) and (235.32,237.67) .. (233.98,229.72) .. controls (232.65,221.77) and (233.76,200.7) .. (267.94,204.26) -- cycle ;
\draw  [fill={rgb, 255:red, 80; green, 227; blue, 194 }  ,fill opacity=0.4 ] (93.02,180.94) .. controls (88.93,159.72) and (115.49,136.77) .. (152.33,129.67) .. controls (189.17,122.56) and (222.36,134) .. (226.45,155.22) .. controls (230.54,176.43) and (203.98,199.39) .. (167.14,206.49) .. controls (130.29,213.59) and (97.11,202.15) .. (93.02,180.94) -- cycle ;
\draw  [fill={rgb, 255:red, 80; green, 227; blue, 194 }  ,fill opacity=0.67 ] (152.23,188.25) .. controls (142.73,182.34) and (142.44,165.66) .. (151.57,151.01) .. controls (160.7,136.35) and (175.8,129.26) .. (185.29,135.17) .. controls (194.78,141.09) and (195.07,157.76) .. (185.94,172.42) .. controls (176.81,187.07) and (161.72,194.16) .. (152.23,188.25) -- cycle ;
\draw [line width=1.5]  [dash pattern={on 5.63pt off 4.5pt}]  (168.15,159.68) .. controls (189.18,174.53) and (303.21,171.69) .. (294.84,139.91) .. controls (286.68,108.93) and (252.27,122.22) .. (210.75,146.85) ;
\draw [shift={(207.54,148.76)}, rotate = 328.97] [fill={rgb, 255:red, 0; green, 0; blue, 0 }  ][line width=0.08]  [draw opacity=0] (13.4,-6.43) -- (0,0) -- (13.4,6.44) -- (8.9,0) -- cycle    ;
\draw [shift={(244.32,167.11)}, rotate = 174.86] [fill={rgb, 255:red, 0; green, 0; blue, 0 }  ][line width=0.08]  [draw opacity=0] (13.4,-6.43) -- (0,0) -- (13.4,6.44) -- (8.9,0) -- cycle    ;
\draw [shift={(248.32,127.56)}, rotate = 339.84] [fill={rgb, 255:red, 0; green, 0; blue, 0 }  ][line width=0.08]  [draw opacity=0] (13.4,-6.43) -- (0,0) -- (13.4,6.44) -- (8.9,0) -- cycle    ;
\draw  [fill={rgb, 255:red, 108; green, 215; blue, 108 }  ,fill opacity=0.6 ] (399.24,208.72) .. controls (410.79,197.72) and (432,202.33) .. (444.08,203.82) .. controls (456.15,205.32) and (454,225.42) .. (451.46,236.54) .. controls (448.93,247.66) and (427.33,257.58) .. (416.34,258.04) .. controls (405.35,258.51) and (387.82,246.24) .. (386.48,238.29) .. controls (385.14,230.34) and (387.7,219.72) .. (399.24,208.72) -- cycle ;
\draw  [fill={rgb, 255:red, 10; green, 146; blue, 23 }  ,fill opacity=0.4 ] (399.24,208.72) .. controls (410.79,197.72) and (432,202.33) .. (444.08,203.82) .. controls (456.15,205.32) and (454,225.42) .. (451.46,236.54) .. controls (448.93,247.66) and (427.33,257.58) .. (416.34,258.04) .. controls (405.35,258.51) and (387.82,246.24) .. (386.48,238.29) .. controls (385.14,230.34) and (387.7,219.72) .. (399.24,208.72) -- cycle ;

\draw (393.94,65.88) node [anchor=north west][inner sep=0.75pt]  [font=\Huge,rotate=-357.58]  {$\mathcal{X}$};
\draw (93.47,78.52) node [anchor=north west][inner sep=0.75pt]  [font=\LARGE,rotate=-357.58]  {$\Rcontoper[\safeInit{s},\safeInit{e}]{\S}$};
\draw (298.09,123.65) node [anchor=north west][inner sep=0.75pt]  [font=\LARGE,rotate=-357.58]  {$z$};
\draw (104.45,176.38) node [anchor=north west][inner sep=0.75pt]  [font=\LARGE,rotate=-357.58]  {$\safeInit{e}$};
\draw (319.19,185.41) node [anchor=north west][inner sep=0.75pt]  [font=\LARGE,rotate=-357.58]  {$\S$};
\draw (146.94,163.57) node [anchor=north west][inner sep=0.75pt]  [font=\LARGE,rotate=-357.58]  {$\safeInit{s}$};
\draw (410.67,218.58) node [anchor=north west][inner sep=0.75pt]  [font=\LARGE,rotate=-357.58]  {$\S$};

\end{tikzpicture}}
    \caption{\looseness -1 Illustration of a reachable returnable safe set, $\Rcontoper[\safeInit{s},\safeInit{e}]{\S}$: 
    The locations within $\S$ that can be reached from $\safeInit{s}$ and returned to $\safeInit{e}$ while following the dynamics and satisfying the safety constraints $\S$.} 
\label{fig:constrain_func}\vspace{-1.4em}
\end{figure}
\looseness -1 The assumption is comparable to the safe initial seed assumption in the discrete domain \cite{turchetta2019safe, turchetta2016safemdp, prajapat2022near} or safe invariant terminal sets assumed in the \mpc literature \cite{wabersich2021predictive,wabersich2023data,saccani2022multitrajectory,soloperto2023safe,kohler2023analysis}. \cref{assump: safe-init} ensures a safe initialization and is satisfied, e.g., by \textit{i)} a fixed safe initial steady state $\safeInit{\n} = \{x_s\}$,  \textit{ii)} a set of steady states in the pessimistic set. 
 Please see \cref{lem:pessi_increase_Xn} in \cref{apx:assum_local_control} for detailed proof that \textit{(ii)} satisfies \cref{assump: safe-init}. 
Note that \cref{assump: safe-init}
allows the safe sets to grow with the pessimistic set, which simplifies the computation of a safe returnable path. As we will see later in \cref{sec:guaranteeingSE}, the \textit{controllable around the terminal set} property of \cref{assump: safe-init} is only required for guaranteeing exploration whereas safety can be ensured with the remaining properties.

\looseness -1 \emph{Reachable returnable safe set:} For safety \eqref{eq:def_constraint}, it does not suffice that we steer the system to a safe measurement location. 
Instead, the dynamics of the system need to be considered to ensure safety during transient operation.
For this, we define the reachable returnable safe set for a  fixed horizon $T>0$ as follows:
\begin{multline}
\label{eqn: pessi-reach-ret-cont}
   \!\!\!\!\!\!\! \Rcontoper[\safeInit{\mathrm{s}}, \safeInit{\mathrm{e}}]{\S} \! \coloneqq \! \{ z \in \R^\sdim~|~ \exists \coninput(\cdot) \! \in \pwcinput, \tsample \!\! \in[0,T], \state(0) \in \safeInit{\mathrm{s}} \\
  \qquad \qquad \qquad \quad \ \dot{\state}(t) =  \dyn(\state(t), \coninput(t)), \state(t) \in \S, \forall t\in[0,T]\\ z=\state(\tsample),\state(T)\in \safeInit{e}~\}.
\end{multline} 
Here, $\Rcontoper[\safeInit{\mathrm{s}}, \safeInit{\mathrm{e}}]{\S}$ denotes the set of all the locations in set $\S$, which are safely reached from $\safeInit{\mathrm{s}}$ following the dynamics and can be safely returned to the set $\safeInit{\mathrm{e}}$ as shown in \cref{fig:constrain_func}. By using $\S \in \{  \optiSet[]{\n}, \pessiSet[]{\n}$, $\constSet[]{}$\}, we obtain $\epsconst-$optimistic, pessimistic and true reachable returnable safe sets, respectively. When $\safeInit{\mathrm{e}}=\safeInit{\mathrm{s}}$, we slightly abuse the notation by defining $\Rcontoper[\safeInit{\mathrm{s}}]{\S} \coloneqq \Rcontoper[\safeInit{\mathrm{s}}, \safeInit{\mathrm{s}}]{\S}$.

\emph{Sampling strategy:} We now define the sampling strategy for a dynamical system utilizing $\Rcontoper[\safeInit{\n}]{\pessiSet[]{\n}}$ which allows us to sample while ensuring a safely reachable path, as well as safely returnable path back to the terminal set. Additionally, to learn the constraint up to $\epsconst$ tolerance, we sample the $(\n+1)^{th}$ location analogous to the rule \eqref{eqn:BOsampling_strategy} given by,
\begin{align}
\mathrm{Find}~\state \in \Rcontoper[\safeInit{\n}] {\pessiSet[]{\n}}: \ubconst[\n](\state) - \lbconst[\n](\state) \geq \epsconst. \label{eqn:sampling_strategy}
\end{align}
Solving \eqref{eqn:sampling_strategy} yields a state and input trajectory $\tilde{\state}(t), \tilde{\coninput}(t)$, $t \in [0,T]$, initial condition $\tilde{\state}(0) \in \safeInit{\n}$, and a sampling location $\tilde{\state}(\tsample) \coloneqq \state, \tsample\in[0,T]$. Problem \eqref{eqn:sampling_strategy} can be formulated similar to a standard MPC problem and thus solved using established tools, see \cref{sec:numerical} for details on the implementation.

\begin{algorithm}[!t]
\caption{Safe guaranteed exploration}
\begin{algorithmic}[1]
\State \textbf{Initialize:} Start at $\state_s \in \safeInit{0}$, \GP, Horizon $T$, Tolerance $\epsilon$.
 \For{$n=0, 1 \hdots$} 
 \State $\tilde{\state}(0), \tilde{\coninput}(t):t \in [0,T], \tsample \leftarrow$ Solve Problem \eqref{eqn:sampling_strategy}. \label{alg:line:SE:sampling}
 \If{\eqref{eqn:sampling_strategy} not feasible} terminate \label{alg:SE:terminate} \EndIf 
 \State Move from $x_s \to \tilde{\state}(0)$ (\cref{assump: safe-init}). \label{alg:SE:MoveInInit}
 \State Move from $\tilde{\state}(0) \to \tilde{\state}(\tsample)$ using $\tilde{\coninput}(t), t\in[0,\tsample)$. \label{alg:SE:MoveToSample}
 \State Update \GP using $y_\n = \constrain(\tilde{\state}(\tsample)) + \eta_\n$.
 \State Move from $\tilde{\state}(t') \!\to\! x(T)$ with $\tilde{\coninput}(t), t\in[\tsample\!,T)$. \label{alg:SE:MoveToInit} 
 \State Set $\state_s \!\leftarrow\! \tilde{\state}(T)$.
\EndFor
\end{algorithmic}
\label{alg:seocp} \vspace{-0.15em}
\end{algorithm}
\looseness -1 
\emph{Safe guaranteed exploration process:}
\cref{alg:seocp} summarizes the proposed framework. The agent starts at location $\state_s$ and navigates to $\tilde{\state}(0) \in \safeInit{\n}$ which can be traversed in $\timeX$ due to \cref{assump: safe-init}.  Next, the agent navigates to the sampling location $\tilde{\state}(\tsample)$ using control action $\tilde{\coninput}$ up to time $\tsample$. Finally, it collects the measurement, updates the \GP \eqref{eqn:posterior_update} and navigates back to $\tilde{\state}(T)\in \safeInit{\n}$. 
Thus, the duration of one sampling iteration is bounded by $T + \timeX$. 
We next resolve Problem \eqref{eqn:sampling_strategy} to identify the next informative location and the process continues until \eqref{eqn:sampling_strategy} is infeasible, i.e., $\ubconst[\n](\state) - \lbconst[\n](\state) <\epsilon, \forall x\in \Rcontoper[\safeInit{\n}] {\pessiSet[]{\n}}$. This implies there are no more safely reachable-returnable informative points and thus the algorithm terminates in \cref{alg:SE:terminate}. 

\looseness -1 In Problem \eqref{eqn:sampling_strategy}, we do not focus on sampling at a particular $\state$, as long as we satisfy the constraint. This provides an additional degree of freedom for prioritizing better sampling locations $\state$ in addition to safe exploration, which we leverage later in \cref{sec:goalSE} in a goal-directed setting. 
Note that we optimize over \mbox{$\state(0) \in \safeInit{\n}$}, and using controllable around terminal set property of \cref{assump: safe-init}, navigate to $\state(0)$ from $\state_s$, in contrast to starting from $\tilde{\state}(0)=\state_s$. This allows us to obtain uniform guarantees w.r.t. the set $\Rcontoper[\safeInit{\n}]{\pessiSet[]{\n}}$, which is independent of $\state_s$. 
The following corollary of \cref{thm:sample_complexity} extends safety and sample complexity results to dynamical systems.

\begin{corollary} [Sample complexity]  \looseness -1 Let \cref{assump:q_RKHS,assump: safe-init,assump:sublinear} hold. Consider $\n^\star$ according to \cref{thm:sample_complexity}. The closed-loop system resulting from \cref{alg:seocp} satisfies $\constrain(\state(t))\geq 0, \forall t \geq 0$ with probability at least $1-\prob$ and terminates in $\nfin \leq n^\star$ iterations.
\label{coro:sample_complexity}
\end{corollary}
\begin{proof} We know from \cref{thm:sample_complexity} that there exists $\nfin\leq \n^\star:$
\begin{align}
\!\!\!\epsilon \!>\! \max_{\state \in \pessiSet[]{\nfin}} \ubconst[\nfin](\state) - \lbconst[\nfin](\state) \geq \!\!\! \max_{\state \in \Rcontoper[\safeInit{\nfin}\!]{\pessiSet[]{\nfin}}} \!\!\!\ubconst[\nfin](\state) - \lbconst[\nfin](\state) \label{eqn:finite_terminate}
\end{align}
using $\Rcontoper[\safeInit{\nfin }]{\pessiSet[]{\nfin }} \subseteq \pessiSet[]{\nfin }$. 
Using sampling rule \eqref{eqn:sampling_strategy}, \cref{eqn:finite_terminate} implies \cref{alg:seocp} terminates within $\nfin\leq \n^\star$ samples.


\looseness -1 In \cref{alg:SE:MoveInInit}, the agent moves inside the set $\pessiSet[]{\n}$ (\cref{assump: safe-init}). In the rest, \cref{alg:SE:MoveToSample,alg:SE:MoveToInit}, the agent moves inside the reachable returnable pessimistic set, $\Rcontoper[\safeInit{\n}]{\pessiSet[]{\n}} \subseteq \pessiSet[]{\n}$. 
Thus $\forall t \geq 0,\state(t) \in \pessiSet[]{\n}$ which by \eqref{eqn:constrained_set} implies  $\lbconst[\n](\state(t)) \geq 0$ and hence by \cref{coro:hp_bounds} we get  $\constrain(\state(t)) \geq 0$.
\end{proof}

Thus, \cref{alg:seocp} retains the sample complexity result of \cref{thm:sample_complexity} while guaranteeing safe paths for a dynamical system. Since every iteration of \cref{alg:seocp} can take up to $T+\timeX$ time, the total physical runtime is bounded by $\n^\star(T+\timeX)$. Next, we focus on guaranteeing maximum safe domain exploration within this time.
\vspace{-0.4em}
\subsection{Guaranteeing full safe exploration for non-linear dynamics}
\label{sec:guaranteeingSE}
\looseness -1 We first define our notion of full exploration and then show how the safe exploration framework guarantees it. We make a mild regularity assumption on the reachable returnable sets to exclude pathological cases shown in \Cref{fig:conti_domain_eg}.
\begin{figure}
    \centering
    \scalebox{0.45}{\input{images/conti_domain_eg}}
    \caption{Illustration of the pathological cases on the left which are excluded by having \cref{assump: regularconnected}, while the domains on the right satisfy~\cref{assump: regularconnected}.} 
    \label{fig:conti_domain_eg}\vspace{-1em}
\end{figure}
\begin{assumption}\label{assump: regularconnected}
For any $\safeInit{\n} \subseteq \R^\sdim$ and $S \in \{  \optiSet[]{\n}, \pessiSet[]{\n}, \constSet[]{}, \X\}$, the reachable returnable set $\Rcontoper[\safeInit{\n}]{S}$ satisfies,
\begin{itemize}
    \item Regular closed set: $\Rcontoper[\safeInit{\n}]{S} = \closure{\interior{\Rcontoper[\safeInit{\n}]{S}}}$.
    \item Locally path-connected interior: $\forall x \in \partial \Rcontoper[\safeInit{\n}]{S}$, there exists a sufficiently small $\epsilon>0$ such that the set $B_{\epsilon}(x) \cap \interior{\Rcontoper[\safeInit{\n}]{S}}$ is path connected.
\end{itemize}
\end{assumption}
\looseness -1 The \textit{regular closed set} ensures the neighbourhood of boundaries has a non-empty interior, which removes degenerate cases, e.g., when the reachable returnable set is a curve in 2D space. The \textit{locally path-connected interior} ensures the set does not have intersecting boundaries, e.g. in 2D, this requires that subsets are (locally) not only connected through a single point.

\looseness-1 \textit{Full safe exploration:} Since, we learn about $\constrain$ up to $\epsilon$ tolerance, it is not possible to guarantee exploration beyond $\constSet[,\epsilon]{}$ safely. Thus, the desired behaviour is to guarantee the safe reachability and returnability to any state in $\Rcontoper[\safeInit{\n}]{\constSet[,\epsconst]{}}$. Hence, the best any algorithm can guarantee is the following:

\begin{tcolorbox}[colframe=white!, top=2pt,left=2pt,right=6pt,bottom=2pt]
\begin{objective}[Maximum safe domain exploration] There exists $\nfin \leq \n^\star$ such that \label{eqn:objective}
\begin{align}\vspace{-0.5em}
\!\!\!\Rcontoper[\safeInit{\nfin}\!]{\!\constSet[,\epsconst]{}}\! \subseteq\!\Rcontoper[\safeInit{\nfin}\!]{\!\optiSet[]{\nfin}} \! \subseteq \!\Rcontoper[\safeInit{\nfin}\!]{\!\pessiSet[]{\nfin}\!} \!\subseteq \!\Rcontoper[\safeInit{\nfin}\!]{\!\constSet[]{}}. \nonumber 
\end{align}
\end{objective}
\end{tcolorbox}
\looseness-1 In other words, we need to show the reachable returnable pessimistic set has expanded enough (or vice versa, the $\epsilon-$reachable returnable optimistic set has shrunk enough) to contain the $\epsilon-$reachable returnable optimistic set using a finite number of measurements $\nfin$. 
Due to monotonicity of the terminal sets $\safeInit{\n}$,  satisfaction of \cref{eqn:objective} also implies 
 $\Rcontoper[\safeInit{\n'}]{\pessiSet[]{\n'}} \supseteq \Rcontoper[\safeInit{0}]{\constSet[,\epsconst]{}}$, i.e.,
the pessimistic reachable returnable set is larger than the $\epsilon-$reachable returnable set with the terminal set $\safeInit{0}$. 

\begin{figure}
    \centering
    \scalebox{0.5}{\tikzset{every picture/.style={line width=0.75pt}} 

\begin{tikzpicture}[x=0.75pt,y=0.75pt,yscale=-1,xscale=1]

\draw  [fill={rgb, 255:red, 177; green, 177; blue, 255 }  ,fill opacity=1 ] (199.5,92.25) .. controls (213,81.25) and (432.33,51.75) .. (488,94.25) .. controls (543.67,136.75) and (517.67,232.5) .. (496,255) .. controls (474.33,277.5) and (417,283) .. (373,284) .. controls (329,285) and (289,283) .. (241,274) .. controls (193,265) and (152.5,245.75) .. (156,205) .. controls (159.5,164.25) and (167,145.25) .. (169.5,133.25) .. controls (172,121.25) and (186,103.25) .. (199.5,92.25) -- cycle ;
\draw  [fill={rgb, 255:red, 108; green, 215; blue, 108 }  ,fill opacity=0.6 ] (298.2,99.2) .. controls (306.6,92) and (353.5,97.75) .. (359.5,99.75) .. controls (365.5,101.75) and (384.6,105) .. (379.6,125) .. controls (374.6,145) and (345.4,134.45) .. (329,139.25) .. controls (312.6,144.05) and (300,158.75) .. (280,149.25) .. controls (260,139.75) and (289.8,106.4) .. (298.2,99.2) -- cycle ;
\draw   (347.39,195.32) -- (347.39,195.32) -- (342.7,200.1) -- (347.48,204.79) -- (347.48,204.79) -- (342.7,200.1) -- (338.01,204.88) -- (338.01,204.88) -- (342.7,200.1) -- (337.92,195.41) -- (337.92,195.41) -- (342.7,200.1) -- cycle ;
\draw [line width=1.5]    (342.7,200.1) -- (342.76,118.59) ;
\draw  [fill={rgb, 255:red, 80; green, 227; blue, 194 }  ,fill opacity=1 ] (336.18,118.59) .. controls (336.18,114.95) and (339.13,112) .. (342.76,112) .. controls (346.4,112) and (349.35,114.95) .. (349.35,118.59) .. controls (349.35,122.23) and (346.4,125.18) .. (342.76,125.18) .. controls (339.13,125.18) and (336.18,122.23) .. (336.18,118.59) -- cycle ;
\draw [line width=1.5]    (365,92) .. controls (352.85,80.55) and (329.33,79.36) .. (314.06,89.89) ;
\draw [shift={(311,92.25)}, rotate = 319.09] [fill={rgb, 255:red, 0; green, 0; blue, 0 }  ][line width=0.08]  [draw opacity=0] (11.61,-5.58) -- (0,0) -- (11.61,5.58) -- cycle    ;
\draw  [color={rgb, 255:red, 246; green, 9; blue, 9 }  ,draw opacity=1 ] (248.89,189.82) -- (248.89,189.82) -- (244.2,194.6) -- (248.98,199.29) -- (248.98,199.29) -- (244.2,194.6) -- (239.51,199.38) -- (239.51,199.38) -- (244.2,194.6) -- (239.42,189.91) -- (239.42,189.91) -- (244.2,194.6) -- cycle ;
\draw   (234,185.75) .. controls (241.5,179.75) and (253,182.25) .. (257.5,183.25) .. controls (262,184.25) and (281.5,191.5) .. (278.5,203.25) .. controls (275.5,215) and (249.5,211.75) .. (242.5,208.25) .. controls (235.5,204.75) and (226.5,191.75) .. (234,185.75) -- cycle ;
\draw  [color={rgb, 255:red, 246; green, 9; blue, 9 }  ,draw opacity=1 ] (261.08,198.33) -- (261.08,198.33) -- (256.39,203.11) -- (261.17,207.8) -- (261.17,207.8) -- (256.39,203.11) -- (251.7,207.89) -- (251.7,207.89) -- (256.39,203.11) -- (251.61,198.42) -- (251.61,198.42) -- (256.39,203.11) -- cycle ;
\draw  [color={rgb, 255:red, 246; green, 9; blue, 9 }  ,draw opacity=1 ] (272.39,191.82) -- (272.39,191.82) -- (267.7,196.6) -- (272.48,201.29) -- (272.48,201.29) -- (267.7,196.6) -- (263.01,201.38) -- (263.01,201.38) -- (267.7,196.6) -- (262.92,191.91) -- (262.92,191.91) -- (267.7,196.6) -- cycle ;
\draw  [color={rgb, 255:red, 246; green, 9; blue, 9 }  ,draw opacity=1 ] (259.39,186.32) -- (259.39,186.32) -- (254.7,191.1) -- (259.48,195.79) -- (259.48,195.79) -- (254.7,191.1) -- (250.01,195.88) -- (250.01,195.88) -- (254.7,191.1) -- (249.92,186.41) -- (249.92,186.41) -- (254.7,191.1) -- cycle ;

\draw (417.73,82.07) node [anchor=north west][inner sep=0.75pt]  [font=\Huge]  {$\mathcal{X}$};
\draw (352.45,106.4) node [anchor=north west][inner sep=0.75pt]  [font=\LARGE]  {$x_{s}$};
\draw (294.7,106.73) node [anchor=north west][inner sep=0.75pt]  [font=\LARGE]  {$\pessiSet[]{\n}$};
\draw (347.48,204.79) node [anchor=north west][inner sep=0.75pt]  [font=\LARGE]  {$o$};

\draw [color={rgb, 255:red, 0; green, 0; blue, 0 }  ,draw opacity=0.5 ][line width=3.75]    (391.57,172.71) .. controls (372.12,161.15) and (361.62,149.58) .. (349.6,127.78) ;
\draw [shift={(346.43,121.86)}, rotate = 60.48] [fill={rgb, 255:red, 0; green, 0; blue, 0 }  ,fill opacity=0.5 ][line width=0.08]  [draw opacity=0] (20.27,-9.74) -- (0,0) -- (20.27,9.74) -- (13.46,0) -- cycle    ;

\draw (390.7,170.73) node [anchor=north west][inner sep=0.75pt]  [font=\LARGE]  {$\Rcontoper[\safeInit{\n}]{\pessiSet[]{\n}}$};

\end{tikzpicture}}
    \caption{\looseness -1 Pendulum pivoting around $o$, with dynamics that restrict it to a counter-clockwise rotation. The red $\times$ represents an unsafe region which is a-priori unknown. Starting from a safe location $
    \state_s$, after gathering $n$ measurements, it identifies a pessimistic constrained set $\pessiSet[]{\n}$ (green). To gain further insight into the unsafe region, it needs to approach the boundary of $\pessiSet[]{\n}$. However, with unidirectional (non-controllable) dynamics, this cannot be done safely since the pendulum cannot turn back, i.e., once the state is close to the unsafe region and detects an obstacle, the collision can no longer be avoided. Such systems are excluded through \cref{assump: equal_boundary}.}
    \label{fig: assumption}\vspace{-1.5em}
\end{figure}

\looseness -1 However, achieving the \cref{eqn:objective} is not possible for general non-linear systems, e.g.,
if a system can not be controlled arbitrarily close to the boundary of the pessimistic constraint set, we cannot guarantee maximum exploration while being safe.
\Cref{fig: assumption} shows a simple example to illustrate this aspect further.
To avoid such uncontrollable systems, we make the following assumption on the dynamics. 
\begin{assumption}\label{assump: equal_boundary}
For any $\safeInit{\n} \subseteq \R^n$, and any continuous function $\lbconst[\n]$, $\partial \Rcontoper[\safeInit{\n}]{\pessiSet[]{\n}} \subseteq \partial \pessiSet[]{\n} \bigcup \partial \Rcontoper[\safeInit{\n}]{\X}.$
\end{assumption}
In \cref{apx:assum_local_control}, we show that \cref{assump: equal_boundary} is satisfied by locally controllable systems, which includes feedback linearizable systems, differential flat systems \cite{FAULWASSER20141151}, 
etc. Many real-world robotics platforms fall into these categories such as cars, wheeled robots, drones, etc.
   
Next, we utilize \cref{coro:sample_complexity} and show that \cref{alg:seocp} achieves the maximum safe domain exploration \eqref{eqn:objective}.
\begin{theorem} \label{thm:SE} 
\looseness -1 Let \cref{assump:q_RKHS,assump: safe-init,assump: regularconnected,assump: equal_boundary,assump:sublinear} hold. Consider $\n^\star$ according to \cref{thm:sample_complexity}. The closed-loop system resulting from \cref{alg:seocp} guarantees \cref{eqn:objective} with probability at least $1 - \prob$.
\end{theorem}
\begin{proof}
Using \cref{coro:sample_complexity} we know there exists $\nfin \leq \n^\star : \ubconst[\nfin](\state) - \lbconst[\nfin](\state) < \epsconst \ \forall \ \state \in \Rcontoper[\safeInit{\nfin}]{\pessiSet[]{\nfin}}$. Thus it remains to show that this implies \cref{eqn:objective}.

We first prove $\Rcontoper[\safeInit{\nfin}]{\optiSet[]{\nfin}}\subseteq \Rcontoper[\safeInit{\nfin}]{\pessiSet[]{\nfin}}$ by contradiction. Let us assume $\exists ~\state^\star \in \Rcontoper[\safeInit{\nfin}]{\optiSet[]{\nfin}} \backslash \Rcontoper[\safeInit{\nfin}]{\pessiSet[]{\nfin}}$. This implies there exists a continuous path $\zeta(b) \in \Rcontoper[\safeInit{\nfin}]{\optiSet[]{\nfin}} \forall b \in [0,1]$ such that $\zeta(0) \in \safeInit{\nfin}, \zeta(b^{\star}) =\state^\star$ and $\zeta(1) \in \safeInit{\nfin}$. Please see \Cref{fig: conti-opti-in-pessi} for a visual description. 
\begin{figure}
    \centering
    \scalebox{0.48}{\tikzset{every picture/.style={line width=0.75pt}} 

\begin{tikzpicture}[x=0.75pt,y=0.75pt,yscale=-1,xscale=1]

\draw  [fill={rgb, 255:red, 128; green, 128; blue, 128 }  ,fill opacity=0.3 ] (40,117.7) .. controls (40,85) and (66.5,58.5) .. (99.2,58.5) -- (523.8,58.5) .. controls (556.5,58.5) and (583,85) .. (583,117.7) -- (583,295.3) .. controls (583,328) and (556.5,354.5) .. (523.8,354.5) -- (99.2,354.5) .. controls (66.5,354.5) and (40,328) .. (40,295.3) -- cycle ;
\draw  [fill={rgb, 255:red, 177; green, 177; blue, 255 }  ,fill opacity=1 ] (156.5,96.5) .. controls (186,83) and (454.33,42.83) .. (510,85.33) .. controls (565.67,127.83) and (571.67,304.83) .. (550,327.33) .. controls (528.33,349.83) and (351.5,337.5) .. (302,332.67) .. controls (252.5,327.83) and (145.13,264.67) .. (133.07,243.83) .. controls (121,223) and (121.06,214.09) .. (119.5,193.5) .. controls (117.94,172.91) and (119.93,128.37) .. (126,122) .. controls (132.07,115.63) and (127,110) .. (156.5,96.5) -- cycle ;
\draw  [fill={rgb, 255:red, 245; green, 166; blue, 35 }  ,fill opacity=0.72 ] (126,122) .. controls (138,111.5) and (205,105.6) .. (227,113.2) .. controls (249,120.8) and (357,70.4) .. (391.4,76) .. controls (425.8,81.6) and (470.74,270.31) .. (453.4,282.8) .. controls (436.06,295.29) and (369.66,303.65) .. (339.8,304.8) .. controls (309.94,305.95) and (221,278.8) .. (210.2,276.8) .. controls (199.4,274.8) and (132.1,255.6) .. (125,230) .. controls (117.9,204.4) and (114,132.5) .. (126,122) -- cycle ;
\draw  [fill={rgb, 255:red, 10; green, 146; blue, 23 }  ,fill opacity=0.7 ] (365,100.4) .. controls (385,90.4) and (373.8,116.4) .. (386.6,112.8) .. controls (399.4,109.2) and (390.8,81.8) .. (410.8,111.8) .. controls (430.8,141.8) and (370.8,160.8) .. (350.8,130.8) .. controls (330.8,100.8) and (345,110.4) .. (365,100.4) -- cycle ;
\draw  [fill={rgb, 255:red, 108; green, 215; blue, 108 }  ,fill opacity=0.6 ] (126,122) .. controls (138,111.5) and (189,95.2) .. (227,113.2) .. controls (265,131.2) and (391,193.67) .. (386,213.67) .. controls (381,233.67) and (226.6,272) .. (210.2,276.8) .. controls (193.8,281.6) and (132.1,255.6) .. (125,230) .. controls (117.9,204.4) and (114,132.5) .. (126,122) -- cycle ;
\draw  [fill={rgb, 255:red, 231; green, 154; blue, 25 }  ,fill opacity=0.9 ] (64.5,305.9) .. controls (59.7,282.7) and (63.7,298.3) .. (69.7,274.3) .. controls (75.7,250.3) and (97.3,258.7) .. (117.7,260.7) .. controls (138.1,262.7) and (124.9,279.1) .. (132.5,290.7) .. controls (140.1,302.3) and (133,338.1) .. (129,342.5) .. controls (125,346.9) and (89.5,346.8) .. (85.5,346) .. controls (81.5,345.2) and (73.7,339.9) .. (69.7,334.7) .. controls (65.7,329.5) and (69.3,329.1) .. (64.5,305.9) -- cycle ;
\draw  [fill={rgb, 255:red, 10; green, 146; blue, 23 }  ,fill opacity=0.7 ] (84.2,266) .. controls (91.6,258.4) and (113,263.2) .. (118.4,272.4) .. controls (123.8,281.6) and (129.2,267.4) .. (128,300) .. controls (126.8,332.6) and (91.6,303.6) .. (85.6,296.8) .. controls (79.6,290) and (76.8,273.6) .. (84.2,266) -- cycle ;
\draw  [line width=3.75]  (421.04,198.85) .. controls (421.04,198.14) and (421.61,197.57) .. (422.32,197.57) .. controls (423.02,197.57) and (423.6,198.14) .. (423.6,198.85) .. controls (423.6,199.55) and (423.02,200.12) .. (422.32,200.12) .. controls (421.61,200.12) and (421.04,199.55) .. (421.04,198.85) -- cycle ;
\draw  [fill={rgb, 255:red, 0; green, 0; blue, 0 }  ,fill opacity=1 ] (382.33,217.25) .. controls (382.33,216.19) and (383.19,215.33) .. (384.25,215.33) .. controls (385.31,215.33) and (386.17,216.19) .. (386.17,217.25) .. controls (386.17,218.31) and (385.31,219.17) .. (384.25,219.17) .. controls (383.19,219.17) and (382.33,218.31) .. (382.33,217.25) -- cycle ;
\draw  [fill={rgb, 255:red, 80; green, 227; blue, 194 }  ,fill opacity=0.67 ] (138.62,239.67) .. controls (132.34,228.44) and (141.56,211.33) .. (159.22,201.45) .. controls (176.88,191.57) and (196.29,192.66) .. (202.57,203.89) .. controls (208.85,215.11) and (199.63,232.23) .. (181.97,242.11) .. controls (164.31,251.99) and (144.91,250.9) .. (138.62,239.67) -- cycle ;
\draw [line width=1.5]  [dash pattern={on 5.63pt off 4.5pt}]  (154.59,225.4) .. controls (174.97,241.13) and (429.33,228.95) .. (422.32,196.85) .. controls (415.3,164.75) and (397.93,122.05) .. (373.4,121.2) .. controls (349.6,120.37) and (222.66,191.99) .. (185.76,212.27) ;
\draw [shift={(182.6,214)}, rotate = 331.39] [fill={rgb, 255:red, 0; green, 0; blue, 0 }  ][line width=0.08]  [draw opacity=0] (13.4,-6.43) -- (0,0) -- (13.4,6.44) -- (8.9,0) -- cycle    ;
\draw [shift={(300.69,228.9)}, rotate = 175.59] [fill={rgb, 255:red, 0; green, 0; blue, 0 }  ][line width=0.08]  [draw opacity=0] (13.4,-6.43) -- (0,0) -- (13.4,6.44) -- (8.9,0) -- cycle    ;
\draw [shift={(403.42,145.55)}, rotate = 61.24] [fill={rgb, 255:red, 0; green, 0; blue, 0 }  ][line width=0.08]  [draw opacity=0] (13.4,-6.43) -- (0,0) -- (13.4,6.44) -- (8.9,0) -- cycle    ;
\draw [shift={(271.49,165.6)}, rotate = 332.31] [fill={rgb, 255:red, 0; green, 0; blue, 0 }  ][line width=0.08]  [draw opacity=0] (13.4,-6.43) -- (0,0) -- (13.4,6.44) -- (8.9,0) -- cycle    ;
\draw  [line width=1.5]  (153.57,225.4) .. controls (153.57,224.79) and (154.02,224.3) .. (154.59,224.3) .. controls (155.15,224.3) and (155.61,224.79) .. (155.61,225.4) .. controls (155.61,226.01) and (155.15,226.51) .. (154.59,226.51) .. controls (154.02,226.51) and (153.57,226.01) .. (153.57,225.4) -- cycle ;
\draw  [fill={rgb, 255:red, 231; green, 154; blue, 25 }  ,fill opacity=0.9 ] (53.6,182.6) .. controls (52.8,174.7) and (58.6,154.6) .. (63.2,147.6) .. controls (67.8,140.6) and (131.2,117.4) .. (126,122) .. controls (120.8,126.6) and (119.4,141.4) .. (119,146.4) .. controls (118.6,151.4) and (118.6,175.4) .. (118.8,182.4) .. controls (119,189.4) and (121.2,211.4) .. (122.6,220.2) .. controls (124,229) and (123.2,225.2) .. (125,230) .. controls (126.8,234.8) and (122.2,228.2) .. (113.2,225.4) .. controls (104.2,222.6) and (75.6,210.8) .. (65.8,201.6) .. controls (56,192.4) and (54.4,190.5) .. (53.6,182.6) -- cycle ;
\draw  [fill={rgb, 255:red, 10; green, 146; blue, 23 }  ,fill opacity=0.6 ] (66.2,156.1) .. controls (73.6,148.5) and (130.8,116) .. (126,122) .. controls (121.2,128) and (119.4,138.8) .. (118.8,148.6) .. controls (118.2,158.4) and (119.2,195.8) .. (121,206.2) .. controls (122.8,216.6) and (122,223.6) .. (125.8,231.4) .. controls (129.6,239.2) and (133.87,245.14) .. (133.07,243.83) .. controls (132.26,242.53) and (73.6,193.7) .. (67.6,186.9) .. controls (61.6,180.1) and (58.8,163.7) .. (66.2,156.1) -- cycle ;
\draw  [line width=3.75]  (382.04,217.85) .. controls (382.04,217.14) and (382.61,216.57) .. (383.32,216.57) .. controls (384.02,216.57) and (384.6,217.14) .. (384.6,217.85) .. controls (384.6,218.55) and (384.02,219.12) .. (383.32,219.12) .. controls (382.61,219.12) and (382.04,218.55) .. (382.04,217.85) -- cycle ;

\draw (85.7,271.5) node [anchor=north west][inner sep=0.75pt]  [font=\LARGE]  {$\pessiSet[]{\nfin}$};
\draw (68.23,77.07) node [anchor=north west][inner sep=0.75pt]  [font=\Huge]  {$\mathcal{X}$};
\draw (151.2,132.73) node [anchor=north west][inner sep=0.75pt]  [font=\LARGE]  {$\Rcontoper[\safeInit{\nfin}]{\pessiSet[]{\nfin}}$};
\draw (138.95,227.9) node [anchor=north west][inner sep=0.75pt]  [font=\LARGE]  {$\state_{s}$};
\draw (301.67,260.67) node [anchor=north west][inner sep=0.75pt]  [font=\LARGE]  {$\Rcontoper[\safeInit{\nfin}]{\optiSet[]{\nfin}}$};
\draw (415.53,203.73) node [anchor=north west][inner sep=0.75pt]  [font=\LARGE]  {$\zeta \left( b^{\star }\right)$};
\draw (364.53,221) node [anchor=north west][inner sep=0.75pt]  [font=\LARGE]  {$\zeta \left( b^{'}\right)$};
\draw (161.23,203.15) node [anchor=north west][inner sep=0.75pt]  [font=\LARGE]  {$\safeInit{\nfin}$};
\draw (423.2,105.73) node [anchor=north west][inner sep=0.75pt]  [font=\LARGE]  {$\Rcontoper[\safeInit{\nfin\!}]{\!\X}$};
\draw (78.7,160) node [anchor=north west][inner sep=0.75pt]  [font=\LARGE]  {$\pessiSet[]{\nfin}$};
\draw (69.2,311.5) node [anchor=north west][inner sep=0.75pt]  [font=\LARGE]  {$\optiSet[]{\nfin}$};

\end{tikzpicture}}
    \caption{Illustration of contradiction case in the proof of \cref{thm:SE} with $\zeta(b^{\star}) \in \Rcontoper[\safeInit{\nfin}]{\optiSet[]{\nfin}}\backslash\Rcontoper[\safeInit{\nfin}]{\pessiSet[]{\nfin}}$.}
    \label{fig: conti-opti-in-pessi}\vspace{-1.4em}
\end{figure}

Since $\zeta(b^\star) \notin \Rcontoper[\safeInit{\nfin}]{\pessiSet[]{\nfin}}$, $\exists~ b'<b^{\star}: \zeta(b') \in \partial \Rcontoper[\safeInit{\nfin}]{\pessiSet[]{\nfin}}$ which using \cref{lem: subtract_domain_boundary} implies $\zeta(b') \in \closure{\partial \Rcontoper[\safeInit{\nfin}]{\pessiSet[]{\nfin}} \backslash \partial \Rcontoper[\safeInit{\nfin}]{\X}}$. Furthermore, \cref{lem:boundary_inclusion} implies  
$\zeta(b') \in \closure{\partial \pessiSet[]{\nfin} \backslash \partial \X}$ and finally using \cref{lem: boundary-pessi-zero} we get $\lbconst[\nfin](\zeta(b'))=0$.

Since, $\zeta(b') \in \optiSet[]{\nfin}$ implies $\ubconst[\n](\zeta(b')) - \epsconst \geq 0$. Hence from the above two equations, $\ubconst[\nfin](\zeta(b')) - \lbconst[\nfin](\zeta(b'))  \geq \epsconst$. But $\zeta(b') \in \Rcontoper[\safeInit{\nfin}]{\pessiSet[]{\nfin}}$ and by \cref{eqn:finite_terminate} in \cref{coro:sample_complexity}, we know, $\ubconst[\nfin](\zeta(b')) - \lbconst[\nfin](\zeta(b')) < \epsconst$, which is a contradiction. Hence $\exists \nfin\! \leq \n^\star: \Rcontoper[\safeInit{\nfin\!}]{\optiSet[]{\nfin}}\!\subseteq\! \Rcontoper[\safeInit{\nfin\!}]{\pessiSet[]{\nfin}}$, which yields the second set inclusion in \cref{eqn:objective}. Moreover, note that using \cref{coro:hp_bounds}, $\forall n \geq 0, \constSet[,\epsconst]{} \subseteq \optiSet[]{\n} ~\mathrm{and}~ \pessiSet[]{\n} \subseteq \constSet[]{}$, which yields the other two set inclusions in \cref{eqn:objective}.
\end{proof}\vspace{-0.15mm}

Thus with sampling in the reachable returnable pessimistic set at locations having uncertainty higher than or equal to $\epsconst$~\eqref{eqn:sampling_strategy}, we achieve maximum domain exploration up to the $\epsconst$ tolerance \eqref{eqn:objective} while being safe at all times for non-linear systems with time duration of exploration bounded by $n^\star(T + \timeX)$.
While on a high level, the guarentees sound similar to discrete domain case \cite{turchetta2016safemdp,turchetta2019safe}, they are fundamentally different due to different definitions of the reachable-returnable safe sets. 
%

\looseness -1 The formulation in the continuous domain has several advantages, particularly in terms of more efficient exploration. Our proof of \cref{thm:SE} guarantees full exploration by relying solely on the condition that uncertainty at the boundary is less than $\epsconst$. 
Thus, it is sufficient to sample only at the boundaries of the reachable returnable pessimistic set to guarantee full exploration, potentially allowing to ignore a large part of the pessimistic set, which makes the algorithm efficient. 
In contrast to existing work for discrete domains \cite{turchetta2016safemdp,turchetta2019safe,prajapat2022near}, we do not explicitly compute the reachable returnable sets $\Rcontoper[\safeInit{\n}]{\pessiSet[]{\n}}$. Instead, we use an implicit characterization of $\Rcontoper[\safeInit{\n}]{\pessiSet[]{\n}}$ in terms of an optimal control problem 
which enables an efficient implementation, see \cref{sec:numerical} for details.

\vspace{-0.25em}
\section{Safe exploration using Lipschitz bound}
\label{sec:SE_expanders}


\looseness -1 In this section, we present the first technique to enhance efficiency: exploiting a Lipschitz bound of the constraint function $\constrain$.
We utilize the prior on function continuity to obtain a larger pessimistic set and additionally perform a targeted exploration using \emph{expanders}. Although the constraint function is a-priori unknown, in many cases, its Lipschitz constant is known. For this, we make the following assumption.
\begin{assumption}\label{assump:Lipschitz}
    The constraint $\constrain$ is $\LipConst$-Lipschitz continuous.
\end{assumption}
\looseness -1 Under \cref{assump:q_RKHS}, the function $\constrain$ satisfies Lipschitz continuity for kernels such as squared exponential and Mat\'ern \cite{safe-bo-sui15}. 
In case the constraint function $\constrain$ is simply the 2-norm distance to the nearest obstacle, the Lipschitz constant is trivial \mbox{$\LipConst = 1$}.
Incorporating the prior of $\LipConst$-Lipschitz continuity, we define the enlarged pessimistic and true constraint sets as:
\begin{align}
    \LpessiSet[]{\n} &\coloneqq \{\state \in \Domain| \exists z \in \Domain, \lbconst[\n](z) - \LipConst \|\state-z\| \geq 0 \}, \label{eqn: lipch-pessi-defi}\\ 
        \LconstSet[]{} &\coloneqq \{\state \in \Domain| \exists z \in \Domain, \constrain(z) - \LipConst \|\state-z\| \geq 0 \}.  \label{eqn: lipch-true-defi-domain}
\end{align}
\begin{lemma} Under \cref{assump:Lipschitz}, $\forall \n, \pessiSet[]{\n} \!\subseteq \LpessiSet[]{\n} \subseteq \LconstSet[]{} = \constSet[]{}$ holds with probability at least $1-\prob$.\label{lem:Lip_pessi_relation}
\end{lemma}
\begin{proof} By \cref{coro:hp_bounds}, $\lbconst[\n](\state) \leq \constrain(\state)$, which implies $\LpessiSet[]{\n} \subseteq \LconstSet[]{},\forall n \geq 1$. For any $x \in \pessiSet[]{\n}$, $\lbconst[\n](x) \geq 0$, hence $\exists z \coloneqq x : \lbconst[\n](z) - \LipConst \|\state-z\| \geq 0$ and thus $\state \in \LpessiSet[]{\n}$. Hence $\pessiSet[]{\n} \subseteq \LpessiSet[]{\n}$. Analogously we can prove $\constSet[]{} \subseteq \LconstSet[]{}$ and hence it suffices to prove $\LconstSet[]{} \subseteq \constSet[]{}$ to show $\LconstSet[]{} = \constSet[]{}$. Consider a $\state \in \LconstSet[]{}$, which implies $\exists z \in \X : \constrain(z) - \LipConst\|\state - z \| \geq 0$. Due to \cref{assump:Lipschitz}, $\constrain(z) - \constrain(\state) \leq \LipConst \| \state - z \|$. Hence $\constrain(z) - \LipConst\|\state - z \| \geq 0 \implies \constrain(\state)\geq 0$ and thus $\state \in \constSet[]{}$. 
\end{proof}
\looseness -1  While the true constraint set $\constSet[]{}$ is not altered due to true Lipschitz constant of the function, the set $\LpessiSet[]{\n}$ may be a bigger set as compared to $\pessiSet[]{\n}$; Naturally, the corresponding reachable returnable set  $\Rcontoper[{\safeInit{\n}}]{\LpessiSet[]{\n}}$ is also larger as shown in \cref{fig:Lipschitz_safe_set}. 
Note that, we do not enlarge the optimistic set since a tighter optimistic set enables faster convergence. 
Next, we define an \emph{expander}, a region around the boundary of a reachable returnable pessimistic safe set, where sampling may expand the pessimistic constraint set. 
\begin{figure}
    \centering
    \includegraphics[width=0.97\columnwidth]{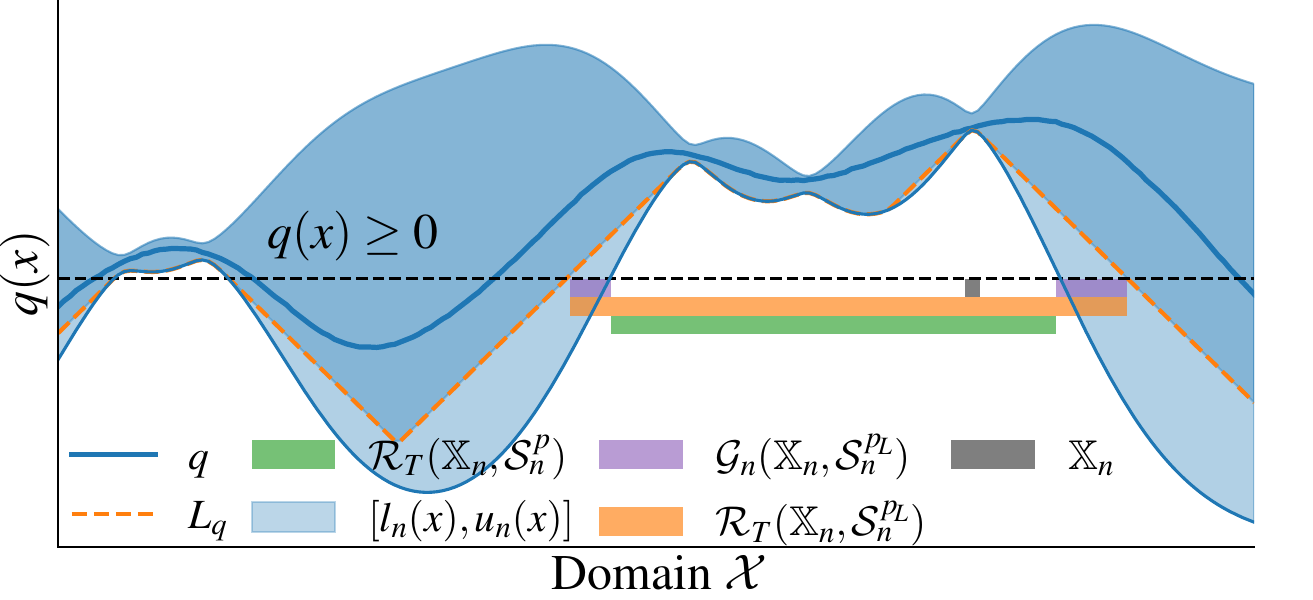}
    \caption{ \looseness -1 The set $\Rcontoper[{\safeInit{\n}}]{\pessiSet[]{\n}}$ in green is obtained with the lower bound $\lbconst[\n](\state)$. Utilizing the known Lipschitz constant, the lower bound is increased (orange line), leading to the enlarged set $\Rcontoper[{\safeInit{\n}}]{\LpessiSet[]{\n}}$ in orange. We efficiently sample from expanders $\expansion[\n](\safeInit{\n}, \LpessiSet[]{\n})$ (purple) near the boundary of $\Rcontoper[{\safeInit{\n}}]{\LpessiSet[]{\n}}$, ignoring a large part of the set $\Rcontoper[{\safeInit{\n}}]{\LpessiSet[]{\n}}$.}
    \label{fig:Lipschitz_safe_set}\vspace{-1.3em}
\end{figure}
The expander is defined as,
\begin{align}
\expansion[\n](\safeInit{\n}, \LpessiSet[]{\n}) \coloneqq \{ \state \in \Rcontoper[{\safeInit{\n}}]{\LpessiSet[]{\n}}~|~\lbconst[\n](\state) \leq 0 \} \label{eq:expander}
\end{align}




\looseness -1 The expander is designed such that reducing uncertainty within it is sufficient to guarantee maximum domain exploration.
Therefore, we can ignore a large portion of the pessimistic set to achieve exploration, which enhances efficiency.
Refer to \Cref{fig:Lipschitz_safe_set} for a visual description.
Similar to \eqref{eqn:sampling_strategy}, we define an exploration strategy in the expanders for determining $(n+1)^{th}$ sample,
\begin{align}
   \mathrm{Find} ~~x \in \expansion[\n](\safeInit{\n}, \LpessiSet[]{\n}) ~:~ \ubconst[\n](x) - \lbconst[\n](x) \geq \epsconst. \label{eqn:sampling_strategy_expander}
\end{align}
\looseness -1 The solution to \eqref{eqn:sampling_strategy_expander} yields a state and input trajectory $\tilde{\state}(t), \tilde{\coninput}(t)$, $t \in [0,T]$, and initial condition $\tilde{\state}(0) \in \safeInit{\n}$, and a sampling location $\state_{\n+1}$ at $\tilde{\state}(\tsample) \in \expansion[\n](\safeInit{\n}, \LpessiSet[]{\n})$ with some $t'\in[0,T]$. The execution remains the same as \cref{alg:seocp}. Similar to \cref{thm:SE}, the following theorem shows finite time safe guaranteed exploration with expanders.


\begin{theorem} \label{thm:SE_exp} 
Let \cref{assump:q_RKHS,assump: safe-init,assump: regularconnected,assump: equal_boundary,assump:Lipschitz,assump:sublinear} hold. 
Consider $\n^\star$ according to \cref{thm:sample_complexity} and \cref{alg:seocp} with \eqref{eqn:sampling_strategy} replaced by \eqref{eqn:sampling_strategy_expander} in \cref{alg:line:SE:sampling}. 
The resulting closed-loop system satisfies $\constrain(\state(t))\geq 0, \forall t \geq 0$, guarantees \cref{eqn:objective} 
under enlarged pessimistic set $\LpessiSet[]{\n}$ with probability at least $1-\prob$ and terminates in $\nfin \leq \n^\star$ iterations. 
\end{theorem}
\looseness -1 Thus, sampling recursively in $\expansion[\n](\safeInit{\n}, \LpessiSet[]{\n})$ with an enlarged pessimistic set guarantees convergence to the true reachable returnable set up to tolerance $\epsilon$. The proof can be found in the extended version \cite[Appendix B]{prajapat2024safe}. Note that in the extreme case of $\LipConst \to \infty$, the expander $\expansion[\n](\safeInit{\n}, \LpessiSet[]{\n})$ reduces to the boundary of the pessimistic set and \cref{thm:SE_exp} ensures that it is sufficient to sample on the boundary for complete exploration. 
\section{Goal directed safe exploration}\label{sec:goalSE}
In this section, we introduce our second technique aimed at enhancing efficiency by considering a goal-directed safe exploration. 
In many applications, such as reaching a goal, maximum domain exploration can be wasteful, 
and we desire to explore only the region essential for the objective. In these cases, safe exploration is a mere consequence of the objective, which is \emph{reaching the goal safely}.
We consider the problem of reaching a \textit{steady-state} that minimizes a loss function $\utility:\R^\sdim \to \R$ while ensuring safety at all times. We assume the loss function $\rho$ is continuous and bounded.
Given that we can learn about the constraint only up to $\epsconst$ precision, the best any algorithm can guarantee safely is the following.
\begin{tcolorbox}[colframe=white!, top=2pt,left=2pt,right=2pt,bottom=2pt]
\begin{objective}[Goal-directed safe exploration with bounded regret] 
\label{eqn:goal_objective}\looseness -1 There exists $\nfin\leq\n^\star$ such that average regret \AvR satisfies,
\begin{align*} \vspace{-0.5em}
\!  \!  \AvR(\tau) \! \coloneqq \frac{1}{\tau} \! \int_{0}^{\tau} \!\!\!\left(\!\utility(\state(t)) - \!\! \min_{\state \in \Rcontoper[\safeInit{\n^\prime}]{\constSet[,\epsconst]{}}}\!\utility(\state)\!\right) \!dt \leq\! \mathcal{O}\!\left(\frac{1}{\tau}\right)\!. 
\end{align*}
\end{objective}
\end{tcolorbox}
\looseness -1 To achieve \cref{eqn:goal_objective}, naively, one could explore the complete domain utilizing safe exploration \cref{alg:seocp} and then plan to minimize $\utility$ in the known safe domain. However, it is inefficient to explore the complete domain without considering the objective. One simple strategy is to explore the locations that greedily minimize the loss in the pessimistic set.
Although a decently performing heuristic, it may get stuck in a local minimum since it lacks active constraint exploration, which is crucial for guaranteeing loss minimization \cite{safe-bo-sui15}. 
We, therefore, propose a strategy that makes use of the optimistic set, which only excludes the locations that are provably unsafe, such that a greedily minimizing strategy in the optimistic set can guarantee loss minimization among the safe locations.
For this, we define the \emph{optimistic} and \emph{pessimistic goal} respectively as,
\begin{align}
    \gostateOpti[\n] &= \argmin_{x \in \Rcontoper[\safeInit{\n}]{\optiSet[]{\n}}} \utility(x): \exists \coninput \in \inputSpace, \dyn(\state, \coninput) = 0, \label{eq:goal_sampling} \\
    \gostatePessi[\n] &=  \argmin\limits_{\state \in \Rcontoper[\safeInit{\n}]{\pessiSet[]{\n}}}  \ \, \utility(\state): \exists \coninput \in \inputSpace, \dyn(\state,\coninput)=0.\label{eq:pessi_goal_sampling}
\end{align}
\begin{algorithm}[!t]
\caption{Goal-directed safe exploration}
\begin{algorithmic}[1]
\State \textbf{Initialize:} Start at $\state_s \in \safeInit{0}$, \GP, Horizon $T$, Tolerance $\epsilon$.
 \For{$n=0, 1 \hdots$}
 \State $\gostateOpti[\n] \leftarrow$ Solve Problem \eqref{eq:goal_sampling}. \label{alg:GO:solveGO}
 \State $\gostatePessi[\n]  \leftarrow$ Solve Problem \eqref{eq:pessi_goal_sampling}. \label{alg:GO:get_pessi_best}
\If {$\utility(\gostatePessi[\n]) \leq \utility(\gostateOpti[\n])$} \label{alg:GO:check_val}
\State Move from  $\state_s \to \gostatePessi[\n]$ and terminate.  \label{alg:GO:moveToSteady}
\EndIf
 \State $x(0), a(t)\!:\!t \in [0,T], t' \leftarrow$ Solve Problem \eqref{eqn:goal_close}. \label{alg:GO:closest_loc}
 \State Go: $x_s \to x(0)$. $x(0) \to x(t')$ using $a(t), t\in[0,t')$. \label{alg:GO:moveToSample}
 \State Update \GP: $y = q(x(t')) + \eta$.
 \State Go: $x(t') \!\to\! x(T)$ with $a(t), t\in[t'\!,T)$, set $\state_s \!\leftarrow\! x(T)$.
\EndFor
\end{algorithmic}
\label{alg:GOsampling}
\end{algorithm}
\looseness -1 \emph{Goal-directed safe exploration algorithm:} The process is summarized in \cref{alg:GOsampling}. In \eqref{eq:goal_sampling}, $\gostateOpti[\n]$ represents a steady-state goal location in the optimistic set, which may or may not be in the pessimistic set. 
If the pessimistic goal $\gostatePessi[\n]$ is better than the optimistic one, $\gostateOpti[\n]$, the agent moves to the location and the algorithm terminates (\cref{alg:GO:moveToSteady}). However, if this is not the case, we find the closest safe location to the optimistic goal state, $\gostateOpti[\n]$ by solving the following in \cref{alg:GO:closest_loc},
\begin{align}
    \min\limits_{\state \in \Rcontoper[\safeInit{\n}]{\pessiSet[]{\n}}} \dist(x, \gostateOpti[\n])~:~\ubconst[\n](\state)-\lbconst[\n](\state)\geq \epsconst. \label{eqn:goal_close}
\end{align}
\looseness-1 $\dist(x,\gostateOpti[\n])$ can be freely designed by a user, see \cref{rem:dist} for cost functions that may speed up exploration by utilizing the optimistic goal in the optimization. We move to $\state(t') $ in \cref{alg:GO:moveToSample} and collect a constraint measurement.
The optimistic and pessimistic constraint sets change with every constraint posterior update. Consequently, we evaluate whether the goal remains safely reachable and returnable after every posterior update in \cref{alg:GO:solveGO}. 
A change in the goal $\gostateOpti[\n]$ implies that there cannot exist any safe path to reach and return from the previous goal.
The following theorem summarizes the theoretical guarantees for Goal-directed safe exploration.



\begin{theorem} 
\label{thm:go} \looseness -1 Let \cref{assump:q_RKHS,assump: safe-init,assump: regularconnected,assump: equal_boundary,assump:sublinear} hold and consider $\n^\star$ as in \cref{thm:sample_complexity}. Then the closed loop system resulting from \cref{alg:GOsampling} satisfies $\constrain(\state(t))\geq 0, \forall t \geq 0$ and guarantees \cref{eqn:goal_objective} with probability at least $1 - \prob$. 
\end{theorem}
\begin{proof}  
\looseness -1 Using \cref{coro:sample_complexity}, there exists $\nfin \leq \n^\star$ such that $\forall x \in \Rcontoper[\safeInit{\nfin}]{\pessiSet[]{\nfin}},~\sumMaxwidth[]{\nfin}(\state) < \epsconst$ and by \cref{thm:SE}, $\Rcontoper[\safeInit{\nfin}]{\optiSet[]{\nfin}} \subseteq \Rcontoper[\safeInit{\nfin}]{\pessiSet[]{\nfin}}$, which implies that
the solution $\gostatePessi[\nfin]$ in the pessimistic set obtained in \cref{alg:GO:get_pessi_best} satisfies $\utility(\gostatePessi[\nfin]) \leq \utility(\gostateOpti[\nfin])$. The algorithm thus terminates in \cref{alg:GO:check_val} within $\nfin \leq \n^\star$ iterations.

\looseness -1 Let us define $\state^\star_{\nfin} \!\! \coloneqq \!\argmin\nolimits_{\state \in \Rcontoper[\safeInit{\nfin}]{\constSet[,\epsconst]{}}}\!\utility(\state)$. After termination, the agent stays at the steady-state $\gostatePessi[\nfin]$, hence $\forall t \geq t_{\nfin}, \utility(\state(t)) \leq \utility(\gostateOpti[\nfin])\leq \utility(\state^\star_{\nfin})$. The last inequality follows since $\Rcontoper[\safeInit{\nfin}]{\constSet[,\epsconst]{}} \subseteq \Rcontoper[\safeInit{\nfin}]{\optiSet[]{\nfin}}$ by \cref{thm:SE}. Since $\utility$ is bounded, for any $x_1,x_2\in\mathcal{X}, \utility(x_1)-\utility(x_2)\leq \bar{\utility}$  with some $\bar{\utility}>0$. Thus, $\forall \tau \geq 0, \int_{0}^{\tau} \! \utility(\state(t)) - \utility(\state^\star_{\nfin}) dt \leq \bar{\utility} t_{n'}\leq \bar{\utility}n^\star(T + \timeX).$ Hence, averaging the bound over $\tau$ implies $\AvR(\tau) \leq \mathcal{O}(1/\tau)$, i.e., \cref{alg:GOsampling} achieves \cref{eqn:goal_objective}.
\end{proof}
\looseness -1 Thus, \cref{alg:GOsampling} is guaranteed to converge to a goal that ensures a suitable bound on the regret (\cref{eqn:goal_objective}). Goal-directed safe exploration mainly differs from full safe exploration (\cref{alg:seocp}) in utilizing the reachable returnable optimistic set, $\Rcontoper[\safeInit{\n}]{\optiSet[]{\n}}$ to set an appropriate goal recursively and terminates early once the objective is achieved.

\begin{remark} \label{rem:dist}
    \looseness -1 $\dist(\state,\gostateOpti[\n])$ can be any function which satisfies $\argmin_{\state \in \X} \dist(\state,\gostateOpti[\n])=\gostateOpti[\n]$. For instance, 
    \dist can be  $\|x-\gostateOpti[\n]\|$, minimum time to go from $x$ to $\gostateOpti[\n]$, or 
    any combination of distance and $\rho$, such as $\dist(x,x_g)=a\|x-x_g\|^2+b\rho(x)$ for some $a,b > 0 \in \R$.
\end{remark}

\vspace{-0.2em}
\section{Safe exploration MPC} \label{sec:MPC}
\looseness -1 In this section, we propose \sempc, our core \mpc-type algorithm. In the described safe exploration approaches so far, the agent returns to the terminal set $\safeInit{\n}$ after each measurement. 
Instead, we can replan for the next-to-go location as soon as we collect a measurement. This can yield a significant reduction in the exploration time. 
In addition to re-planning in receding horizon fashion, to enhance efficiency \sempc can incorporate a goal-directed technique (\cref{sec:goalSE}) and exploit a known Lipschitz constant to enlarge the pessimistic set (\cref{sec:SE_expanders}). For simplicity of exposition, we present \sempc directly as goal-directed and without utilizing Lipschitz continuity, while the corresponding variations are discussed in \cref{rem:seMPCwithexpanders,rem:GoaltoSE}, respectively. 
\begin{algorithm}[!t]
\caption{\sempc}
\begin{algorithmic}[1]
\State \textbf{Initialize:} Start at $\state_s \in \safeInit{0}$, \GP, Horizon $T$, Tolerance $\epsilon$.
\For{$\n = 0, 1, \hdots $} 
 \State $\gostateOpti[\n] \leftarrow$ Solve Problem \eqref{eq:goal_sampling}. \label{alg:MPC:get_opti_goal} 
 \State $\gostatePessi[\n]  \leftarrow$ Solve Problem \eqref{eq:pessi_goal_sampling}. \label{alg:MPC:get_pessi_best}
\If {$\min(\utility(\state(t_\n)), \utility(\gostatePessi[\n])) \leq \utility(\gostateOpti[\n])$} 
 \label{alg:MPC:go_termination}
\State Accordingly move to $\state(t_\n)\!$ or $\gostatePessi[\n]\!$ and terminate. 
\label{alg:MPC:move_SteadyState}
\EndIf
\State $\tilde{\state}(t):t \in [0,T], t', \nu \leftarrow$ Solve Problem \eqref{eq:exact_prob}.  \label{alg:MPC:exact_xtn}
\If{$\nu = 0$} \label{eq:case_distinction} 
\State Move from $\state(t_{\n}) \to \state(t_{\n+1}) = \tilde{\state}(t')$. \label{alg:MPC:moveTOsampling}
\Else
\State Move from  $\state(t_{\n}) \to \tilde{\state}(T) \in \safeInit{\n}$. \label{alg:MPC:last_control_input}
 
\State $\tilde{\state}(t)\!:t \in [0,T], t' \!\leftarrow$ Solve Problem \eqref{eqn:goal_close}. \label{alg:MPC:exact_Xn}
 \State Move to $\state(t_{\n+1}) = \tilde{\state}(t')$. \label{alg:MPC:move2}
\EndIf
\State Update \GP with $y = q(x(t_{\n+1})) + \eta$. \label{alg:MPC:updateGP}
\EndFor
\end{algorithmic}
\label{alg:go-se-mpc}
\end{algorithm}
\looseness -1 Analogously to goal-directed safe exploration, we set an optimistically safe reachable returnable goal, $\gostateOpti[\n]$, by solving \eqref{eq:goal_sampling}. 
If the goal $\gostateOpti[\n]$ is not in the reachable returnable pessimistic set, we shall determine a safe location, close to $\gostateOpti[\n]$ with uncertainty above $\epsilon$ to ensure sufficient information to keep exploring.
This naturally requires us to solve a constrained optimization problem \eqref{eqn:goal_close}, however, the main difference from goal-directed safe exploration is that \sempc uses the current location as the starting point $\state(t_\n)$ to ensure re-planning instead of returning to the safe set. 
Note that even if Problem \eqref{eqn:goal_close} is feasible, re-planning starting from the current state may result in an infeasible problem,  $\nexists x \in \Rcontoper[\{ \state(t_\n) \},\safeInit{\n}]{\pessiSet[]{\n}}: \sumMaxwidth[]{\n}(\state) \geq \epsconst$, i.e., there is no location with uncertainty at least $\epsconst$ that can be reached from the current location and return to the terminal set $\safeInit{n}$ pessimistically in time $T$.
Hence, 
we soften this constraint and solve the following optimization problem to ensure the feasibility of \sempc, 
\begin{align}
\min\limits_{\nu, \state \in \Rcontoper[\{\state(t_{\n})\},\safeInit{\n}]{\pessiSet[]{\n}}} \;\; &\lambda \nu + \dist(\state, \gostateOpti[\n]) \nonumber \\
  \mathrm{s.t.} \quad &\epsconst - \sumMaxwidth[]{\n}(x) \leq \nu, \nu \geq 0, \label{eq:exact_prob}
\end{align}
where $\lambda>0$ is a penalty and $\nu$ is a slack variable to ensure the feasibility of the problem. 

\looseness -1 To achieve \cref{eqn:goal_objective}, we want to guarantee exploration of the set $\Rcontoper[\safeInit{\n}]{\pessiSet[]{\n}}$, however the Problem \eqref{eq:exact_prob} solves for a sampling location in the set $\Rcontoper[\{\state(t_\n)\},\safeInit{\n}]{\pessiSet[]{\n}} \neq \Rcontoper[\safeInit{\n}]{\pessiSet[]{\n}}$. 
To address this issue, instead we solve Problem \eqref{eqn:goal_close} when the slack value $\nu>0$, i.e., $\sumMaxwidth[]{\n}(\tilde{x}(t')) < \epsconst$.


\looseness -1 \emph{\sempc steps:} The resulting approach is summarized in \cref{alg:go-se-mpc}. Analogous to Goal-directed safe exploration (\cref{alg:GOsampling}), we start by solving for the optimistic goal, $\gostateOpti[\n]$, in \cref{alg:MPC:get_opti_goal}. To sample close to the goal $\gostateOpti[\n]$, we recursively solve Problem \eqref{eq:exact_prob} in \cref{alg:MPC:exact_xtn}. Problem \eqref{eq:exact_prob} also ensures instantaneously re-planning from $\state(t_{\n})$ instead of returning to the terminal set, as shown in \Cref{fig: conti-algorithm-steps}.  
In Problem \eqref{eq:exact_prob},  with $\nu=0$ we are guaranteed to have \emph{sufficient information} with $\sumMaxwidth[]{\n}(\tilde{\state}(t')) \geq \epsilon$ and a \emph{goal directed} approach by being $\tilde{\state}(t')$ close to the goal location $\gostateOpti[\n]$ encoded by the \dist function.  
However if this is not the case, i.e., $\nu>0$ which implies $\sumMaxwidth[]{\n}(\tilde{\state}(t')) < \epsilon$,  the agent returns to the terminal set, $\safeInit{\n}$ in \cref{alg:MPC:last_control_input} as shown in \Cref{fig: conti_mpc_time} and solves Problem \eqref{eqn:goal_close} in \cref{alg:MPC:exact_Xn} to obtain the next goal to go in the set $\Rcontoper[\safeInit{\n}]{\pessiSet[]{\n}}$. The agent first moves to $\tilde{x}(0)$ in time $\timeX$ due to the controllability property from \cref{assump: safe-init}. Afterwards, the agent moves to the sampling location $\tilde{\state}(t')$ in \cref{alg:MPC:move2} where it collects a constraint sample, performs a \GP update (\cref{alg:MPC:updateGP}), and the goal-directed safe exploration continues until the termination criteria in \cref{alg:MPC:go_termination} gets satisfied.
Notably, the case of $\nu>0$ rarely occurs in practice (see \cref{sec:demo_car_dynamics}) and thus \sempc mainly determines the next sampling location based on the current state in Problem \eqref{eq:exact_prob}, moves to the location for sampling and updates the \GP in \cref{alg:MPC:updateGP}.

\begin{figure}
    \centering
    \scalebox{0.5}{\input{images/conti-algorithm-steps}}
    \caption{\looseness 0 Illustration of re-planning in \sempc: The agent starts at $\state_s \in \safeInit{\n}$ and plans a trajectory in $\Rcontoper[\safeInit{\n}]{\pessiSet[]{\n}}$ that reaches the $(n+1)^{th}$ sampling location, $\state(t_{\n+1})$ (blue line) and returns to the terminal set (black line). After sampling at $\state(t_{\n+1})$, the agent updates the optimistic and pessimistic constraint sets. Instead of returning, the agent re-plans for the new sampling location, $\state(t_{\n+2})$ in the set $\Rcontoper[\{\state(t_{\n+1})\},\safeInit{\n+1}]{\pessiSet[]{\n+1}}$. The process continues until it reaches optimistic goal $\gostateOpti[\n]$.}
    \label{fig: conti-algorithm-steps} \vspace{-1.0em}
\end{figure}

In the following, we prove that \sempc is recursively feasible and ensures the safety of the non-linear system.
\begin{proposition}[Recursive feasibility]\label{prop:recursive_feas} Let 
\cref{assump:q_RKHS,assump: safe-init,assump: regularconnected,assump: equal_boundary,assump:sublinear} hold and suppose $\Rcontoper[\safeInit{0}]{\pessiSet[,\epsconst]{0}}\not=\emptyset$. Then all the optimization problems in \cref{alg:go-se-mpc} are feasible for $\forall \n \geq 0$. Furthermore, with probability at least $1-\prob$, the resulting closed-loop system satisfies $\constrain(\state(t))\geq0 ~\forall t \geq 0$. 
\end{proposition}
\begin{proof} 
Note that $\safeInit{0} \subseteq \safeInit{\n}$ by \cref{assump: safe-init}, and
$\forall \n \geq 0, 
\optiSet[]{\n} \supseteq \pessiSet[,\epsconst]{\n} \supseteq \pessiSet[,\epsconst]{0}$, 
together imply $\Rcontoper[\safeInit{\n}]{\optiSet[]{\n}} \supseteq \Rcontoper[\safeInit{0}]{\pessiSet[,\epsconst]{0}} \neq \emptyset$. Hence the problem in \cref{alg:MPC:get_opti_goal} is always feasible.
Analogously, the problem in \cref{alg:MPC:get_pessi_best} is feasible 
since $\Rcontoper[\safeInit{\n}]{\pessiSet[]{\n}} \supseteq \Rcontoper[\safeInit{0}]{\pessiSet[,\epsconst]{0}}  \neq \emptyset$. 

Next, we show that Problem \eqref{eqn:goal_close} in \cref{alg:MPC:exact_Xn} is always feasible. For contradiction assume that Problem \eqref{eqn:goal_close} is infeasible, i.e., $\sumMaxwidth[]{\n}(\state)<\epsilon \forall x \in \Rcontoper[\safeInit{\n}]{\pessiSet[]{\n}}$ which by \cref{thm:SE} implies $\Rcontoper[\safeInit{\n}]{\optiSet[]{\n}} \subseteq \Rcontoper[\safeInit{\n}]{\pessiSet[]{\n}}$. This implies 
$\utility(\gostatePessi[\n]) \leq \utility(\gostateOpti[\n])$ and the algorithm should have terminated in \cref{alg:MPC:go_termination} already.


\looseness -1 Feasibility of  Problem \eqref{eq:exact_prob} in \cref{alg:MPC:exact_xtn} at any $t_\n \geq 0$ is ensured using the standard \mpc candidate input $\hat{\coninput}(\cdot)\in \pwcinput, \hat{\state}(\cdot) \subseteq \Domain$: shifting the previous feasible solution $\hat{\coninput}(t) = \tilde{\coninput}(t' + t), t \in [0,T-t')$ and appending  $\hat{\coninput}(t)=\kappa_n(\hat{\state}(t))$, $t\in[T-t',T)$ from \cref{assump: safe-init}. 
Feasibility follows from $\hat{\state}(t) \in \pessiSet[]{\n-1} \subseteq \pessiSet[]{\n}, \forall t \in [0, T-t']$ (previous feasible solution); and $\hat{\state}(t) \in \safeInit{\n-1} \subseteq \safeInit{\n} \subseteq \pessiSet[]{\n}, \forall t\in[T-t',T]$ respectively due to ``control invariance", ``monotonicity" and ``pessimistically safe" properties (\cref{assump: safe-init}). 
Moreover, note that the constraint $\epsconst-\sumMaxwidth[]{\n}(\state)\leq \nu$ in Problem \eqref{eq:exact_prob} is always feasible by choosing $\nu$ sufficiently large.
This implies all the optimization problems from \cref{alg:MPC:get_opti_goal,alg:MPC:get_pessi_best,alg:MPC:exact_Xn,alg:MPC:exact_xtn,} are feasible $\forall \n \geq 0$. 

In \cref{alg:go-se-mpc},  it holds that, $x(t)\in \pessiSet[]{\n}$ which by definition~\eqref{eqn:constrained_set} implies $\lbconst[\n](\state(t)) \geq 0$ and hence by  \cref{coro:hp_bounds} implies $\constrain(\state(t)) \geq 0$ with probability at least $1-\prob$.
\end{proof}

\looseness -1 We next establish that \sempc recovers the theoretical guarantees of the safe exploration framework (\cref{sec:goalSE}) by combining results of \cref{thm:SE,thm:go}.
    \begin{theorem} \label{coro:sempc} \looseness -1 Let \cref{assump:q_RKHS,assump: safe-init,assump: regularconnected,assump: equal_boundary,assump:sublinear} hold, and consider $\n^\star$ as in \cref{thm:sample_complexity}. Then the closed-loop system resulting from \cref{alg:go-se-mpc} guarantees \cref{eqn:goal_objective} with probability at least $1 - \prob$.
\end{theorem}
 \begin{proof} \looseness -1 \cref{alg:go-se-mpc} ensures that in \cref{alg:MPC:updateGP}, we sample at a location $\state(t_{\n+1}): \sumMaxwidth[]{\n}(\state(t_{\n+1}))\geq \epsilon$. 
  In particular, $\state(t_{\n+1})$ is chosen based on Problem \eqref{eq:exact_prob} if $\nu=0$ or Problem \eqref{eqn:goal_close}, which both enforce $\sumMaxwidth[]{\n}(\state(t_{\n+1}))\geq \epsilon$.
 For this sampling rule, using \cref{coro:sample_complexity}, there exists $\nfin \leq \n^\star$ such that $\forall x \in \Rcontoper[\safeInit{\nfin}]{\pessiSet[]{\nfin}},~\sumMaxwidth[]{\nfin}(\state) < \epsconst$. 
Finally this yields $\Rcontoper[\safeInit{\nfin}]{\optiSet[]{\nfin}} \subseteq \Rcontoper[\safeInit{\nfin}]{\pessiSet[]{\nfin}}$ by \cref{thm:SE}, which implies $\utility(\gostatePessi[\nfin]) \leq \utility(\gostateOpti[\nfin])$ on solving Problems in \cref{alg:MPC:get_opti_goal,alg:MPC:get_pessi_best}. Thus the algorithm terminates at \cref{alg:MPC:go_termination} in $\nfin \leq \n^\star$ iterations.

\looseness -1 The total time is bounded by $\n^\star(T + \timeX)$. As shown in \cref{fig: conti_mpc_time}, in the extreme case, $\nu\! >\! 0$,\! the agent returns from $\state(t_\n)$ into $\safeInit{\n}$ using the last control sequence (\cref{alg:MPC:last_control_input}) in the same time T. It re-evaluates (\cref{alg:MPC:exact_Xn}) the next sampling location, where the time for iteration is again bounded by $T + \timeX$. Utilising this and \cref{alg:MPC:go_termination}, the regret guarantees follow similar to \cref{thm:go}.\!\!\!\!
\end{proof}

\begin{figure}
    \centering
    \scalebox{0.5}{\input{images/conti_mpc_time}}
        \caption{\looseness -1 Illustration of $\nu>0$ case in \cref{alg:MPC:exact_xtn} \cref{alg:go-se-mpc} at location $\state(t_{\n})$. The solution of Problem~\eqref{eq:exact_prob} results in $\nu>0$ (red trajectory) which corresponds to uncertainty less than $\epsconst$. Hence the agent returns back to the safe set $\safeInit{\n}\!$ from $\state(t_{\n})$ following the last computed trajectory (blue line) in \cref{alg:MPC:last_control_input}. 
    Then it optimizes for the next sampling location $\state(t_{\n+1})$ in \cref{alg:MPC:exact_Xn} and the exploration process continues. Note that $x(t_{\n+1}) \not\in \Rcontoper[\{\state(t_{\n}), \safeInit{\n}]{\pessiSet[]{\n}}$ and thus we return back to the safe set $\safeInit{\n}$ and then go to $\state(t_{\n+1})$.}
    \label{fig: conti_mpc_time} \vspace{-1.0em}
\end{figure}
\begin{remark}[\sempc for maximum safe domain exploration]\label{rem:GoaltoSE}
\looseness -1 The maximum safe domain exploration can be recovered as a special case of \cref{alg:go-se-mpc} by two crucial changes. 
First, replace the termination criterion in \cref{alg:MPC:go_termination} by infeasibility check of Problem \eqref{eqn:goal_close} in \cref{alg:MPC:exact_Xn}. 
Second, for efficient domain exploration, define optimistic goal $\gostateOpti[\n] \coloneqq \argmax_{x \in \Rcontoper[\safeInit{\n}]{\optiSet[]{\n}}}\sumMaxwidth[]{\n}(\state)$ and Problem \eqref{eq:exact_prob} ensures sampling closer to $\gostateOpti[\n]$ encoded by \dist while satisfying $\sumMaxwidth{\n}(x) \geq \epsilon$ until termination, which is sufficient to guarantee maximum domain exploration (\cref{thm:SE}).
\end{remark}
\begin{remark}[\sempc with expanders] \label{rem:seMPCwithexpanders}
\looseness -1 If \cref{assump:Lipschitz} is satisfied, \sempc can be modified to be more sample efficient by exploiting $\LipConst$. 
This can be done by using a larger pessimistically reachable returnable safe set, i.e., replacing $\pessiSet[]{\n}$ by $\LpessiSet[]{\n}$ in Problem~\eqref{eq:exact_prob}.
Furthermore, we can restrict sampling within the expander, $\expansion[\n](\safeInit{\n},\LpessiSet[]{\n})$ by adding the constraint $\lbconst[\n]{}(\state) \leq 0$ in \eqref{eq:exact_prob}.
Combining both results in more efficient exploration while preserving the guarantees of \cref{coro:sempc}.
\end{remark}




\section{Simulation Results}
\label{sec:simulations}
\looseness -1 In this section, we first discuss the numerical implementation of \sempc in \cref{sec:numerical}. Next, we validate our theory and compare different safe exploration techniques on the task of maximum domain exploration and goal-directed safe exploration in \cref{sec:comparison}. Finally, we demonstrate our core algorithm, \sempc with car dynamics in challenging environments in \cref{sec:demo_car_dynamics}. For additional details on simulation, including exact parameters and videos, we refer readers to our codebase at \href{https://github.com/manish-pra/sagempc}{https://github.com/manish-pra/sagempc}.
\vspace{-2mm}
\subsection{Numerical solution to \sempc}\label{sec:numerical}
\looseness -1 In the following, we discuss how to formulate Problem \eqref{eq:exact_prob} as a finite-dimensional discrete-time problem to solve it efficiently. Specifically, we frame the problem over a finite horizon $H\in \mathbb{N}$ and optimize over the discrete-time differences, $\Delta t_k \coloneqq t(k+1) - t(k)$, at $k = 0,1, \hdots, H-1$. Thus, Problem \eqref{eq:exact_prob} can be posed as the following non-linear program (\nlp):
\begin{subequations} \label{eqn:sagempc_nlp}
\begin{align}
    \min_{\nu, \state_{k|t_\n}, \coninput_{k|t_\n}, \Delta t_k}& \lambda \nu + \dist(\state_{\horizonmid|t_\n}, \gostateOpti[\n]) \\
    \mathrm{s.t.}~&\state_{k+1|t_\n} = \dyn_d(\state_{k|t_\n}, \coninput_{k|t_\n}, \Delta t_k)  \label{nlp:dynamics}\\
        &\state_{k|t_\n} \in \Domain, \coninput_{k|t_\n} \in \mathcal{\inputSpace}\label{nlp:stateInputConstraint}\\
    &\lbconst[\n](\state_{k|t_\n}) \geq 0, ~k = 0, 1, \hdots, H-1  \label{nlp:pessiConstraint}\\
    & \epsconst - \sumMaxwidth[]{\n}(\state_{\horizonmid|t_\n}) \leq \nu,~\nu \geq 0 \label{nlp:uncertainityConstraint}\\
    & \sum\nolimits_{k=0}^{H-1} \Delta t_k \leq T ,~ \Delta t_k \geq 0\label{nlp:timeConstraint}\\
    & \state_{0|t_\n} = \state(t_{\n}), \state_{H|t_\n} \in \safeInit{\n}. \label{nlp:InitReturnConstraint}
\end{align}
\end{subequations}
%
\looseness -1 Here, $f_d(x, a, \Delta t)$ is the integration of the continuous time dynamics \eqref{eq: dyn} over a time interval $\Delta t$. The integration is performed using the $4^{th}$ order Runge-Kutta (\textsc{\small{RK4}}) method. 
In the \nlp, \eqref{nlp:dynamics} yields a predicted state sequence, \eqref{nlp:stateInputConstraint}-\eqref{nlp:pessiConstraint} ensure that this sequence satisfies state, input, and the pessimistic safety constraints, and \eqref{nlp:InitReturnConstraint} ensures that this sequence starts at the current state $\state(t_\n)$ and ends in the safe set $\safeInit{\n}$. As per sampling strategy \eqref{eq:exact_prob}, we aim to find
the optimal time $t'\in[0,T]$ for the agent to collect a sample. For this, we fix a sampling horizon $\horizonmid = \lfloor H/2 \rfloor$ w.l.o.g. and the solution to \nlp~\eqref{eqn:sagempc_nlp} results in $\state_{\horizonmid|t_\n}$ corresponding to $\tilde{\state}(t') \in \Rcontoper[\{\state(t_{\n})\},\safeInit{\n}]{\pessiSet[]{\n}}$ in Problem \eqref{eq:exact_prob}. 
The constraint \eqref{nlp:uncertainityConstraint} ensures that the uncertainty at $\state_{\horizonmid|t_\n}$ is larger than $\epsilon$. Note that, we also optimize over horizon time while limiting the overall time to $T$ in \eqref{nlp:timeConstraint}. Due to the ``control invariance" property (\cref{assump: safe-init}) the agent can remain in $\safeInit{\n}$ for the excess time $T-t(H)$, which results in the equivalent reachable returnable set to that of fixed $T$.

\looseness -1 We use a Sequential Quadratic Programming (\SQP) method to solve the Problem \eqref{eqn:sagempc_nlp} with the Acados framework \cite{verschueren2022acados}, 
implemented in Python. We model \GP\!\!s in GPytorch \cite{gardner2018gpytorch} and provide the gradient of $\lbconst[\n](\cdot), \ubconst[\n](\cdot)$ to each \textsc{\small{QP}} iteration externally for improved efficiency as proposed in \cite{lahr2023zero}. 
For simplicity, we use $\lbconst[\n](\state) = \muconst[\n-1](\state) - \sqrt{\betaconst[\n]} \sigconst[\n-1](\state)$, $\ubconst[\n](\state) = \muconst[\n-1](\state) + \sqrt{\betaconst[\n]} \sigconst[\n-1](\state)$ and $\sqrt{\beta_\n} = 3, \forall n \geq 1$ as done in \cite{safe-bo-sui15,sui2018stagewise,berkenkamp2023bayesian,turchetta2016safemdp,turchetta2019safe,prajapat2022near}. 


\looseness -1 In our experimental setup, the \emph{a-priori} unknown constraint~\eqref{eq:def_constraint} only depends on the cartesian coordinates $[x_p,y_p]$. To create environments, we randomly sample the constraint functions from a \GP and define the constraint sets on the cartesian space where unsafe regions represent e.g., the obstacles lying in the 2D space. We use squared exponential and sufficiently smooth Mat\'ern kernels which satisfy \cref{assump:q_RKHS,assump:sublinear}. 
In the environments, we find a safe initial position $p_s = [\xpos,\ypos]^\top$ and define a prior such that $\lbconst[0](p_s) \geq 0$.
\vspace{-0.1em}
\subsection{Comparison of safe exploration techniques} 
\label{sec:comparison}
\begin{figure}
\hspace{-1.50mm}
    \begin{subfigure}[t]{0.5\columnwidth}
  	\centering
  	\includegraphics[scale=0.75]{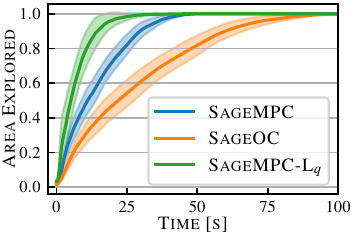}
    \caption{Fraction of area explored}
    \label{fig:maximum_domain_SE}
    \end{subfigure}
    \hspace{-3.00mm}
~
    \begin{subfigure}[t]{0.5\columnwidth}
  	\centering
  	\includegraphics[scale=0.75]{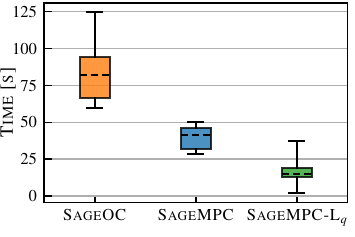}
    \caption{Exploration end time statistics}
    \label{fig:maximum_domain_SE_time}
    \end{subfigure} \vspace{-0.1em}
\caption{\looseness -1 Comparison of \seocp, \sempc and $\sempcL$ on the task of maximum safe domain exploration. (a) Shows the area explored normalized with the maximum explored area along with  95\% confidence bounds across different environments. (b) Shows box plot for the finish time of the exploration process. While all variants achieve the maximum domain exploration, techniques such as exploiting the Lipschitz bound and re-planning in receding horizon style make the algorithm more efficient.} \vspace{-1.4em}
\end{figure}
\looseness 0 In this section, we validate our theory and compare different techniques of safe guaranteed exploration framework on a task of maximum domain exploration and goal-directed safe exploration. For this, we consider a unicycle robot whose non-linear continuous-time dynamic model is,\vspace{-1mm}
\begin{equation}
\label{eqn:unicycle_dynamics}
 \!\!\!\! \dxpos = v \cos(\theta), \, \dypos = v \sin(\theta), \, \dot{v} = \alpha, \, \dot{\theta} =  \omega, \, \dot{\omega} = \psi \! \vspace{-1mm}
\end{equation}
where $[\xpos, \ypos]$ represent the position in cartesian space, while $v$ stands for velocity, $\theta$ denotes the absolute heading angle and $\omega$ represents angular velocity. The control inputs $\alpha,\psi$ correspond to linear and angular acceleration. The system is locally controllable and thus satisfies \cref{assump: equal_boundary}, as shown in \cref{apx:assum_local_control}. 
We define the terminal set $\safeInit{\n}\coloneqq\{\state_s\}$, where $\state_s$ is a safe steady state corresponding to the safe position $p_s$. 
We compare three variants \emph{i)} \seocp, safe guaranteed exploration using optimal control presented in \cref{alg:seocp}/~\ref{alg:GOsampling}, 
\emph{ii)} \sempc, presented in \cref{alg:go-se-mpc}, and \emph{iii)} $\sempcL$, which exploits Lipschitz continuity to obtain an enlarged pessimistic constraint set (\cref{rem:seMPCwithexpanders}). We utilize the true Lipschitz constant of the constraint function. 
All of the variants use the same horizon $H=80$ and $T=1$s. We compared these variants on both the tasks of maximum domain exploration and goal-directed safe exploration explained below. 

\subsubsection{Maximum safe domain exploration}
\looseness -1 The robot is tasked to explore the domain and achieve Objective~\eqref{eqn:objective} up to an $\epsilon$ tolerance. A higher $\epsilon$ naturally leads to faster exploration but might result in a smaller explored domain. To achieve the objective faster, each of the variants determines the highest uncertainty location in the reachable returnable pessimistic set. Then they sample at that location, obtaining a measurement corrupted with $\noisevar = 10^{-4}$ Gaussian noise. 

\looseness -1 We utilize 10 different environments, each having a randomly generated constraint function, and run 4 instances for each environment resulting in a total of 40 runs. \Cref{fig:maximum_domain_SE} shows the fraction of the domain explored over time. 
All the variants of \seocp achieve maximum domain exploration while being safe always. 
\Cref{fig:maximum_domain_SE_time} shows statistics regarding the finish time of the exploration process.
We observe that \sempc is more efficient as compared to \seocp, since the robot instantaneously re-plans with every new information and does not return to the terminal set. $\sempcL$ additionally takes advantage of Lipschitz continuity, thereby obtaining an enlarged pessimistic set, which accelerates the exploration process. More importantly, $\sempcL$ samples only in the expanders $\expansion[\n](\safeInit{\n}, \LpessiSet[]{\n})$, which allows for early termination of the exploration when no location in $\expansion[\n](\safeInit{\n}, \LpessiSet[]{\n})$ has uncertainty at least $\epsilon$.


\begin{figure}
\setlength{\abovecaptionskip}{8pt}
    \centering
    \includegraphics[width=0.8\columnwidth]{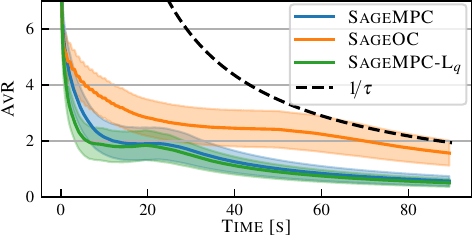}
    \caption{\looseness -1 Comparison of \seocp, \sempc and $\sempcL$ on the task of goal-directed safe exploration. The figure shows average regret (\AvR) defined in \cref{eqn:goal_objective} with 1-sigma confidence bounds across different environments.  
    While all variants minimize the loss and achieve sublinear regret, i.e., \AvR is bounded by the dashed line ($1/\tau$), techniques such as using an enlarged pessimistic set using Lipschitz bound and re-planning in receding horizon style make the algorithm more efficient.}
    \label{fig:goal_directed_SE} \vspace{-1.4em}
\end{figure}
\begin{figure*}[]
\centering
    \begin{subfigure}[t]{0.33\linewidth}
    \centering
    \includegraphics[width=1\textwidth]{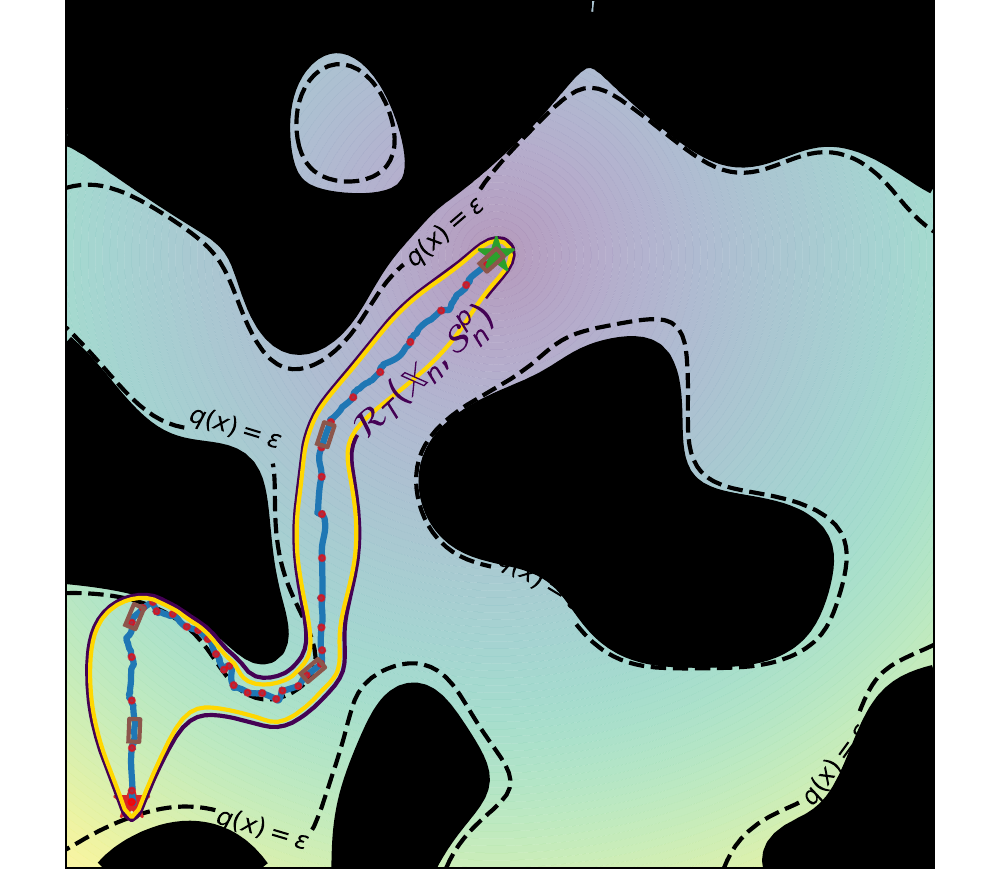}
    \caption{A cluttered environment}
    \label{fig:multi-constraints}
    \end{subfigure}
    \hspace{-3.00mm}
~
    \begin{subfigure}[t]{0.33\linewidth}
    \centering
    \includegraphics[width=1\textwidth]{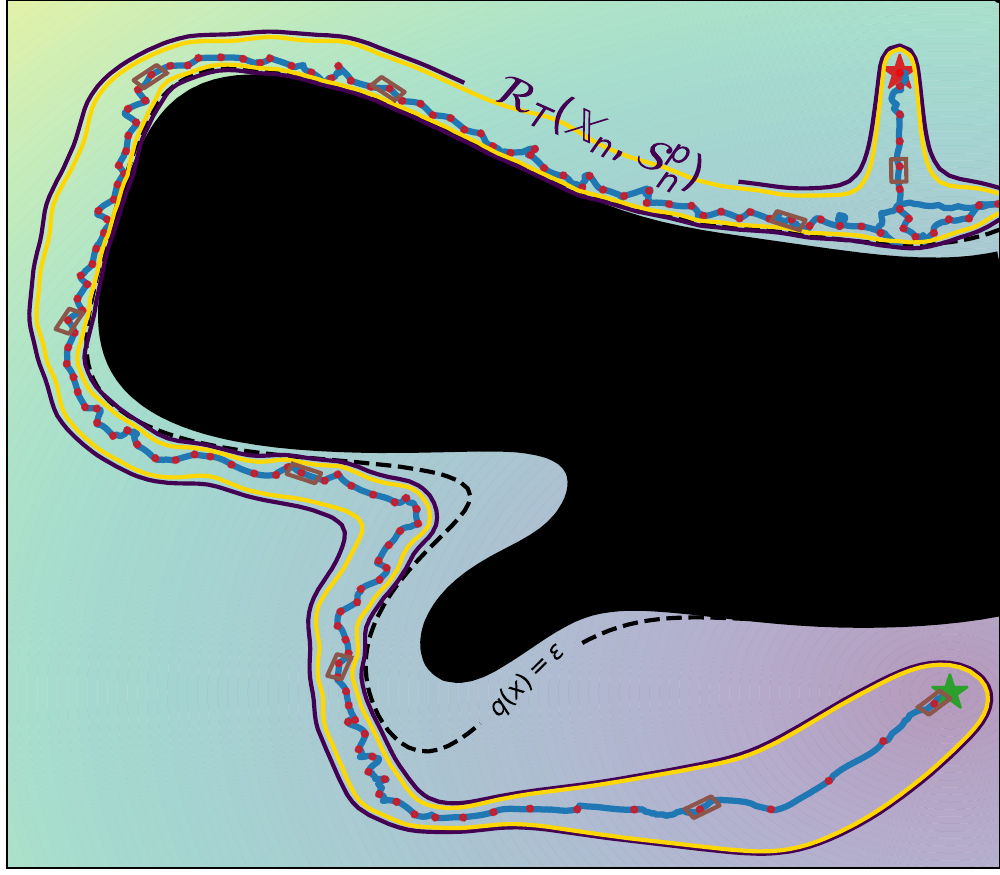}
    \caption{Environment with a large obstacle}
    \label{fig:large-constraints}
    \end{subfigure}
    \hspace{-4.00mm}
~
\hspace{-2.00mm}
    \begin{subfigure}[t]{0.33\linewidth}
    \centering
    \includegraphics[width=1\textwidth]{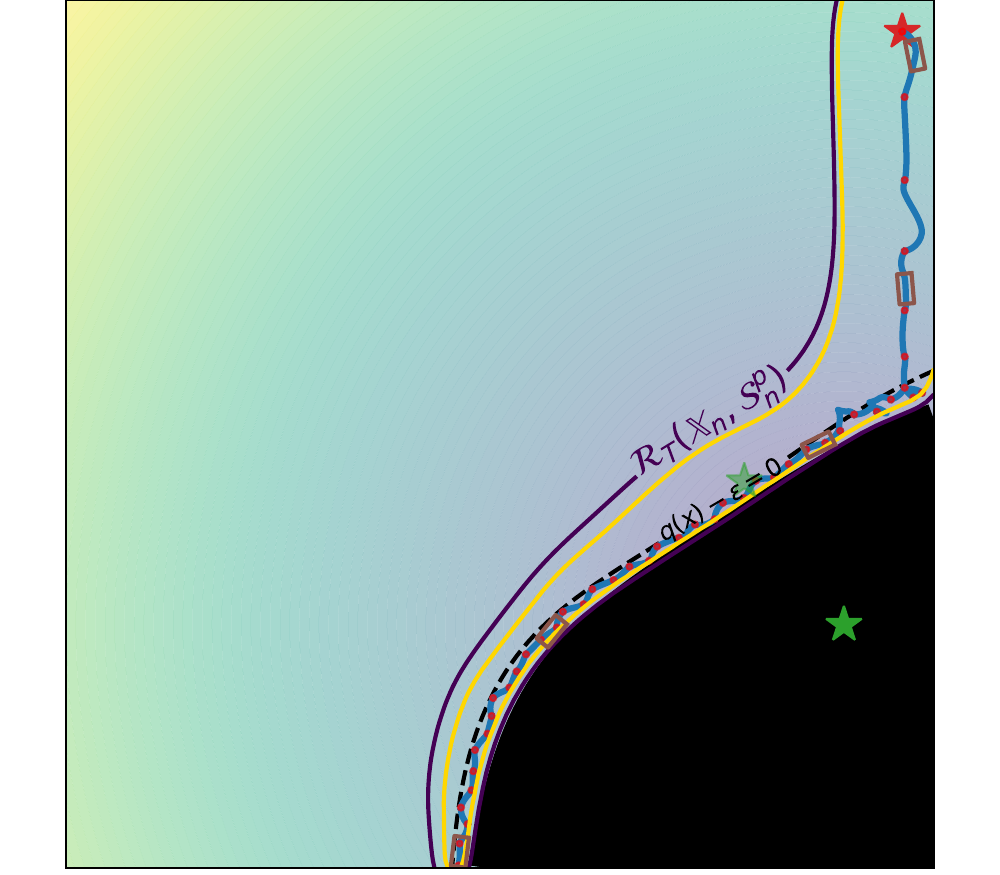}
    \caption{Environment with non-reachable initial goal}
    \label{fig:goal-change}
    \end{subfigure} \vspace{-0.3em}
    \caption[
Demonstration of the task of goal-directed safe exploration by a car in unknown, challenging environments. The car, depicted by the brown box, starts safely at the red star and needs to navigate through a priori unknown obstacles represented by the black region to minimize the loss function, shown by contours, which captures the task of reaching the goal location marked by the green star. To explore the region, it collects measurements depicted by the red points and gradually grows its reachable returnable pessimistic set, shown by the black-yellow lines. Throughout the process, the car does not violate any of the safety-critical constraints, and the resulting safe trajectory traversed by the car is depicted by the blue line.]{\looseness -1 Demonstration of the task of goal-directed safe exploration by a car in unknown, challenging environments. The car, depicted by the \brown{brown box} ({\scalebox{0.2}{\tikzset{every picture/.style={line width=0.75pt}} 

\begin{tikzpicture}[x=0.75pt,y=0.75pt,yscale=-1,xscale=1]
\draw  [color={rgb, 255:red, 140; green, 86; blue, 75 }  ,draw opacity=1 ][line width=5]  (198,29) -- (246,29) -- (246,53) -- (198,53) -- cycle ;
\end{tikzpicture}}}), starts safely at the \red{red star} ({\scalebox{0.15}{\tikzset{every picture/.style={line width=0.75pt}} 

\begin{tikzpicture}[x=0.75pt,y=0.75pt,yscale=-1,xscale=1]

\draw  [color={rgb, 255:red, 214; green, 39; blue, 40 }  ,draw opacity=1 ][fill={rgb, 255:red, 214; green, 39; blue, 40 }  ,fill opacity=1 ] (137,110) -- (145.46,131.93) -- (172.19,132.11) -- (150.68,145.85) -- (158.75,167.89) -- (137,154.44) -- (115.25,167.89) -- (123.32,145.85) -- (101.81,132.11) -- (128.54,131.93) -- cycle ;

\end{tikzpicture}}}) and needs to navigate through a-priori unknown obstacles, represented by the black region ({\scalebox{0.15}{\tikzset{every picture/.style={line width=0.75pt}} 

\begin{tikzpicture}[x=0.75pt,y=0.75pt,yscale=-1,xscale=1]

\draw  [draw opacity=0][fill={rgb, 255:red, 0; green, 0; blue, 0 }  ,fill opacity=1 ] (509,201) -- (585,201) -- (585,252) -- (509,252) -- cycle ;

\end{tikzpicture}}}), to minimize the loss function, shown by contours ({\scalebox{0.15}{\input{images/countours}}}), which captures the task of reaching the goal location marked by the \green{green star} ({\scalebox{0.15}{

\begin{tikzpicture}[x=0.75pt,y=0.75pt,yscale=-1,xscale=1]
\protect\draw  [color={rgb, 255:red, 44; green, 160; blue, 44 }  ,draw opacity=1 ][fill={rgb, 255:red, 44; green, 160; blue, 44 }  ,fill opacity=1 ] (301,110) -- (309.46,131.93) -- (336.19,132.11) -- (314.68,145.85) -- (322.75,167.89) -- (301,154.44) -- (279.25,167.89) -- (287.32,145.85) -- (265.81,132.11) -- (292.54,131.93) -- cycle ;
\end{tikzpicture}}}). To explore the region, it collects measurements depicted by the \red{red points} ({\scalebox{0.15}{\tikzset{every picture/.style={line width=0.75pt}} 

\begin{tikzpicture}[x=0.75pt,y=0.75pt,yscale=-1,xscale=1]
\path (450,180);
\draw  [color={rgb, 255:red, 214; green, 39; blue, 40 }  ,draw opacity=1 ][fill={rgb, 255:red, 214; green, 39; blue, 40 }  ,fill opacity=1 ] (434,162.5) .. controls (434,158.36) and (437.36,155) .. (441.5,155) .. controls (445.64,155) and (449,158.36) .. (449,162.5) .. controls (449,166.64) and (445.64,170) .. (441.5,170) .. controls (437.36,170) and (434,166.64) .. (434,162.5) -- cycle ;

\end{tikzpicture}}}) and gradually grows its reachable returnable pessimistic set, $\Rcontoper[\safeInit{\n}]{\pessiSet[]{\n}}$, shown by the black-\yellow{yellow} lines (\!\!{\scalebox{1}{\tikzset{every picture/.style={line width=0.75pt}} 

\begin{tikzpicture}[x=0.75pt,y=0.1pt,yscale=-0.5,xscale=0.2]
\draw [color={rgb, 255:red, 255; green, 215; blue, 0 }  ,draw opacity=1 ][line width=1]    (290,226) .. controls (312,182) and (310,265) .. (339,220) ;
\draw [color={rgb, 255:red, 44; green, 31; blue, 78 }  ,draw opacity=1 ][line width=1]    (289,256) .. controls (311,212) and (309,294) .. (338,250) ;
\end{tikzpicture}}}\!\!). Throughout the process, the car does not violate any of the safety-critical constraints, and the resulting safe trajectory traversed by the car is depicted by the \blue{blue line} (\!\!{\scalebox{1}{\tikzset{every picture/.style={line width=0.75pt}} 
\begin{tikzpicture}[x=0.75pt,y=0.1pt,yscale=0.5,xscale=0.2]
\draw [color={rgb, 255:red, 31; green, 119; blue, 180 }  ,draw opacity=1 ][line width=1]    (94,265) .. controls (116,221) and (114,303) .. (143,259) ;
\end{tikzpicture}}}\!\!).}
\vspace{-1.3em} \label{fig:challenging_env}
\end{figure*} 
\subsubsection{Goal-directed safe exploration}
\looseness -1 The robot is tasked to move to a pre-specified goal $\state_g$ by minimizing the loss $\utility(\state)\coloneqq \|\state - \state_g\|^2$. 
We consider 10 diverse environments with goal and constraint functions being randomly generated. The goal $\state_g$ may be safe or unsafe with respect to unknown constraints. If it is unsafe, \seocp will identify that it is not possible to reach it safely and go to the closest point which is in the safe reachable returnable set.

\looseness -1 
In \eqref{eq:exact_prob}, we choose the $\dist(\state,\gostateOpti[\n]) = \|\state - \gostateOpti[\n]\|$, which returns a sampling location closer to $\gostateOpti[\n]$ while ensuring uncertainty at least $\epsilon$. We run 4 instances for each of the environments to account for variations in noise, resulting in a total of 40 runs. \Cref{fig:goal_directed_SE} shows the average regret (\cref{eqn:goal_objective}) for each algorithm over time. 
All the variants of \seocp reach the desired best safe position and achieve sublinear regret. We observe that \sempc is more efficient as compared to \seocp, since it re-plans immediately with every new information. $\sempcL$ additionally utilizes an enlarged pessimistic set obtained using Lipschitz continuity, which makes the process more efficient. 
Notably, in the goal-direct task the agent generally samples near the boundary of the reachable returnable pessimistic set to reach the goal faster and hence the benefits of expanders $\expansion[\n](\safeInit{\n}, \LpessiSet[]{\n})$ are small compared to that in the maximum domain exploration task.

\vspace{-0.4em}
\subsection{\sempc with car dynamics in challenging environments }
\label{sec:demo_car_dynamics}
\label{sec:car_demo}

\looseness -1 In this section, we demonstrate \sempc, with car dynamics in challenging a-priori unknown environments. We model the car dynamics with a non-linear bicycle model,
\begin{align*}
\dxpos = v \cos(\theta + \beta), ~&\dypos= v \sin(\theta + \beta), ~\dot{\theta} = \frac{v}{l_r} \sin(\beta) \\
 \dot{v} = \alpha,~\beta &=\tan^{-1}\left(\frac{l_r}{l_f + l_r} \tan(\delta)\right) \numberthis \label{eqn:car_dynamics}
\end{align*}
which includes slip angle $\beta$, distance from the centre of gravity to the rear and the front wheel represented by $l_f=1.105$ and $l_r=1.738$ respectively. The control inputs are the front steering angle $\delta$ and the linear acceleration $\alpha$. In contrast to the unicycle model, the car cannot rotate on the spot and needs to turn with a wider radius. This makes the safe exploration interesting as the algorithm needs to plan curved paths taking dynamics into account. 


\looseness -1 We consider the task of goal-directed safe exploration on challenging instances as shown in \Cref{fig:challenging_env}. The car minimizes the loss $\utility(\state) \coloneqq \|\state - \state_g\|^2$ from a pre-specified goal $\state_g$. 
In contrast to \cref{sec:comparison}, the terminal set $\safeInit{\n}$ is defined as the set of steady states within the pessimistic set, see \cref{eq:growing_safe_set} for details, which grows the terminal set $\safeInit{\n}$ in \sempc.
The system \eqref{eqn:car_dynamics} is locally controllable and thus the set $\safeInit{\n}$ \eqref{eq:growing_safe_set} satisfies \cref{assump: safe-init}, as shown in \cref{apx:assum_local_control}. Note that the entire cartesian space is a steady state with $v=0$, which simplifies the computation of a safe retunable path since it can be just a trajectory slowing down the car to stop in the pessimistically safe set instead of a path back to the starting location. This allows us to use \sempc with a relatively small horizon of $H=25$. On an Intel i7-11800H @ 2.30GHz processor, one \SQP-iteration takes on average $9.82 \pm 1.16$ ms. We restrict the maximum number of \SQP iterations to $20$ and the maximal solve time for \nlp \eqref{eqn:sagempc_nlp} over all runs is $258.7$ ms which occurs in the environment with large obstacle (\cref{fig:large-constraints}) due to increasing number of data samples with $\n^\prime=146$.




\looseness -1 In \Cref{fig:multi-constraints}, the car navigates an unknown environment cluttered with multiple objects and achieves the task of reaching a location marked with a green star. In \Cref{fig:large-constraints}, the car needs to explore almost the complete domain to overcome the large obstacle and reach the goal on the other side. Finally in \Cref{fig:goal-change}, the initial goal location is unsafe and the car identifies that it cannot be reached and converges to the goal location in the safely reachable returnable set (faded green star) without exploring the whole domain. 
In all of the above examples, we observe that \sempc plans curvy paths to efficiently reach informative locations while respecting the non-holonomic car dynamics.
Moreover, throughout the exploration process, the car does not violate the safety-critical constraint and achieves the goal of minimizing the loss. 

\section{Conclusion}
\looseness -1 We propose a novel framework for guaranteeing exploration in finite time while being provably safe at all times for non-linear systems. 
One of the remarkable aspects of the framework is its broad applicability, as it requires only mild assumptions and can be readily employed in various complex real-world scenarios with non-linear dynamics and a-priori unknown domains.
Utilizing this framework, we propose an efficient algorithm, \sempc, by exploiting known Lipschitz bound, incorporating a goal-directed approach for safe exploration, and introducing re-planning in a receding horizon fashion with every new information update. These techniques enhance sample efficiency while maintaining the desired guarantees of the framework, i.e., safe guaranteed exploration in finite time for non-linear systems. We substantiate the practicality of our contributions by numerical experiments on challenging unknown domains with non-linear car dynamics. Our work opens up many possibilities for safe exploration in continuous domains with theoretical guarantees for a wide range of use cases. As future work addressing additional uncertainty in dynamics \cite{prajapat2024towards} and deploying it in real-world exploration experiments would be interesting.

\bibliographystyle{IEEEtran} 
\bibliography{ref}

\appendix
\looseness -1 In the appendix, we discuss locally controllable systems and provide sufficient conditions for \cref{assump: safe-init,assump: equal_boundary} in \cref{apx:assum_local_control}. Later we have auxiliary lemmas and proof for \cref{thm:SE_exp} in \cref{apx:expander_proof} and lastly, for completeness, we included a section on mutual information in \cref{apx:mutual_info}.
\vspace{-1em}
\subsection{Locally controllable systems}
\label{apx:assum_local_control}
\looseness -1 In this section, we introduce locally controllable systems and the path-connected sets. Utilizing these, we prove a sufficient condition for \cref{assump: safe-init} in \cref{lem:pessi_increase_Xn} and show that locally controllable systems satisfy \cref{assump: equal_boundary} in \cref{lem: local_control}. 


\begin{definition}[Local controllability, Definition 3.7.4, \cite{sontag2013mathematical}] \looseness -1 Let $\state(t)$ and $\coninput(t)$ be state and input trajectory, of a dynamical system \eqref{eq: dyn} on an interval $t \in [0,\delta T]$, such that,
\begin{align*}
    \state_0 \coloneqq \state(0), \qquad    \state_1 \coloneqq \state(\delta T) 
\end{align*}
\looseness -1 The system \eqref{eq: dyn} is \textbf{locally controllable along} $\state(t)$ if for each $\epsilon>0$ there is some $\delta>0$ such that the following property holds: For each $\tilde{\state}_0, \tilde{\state}_1 \in \R^\sdim$ with $\|\tilde{\state}_0-\state_0\| < \delta$ and $ \|\tilde{\state}_1-\state_1\| < \delta$ there is some state and input trajectories $\tilde{\state}(t)$, $\tilde{\coninput}(t), t\in [0, \delta T]$ of dynamical system $\dyn$ such that, 
\begin{align*}
\tilde{\state}(0) = \tilde{\state}_0,  \quad   &\tilde{\state}(\delta T) = \tilde{\state}_1,\quad\textit{and}~\\
\|\tilde{\state}(t) - \state(t)\| <\epsilon, ~~\|\tilde{\coninput}(t) &- \coninput(t)\| <\epsilon, ~~\forall t \in [0, \delta T].
\end{align*}\label{def: local_control}\vspace{-1em}
\end{definition} \vspace{-1em}
\begin{definition}[Path connected with finite length] \looseness -1  A topological space $\X$ is a path connected with finite length if for any two points $\state_0,\state_1 \in \X$ there is a continuously differentiable map $\zeta:[0,1] \to \X$ 
 s.t. $\zeta(0) = \state_0$, $\zeta(1) = \state_1$ and $\frac{d \zeta}{d s}\leq L < \infty$.
\label{def: path_connected}
\end{definition}
\looseness -1  $\zeta$ is a continuous function for a path connected set; however, we additionally assume continuously differentiable $\zeta$ with bounded derivative, which implies finite length.

 \looseness -1 The next lemma proves a sufficient condition for \cref{assump: safe-init}. 
Consider the terminal set $\safeInit{\n}$ after $\n$ samples defined as:
\begin{equation}
\safeInit{\n} \coloneqq \{x \in \X^\epsilon | \exists a \in \A^{\epsilon}\!: \dyn(x, a) = 0, x \in \pessiSet[,\epsconst]{\n} \},\label{eq:growing_safe_set}
\end{equation}
for a fixed $\epsilon>0$ with compact set $\X^{\epsilon} \subseteq \X \ominus B_\epsilon(0)$ and compact set $\A^{\epsilon} \subseteq \A \ominus B_\epsilon(0)$. 
\begin{figure}
    \centering
    \setlength{\abovecaptionskip}{3pt}
    \scalebox{0.45}{\tikzset{every picture/.style={line width=0.75pt}} 

\begin{tikzpicture}[x=0.75pt,y=0.75pt,yscale=-1,xscale=1]

\draw  [fill={rgb, 255:red, 128; green, 128; blue, 128 }  ,fill opacity=0.3 ] (81,142.2) .. controls (81,112.82) and (104.82,89) .. (134.2,89) -- (570.8,89) .. controls (600.18,89) and (624,112.82) .. (624,142.2) -- (624,301.8) .. controls (624,331.18) and (600.18,355) .. (570.8,355) -- (134.2,355) .. controls (104.82,355) and (81,331.18) .. (81,301.8) -- cycle ;
\draw  [fill={rgb, 255:red, 177; green, 177; blue, 255 }  ,fill opacity=1 ] (220,100) .. controls (287,97) and (481.33,88.5) .. (537,131) .. controls (592.67,173.5) and (571.67,304.83) .. (550,327.33) .. controls (528.33,349.83) and (351,348.6) .. (302.6,346.2) .. controls (254.2,343.8) and (147,354.2) .. (121,332.2) .. controls (95,310.2) and (93,193) .. (112,152) .. controls (131,111) and (153,103) .. (220,100) -- cycle ;
\draw  [fill={rgb, 255:red, 245; green, 166; blue, 35 }  ,fill opacity=0.72 ] (149,122) .. controls (185,96) and (363,104) .. (418,115) .. controls (473,126) and (525,175) .. (509,227) .. controls (493,279) and (456,305) .. (419,315) .. controls (382,325) and (151.8,363.8) .. (123.4,310.6) .. controls (95,257.4) and (113,148) .. (149,122) -- cycle ;
\draw  [fill={rgb, 255:red, 108; green, 215; blue, 108 }  ,fill opacity=0.6 ] (149,122) .. controls (174,100) and (347,100) .. (412,114) .. controls (477,128) and (518.2,168) .. (494.6,217) .. controls (471,266) and (439.4,270.2) .. (423,275) .. controls (406.6,279.8) and (215,287) .. (172,260) .. controls (129,233) and (124,144) .. (149,122) -- cycle ;
\draw [color={rgb, 255:red, 80; green, 227; blue, 194 }  ,draw opacity=0.9 ][line width=6]    (168.67,236) .. controls (178.67,206.67) and (211.88,176.93) .. (267.1,208.8) .. controls (322.33,240.67) and (390.67,251.27) .. (378,225.8) .. controls (365.33,200.33) and (426.33,183.67) .. (467,153) ;
\draw  [line width=2.25]  (268.2,208.8) .. controls (268.2,207.97) and (267.71,207.3) .. (267.1,207.3) .. controls (266.5,207.3) and (266.01,207.97) .. (266.01,208.8) .. controls (266.01,209.63) and (266.5,210.3) .. (267.1,210.3) .. controls (267.71,210.3) and (268.2,209.63) .. (268.2,208.8) -- cycle ;
\draw [line width=1.5]  [dash pattern={on 5.63pt off 4.5pt}]  (267.1,210.3) .. controls (287.49,226.02) and (384.02,258.9) .. (377,226.8) .. controls (370.16,195.5) and (399.47,198.24) .. (436.63,172.1) ;
\draw [shift={(439.5,170.03)}, rotate = 143.61] [fill={rgb, 255:red, 0; green, 0; blue, 0 }  ][line width=0.08]  [draw opacity=0] (13.4,-6.43) -- (0,0) -- (13.4,6.44) -- (8.9,0) -- cycle    ;
\draw [shift={(331.55,236.3)}, rotate = 194.49] [fill={rgb, 255:red, 0; green, 0; blue, 0 }  ][line width=0.08]  [draw opacity=0] (13.4,-6.43) -- (0,0) -- (13.4,6.44) -- (8.9,0) -- cycle    ;
\draw [shift={(404.3,190.38)}, rotate = 154.16] [fill={rgb, 255:red, 0; green, 0; blue, 0 }  ][line width=0.08]  [draw opacity=0] (13.4,-6.43) -- (0,0) -- (13.4,6.44) -- (8.9,0) -- cycle    ;
\draw  [line width=3]  (438.6,170.03) .. controls (438.6,169.53) and (439,169.13) .. (439.5,169.13) .. controls (440,169.13) and (440.4,169.53) .. (440.4,170.03) .. controls (440.4,170.53) and (440,170.93) .. (439.5,170.93) .. controls (439,170.93) and (438.6,170.53) .. (438.6,170.03) -- cycle ;
\draw   (218.72,208.8) .. controls (218.72,182.47) and (240.38,161.12) .. (267.1,161.12) .. controls (293.83,161.12) and (315.49,182.47) .. (315.49,208.8) .. controls (315.49,235.13) and (293.83,256.48) .. (267.1,256.48) .. controls (240.38,256.48) and (218.72,235.13) .. (218.72,208.8) -- cycle ;
\draw   (230.38,208.8) .. controls (230.38,189.49) and (246.82,173.83) .. (267.1,173.83) .. controls (287.38,173.83) and (303.83,189.49) .. (303.83,208.8) .. controls (303.83,228.11) and (287.38,243.77) .. (267.1,243.77) .. controls (246.82,243.77) and (230.38,228.11) .. (230.38,208.8) -- cycle ;
\draw    (267.1,208.8) -- (267.1,165.24) ;
\draw [shift={(267.1,162.24)}, rotate = 90] [fill={rgb, 255:red, 0; green, 0; blue, 0 }  ][line width=0.08]  [draw opacity=0] (8.93,-4.29) -- (0,0) -- (8.93,4.29) -- cycle    ;
\draw    (267.1,208.8) -- (279.06,179.58) ;
\draw [shift={(280.2,176.8)}, rotate = 112.26] [fill={rgb, 255:red, 0; green, 0; blue, 0 }  ][line width=0.08]  [draw opacity=0] (6.25,-3) -- (0,0) -- (6.25,3) -- cycle    ;
\draw  [line width=3]  (298,226.43) .. controls (298,225.93) and (298.4,225.53) .. (298.9,225.53) .. controls (299.4,225.53) and (299.8,225.93) .. (299.8,226.43) .. controls (299.8,226.93) and (299.4,227.33) .. (298.9,227.33) .. controls (298.4,227.33) and (298,226.93) .. (298,226.43) -- cycle ;
\draw [line width=1.5]    (268.2,208.8) .. controls (257.28,258.72) and (266.51,268.33) .. (295.71,229.51) ;
\draw [shift={(298,226.43)}, rotate = 126.19] [fill={rgb, 255:red, 0; green, 0; blue, 0 }  ][line width=0.08]  [draw opacity=0] (11.61,-5.58) -- (0,0) -- (11.61,5.58) -- cycle    ;
\draw  [line width=3]  (412.6,186.03) .. controls (412.6,185.53) and (413,185.13) .. (413.5,185.13) .. controls (414,185.13) and (414.4,185.53) .. (414.4,186.03) .. controls (414.4,186.53) and (414,186.93) .. (413.5,186.93) .. controls (413,186.93) and (412.6,186.53) .. (412.6,186.03) -- cycle ;
\draw  [line width=3]  (344.6,239.03) .. controls (344.6,238.53) and (345,238.13) .. (345.5,238.13) .. controls (346,238.13) and (346.4,238.53) .. (346.4,239.03) .. controls (346.4,239.53) and (346,239.93) .. (345.5,239.93) .. controls (345,239.93) and (344.6,239.53) .. (344.6,239.03) -- cycle ;
\draw  [line width=3]  (381.6,203.03) .. controls (381.6,202.53) and (382,202.13) .. (382.5,202.13) .. controls (383,202.13) and (383.4,202.53) .. (383.4,203.03) .. controls (383.4,203.53) and (383,203.93) .. (382.5,203.93) .. controls (382,203.93) and (381.6,203.53) .. (381.6,203.03) -- cycle ;
\draw  [line width=3]  (375.6,229.03) .. controls (375.6,228.53) and (376,228.13) .. (376.5,228.13) .. controls (377,228.13) and (377.4,228.53) .. (377.4,229.03) .. controls (377.4,229.53) and (377,229.93) .. (376.5,229.93) .. controls (376,229.93) and (375.6,229.53) .. (375.6,229.03) -- cycle ;

\draw (577.73,120) node [anchor=north west][inner sep=0.75pt]  [font=\Huge]  {$\X$};
\draw (298,119) node [anchor=north west][inner sep=0.75pt]  [font=\LARGE]  {$\Rcontoper[\safeInit{\n}]{\pessiSet[]{\n}}$};
\draw (246.53,205) node [anchor=north west][inner sep=0.75pt]  [font=\LARGE]  {$u$};
\draw (445.27,170.07) node [anchor=north west][inner sep=0.75pt]  [font=\LARGE]  {$v$};
\draw (178,200) node [anchor=north west][inner sep=0.75pt]  [font=\LARGE]  {$\safeInit{\n}$};
\draw (379.2,223.8) node [anchor=north west][inner sep=0.75pt]  [font=\LARGE]  {$\zeta $};
\draw (256,138.4) node [anchor=north west][inner sep=0.75pt]  [font=\LARGE]  {$\bar{\epsilon}$};
\draw (274.6,179.73) node [anchor=north west][inner sep=0.75pt]  [font=\LARGE]  {$\delta $};
\draw (267.1,256.48) node [anchor=north west][inner sep=0.75pt]  [font=\LARGE]  {$\Tilde{x}(t)$};
\draw (282,199) node [anchor=north west][inner sep=0.75pt]  [font=\LARGE]  {$u'$};
\draw (334.31,210) node [anchor=north west][inner sep=0.75pt]  [font=\LARGE]  {$u''$};
\draw (450,300) node [anchor=north west][inner sep=0.75pt]  [font=\LARGE]  {$\Rcontoper[\safeInit{\n}]{\X}$};
\draw (140.6,284) node [anchor=north west][inner sep=0.75pt]  [font=\LARGE]{$\Rcontoper[\safeInit{\n}]{\optiSet[]{\n}}$};

\end{tikzpicture}}
    \caption{\looseness -1 The set $\safeInit{\n}$ (cyan line) is a set of steady states in $\pessiSet[,\epsilon]{\n}$. Consider a path $\zeta$ connecting $u \to v$. \cref{lem:pessi_increase_Xn} shows that for any $u, v \in \safeInit{\n}$, we can reach in finite time for locally controllable $\dyn$.}\vspace{-1.26em}
    \label{fig: conti-safe-init-assump}
\end{figure}
\begin{lemma}\label{lem:pessi_increase_Xn} Suppose the system $\dyn$ is locally controllable (\cref{def: local_control}) and the terminal set $\safeInit{\n}$ \eqref{eq:growing_safe_set} is path-connected with finite length (\cref{def: path_connected}), then $\safeInit{\n}$ satisfies \cref{assump: safe-init}.
\end{lemma}
\begin{proof} In the following, we show that $\safeInit{\n}$ satisfies each of the four properties of \cref{assump: safe-init}.

\emph{1. Pessimistically safe:} By definition, $\safeInit{\n} \subseteq \pessiSet[,\epsilon]{\n} \subset \pessiSet[]{\n} $, since $\epsilon>0$. Hence $\safeInit{\n}$ is pessimistically safe.

\emph{2. Monotonicity:}  By def., $\pessiSet[,\epsconst]{\n} \subseteq \pessiSet[,\epsconst]{\n+1}$ implies $\safeInit{\n} \subseteq \safeInit{\n+1}$, since $\dyn$ is independent is of $\n$. Hence $\safeInit{\n}$ is non-decreasing. 

\emph{3. Control invariance:} By def., $\forall \state \in \safeInit{\n} \subseteq \X^\epsconst, \exists a \in \A^\epsilon \subset \A: f(x,a) = 0$. Hence the $\safeInit{\n}$ satisfies control invariance.

\emph{4. Controllable around terminal set}: We need to show
$\forall u, v \in \safeInit{\n}: \exists \coninput(\cdot),  \dot{\state}(t) = \dyn(\state(t),\coninput(t)), \state(0) = u, \state(\timeX) = v, (\state(t),\coninput(\cdot))\in\pessiSet[]{\n} \times \pwcinput, t \in [0, \timeX)$.
        
        Given $\safeInit{\n}$ is path-connected with finite length (\cref{def: path_connected}), for any $u, v \in \safeInit{\n}$ there is a continuously differentiable path $\zeta:[0,1] \to \safeInit{\n}$ such that $\zeta(0) = u$, $\zeta(1) = v$ and $\frac{d \zeta}{d s}\leq L < \infty$. Hence the path length is $\int_{0}^1\frac{d \zeta}{d s} ds \leq \int_{0}^1 L ds = L < \infty$.
        


        Define the largest constant $\bar{\epsilon}>0$ such that $\A^\epsilon \oplus B_{\bar{\epsilon}}\subseteq \A, \X^\epsilon\oplus B_{\bar{\epsilon}}\subseteq \X$; $\pessiSet[,\epsconst]{\n} \oplus B_{\bar{\epsilon}} \subseteq \pessiSet[]{\n}$. This $\bar{\epsilon}$ exists since $\A^\epsilon, \X^\epsilon$ lie in the interior of $\A, \X$ and $\pessiSet[,\epsconst]{} \subset \pessiSet[]{}$ with $\epsilon>0$ since $\lbconst[\n]$ is uniformly continuous. Next, we construct a way to go from any $u\in\safeInit{\n}$ to any $v\in\safeInit{\n}$ while satisfying the constraints.
        
        Given the $\bar{\epsilon}$, let $\delta$ be defined as per local controllability of $f$. We split $\zeta$ with $s_k = k/\lceil L/\delta \rceil, k=1, \hdots ,\lceil L/\delta - 1 \rceil$ into $\lceil L/\delta \rceil$ parts with distance lower than $\delta$ apart. For reference   marked as $u', u'', u''', \hdots $ in \cref{fig: conti-safe-init-assump}.
        Since $u$ is a steady state, there exists a state and input trajectory $\state(t)\in \safeInit{\n}$ and $a(t) \in \A^\epsilon$  $\forall t \in [0,\delta T]$, such that  $\state(t)=u, t\in [0, \delta T]$.
        Since $\dyn$ is locally controllable, for given $\bar{\epsilon} , \exists \delta>0: \exists~\text{trajectories}~\tilde{\coninput}(t), \tilde{\state}(t), \forall t \in [0,\delta T]: \tilde{\state}(0) = u, \tilde{\state}(\delta T) = u', \|\tilde{\state}(t)-\state(t)\| < \bar{\epsilon}$ and $\|\tilde{\coninput}(t) -\coninput(t)\| < \bar{\epsilon}, \forall t \in [0, \delta T]$. 
        
         Since $\forall t \in [0, \delta T], \state(t) \in \pessiSet[,\epsilon]{\n}, \state(t) \in \X^{\epsconst}$ and $\|\tilde{\state}(t)-\state(t)\| < \bar{\epsilon}$, implies $\tilde{\state}(t) \in \pessiSet[]{\n}$. 
         Similarly $\forall t \in [0, \delta T], \coninput(t) \in \A^\epsilon$ and $\|\tilde{\coninput}(t)-\coninput(t)\| < \bar{\epsilon}$, implies $\tilde{\coninput}(t) \in \A$. We can recursively perform the above steps and jump from point to point, i.e., $u' \to u'' \to u'''$ 
         and so on, until we reach $v$.
        Combining all the state and input trajectories, $\tilde{\state}(\cdot),\tilde{\coninput}(\cdot)$, we can reach from any $u\in\safeInit{\n}$ to any $v\in \safeInit{\n}$ while staying the pessimistic constraint set, $\tilde{\state}(t) \subseteq \pessiSet[]{\n}$ and satisfying the input constraints $\tilde{\coninput}(t) \subseteq \A$, $\forall t \in [0, \timeX]$, where $\timeX$ is the time taken to reach  $v$. Next we show that $\timeX$ is finite.
        
        \emph{Finite time:} 
        The time taken to travel through any $\delta$ length, e.g., $(u \to u')$ is $\delta T$.
        Since the length of $\zeta \leq L$, we get an upper bound on time taken to travel between any two points in $\safeInit{\n}$ as $\timeX \leq \lceil L/\delta \rceil \delta T$. Hence any two points in $\safeInit{\n}$ can be travelled in finite time.
        \end{proof}

\begin{lemma} \looseness -1 Suppose the system $\dyn$ is locally controllable (\cref{def: local_control}) and $\A = \R^m$. Then \cref{assump: equal_boundary} holds, i.e, for any $\safeInit{\n},\pessiSet[]{\n}$, it holds, $\partial \Rcontoper[\safeInit{\n}]{\pessiSet[]{\n}} \subseteq \partial \pessiSet[]{\n} \bigcup \partial \Rcontoper[\safeInit{\n}]{\X}$. \label{lem: local_control}
\end{lemma}
\begin{proof} We will prove it by contradiction.
Suppose there exists $x' \in \partial \Rcontoper[\safeInit{\n}]{\pessiSet[]{\n}} \backslash \left(\partial \pessiSet[]{\n} \bigcup \partial \Rcontoper[\safeInit{\n}]{\X}\right)$.  
\begin{figure}
    \centering
    \setlength{\abovecaptionskip}{5pt}
    \scalebox{0.5}{\input{images/conti-local-controllability}}
    \caption{Shows a contradictory case for \cref{lem: local_control}, that $x' \in \partial \Rcontoper[\safeInit{\n}]{\pessiSet[]{\n}}$ }
    \label{fig: local_control_contradict} \vspace{-1.5em}
\end{figure}
\begin{itemize}
    \item Using $x'\in \Rcontoper[\safeInit{\n}]{\pessiSet[]{\n}} \subseteq \pessiSet[]{\n}$ but $x'\notin\partial \pessiSet[]{\n}$ and $x'\in \Rcontoper[\safeInit{\n}]{\pessiSet[]{\n}} \subseteq \Rcontoper[\safeInit{\n}]{\X}$ but $x'\notin\partial \Rcontoper[\safeInit{\n}]{\X}$, there exists a sufficiently small $\varphi>0$ such that $B_{\varphi}(x')\subseteq\pessiSet[]{\n}$ and $B_{\varphi}(x') \subseteq \Rcontoper[\safeInit{\n}]{\X}$. (Please see \cref{fig: local_control_contradict})
    
    
    \item Since $x' \in \Rcontoper[\safeInit{\n}]{\pessiSet[]{\n}}, \exists$ trajectory $\state(t), \coninput(t), t\in [0,T]$ of dynamics \eqref{eq: dyn}, such that $\state(t) \in \Rcontoper[\safeInit{\n}]{\pessiSet[]{\n}}, \forall t\in[0,T]$ and $\state(0) \in \safeInit{\n}, \state(t')=x', \state(T)\in \safeInit{\n}$. 
    
    Given $\state(\cdot)$ is continuous, there exists sufficiently small $\delta t>0, \epsilon>0$ such that $B_\epsilon(\state(t)) \subseteq B_\varphi(\state(t')), \forall t\in[t'-\delta t,t'+\delta t]$.
    %
    
    
    \item 
    \looseness -1 Given local controllablility (\cref{def: local_control}) around trajectory $\state(t), \coninput(t), t\in [t'-\delta t, t']$, we get, for $\epsilon$ defined above and some $\delta\leq\epsilon$, $\exists$ trajectories $\tilde{\state}(t), \tilde{\coninput}(t), t\in [t'-\delta t, t']: \tilde{\state}(t'-\delta t) = \state(t'-\delta t) , \tilde{\state}(t') = p$, such that $\|p - \state(t') \| <\delta$ and $\forall t\in[t-\delta t, t'],~ \|\tilde{\state}(t) - \state(t)\|<\epsilon, \|\tilde{\coninput}(t) - \coninput(t)\|<\epsilon$. 
    Since the trajectory $\tilde{\state}(\cdot)$ is $\epsilon$ close to $\state(\cdot)$, we get, $ \tilde{\state}(t) \in B_{\varphi}(\state') \subseteq \pessiSet[]{\n}, \forall t \in [t'-\delta t, t']$. Since $\A = \R^m$, $\tilde{\coninput}(t) \in \A$.

    \item Similarly, exploiting the local controllability, a return path can be constructed from $p \to \state(t+\delta t)$. From $\state(t+\delta t)$, it is trivial to follow the same path as $\state(\cdot)$.
    \item Hence, $\forall p \in B_{\delta}(\state')$, we can reach to $p$ and return back to $\safeInit{\n}$ in time $T$. This implies $B_{\delta}(\state') \subset \Rcontoper[\safeInit{\n}]{\pessiSet[]{\n}}$, hence $x' \notin \partial \Rcontoper[\safeInit{\n}]{\pessiSet[]{\n}}$, which  is a contradiction. Hence, $\partial \Rcontoper[\safeInit{\n}]{\pessiSet[]{\n}} \subseteq \partial \pessiSet[]{\n} \bigcup \partial \Rcontoper[\safeInit{\n}]{\X}$. \vspace{-0.5em}
\end{itemize}\vspace{-1em}
\end{proof}
Note that $B_{\varphi}(x') \subseteq \pessiSet[]{\n} \bigcap \Rcontoper[\safeInit{\n}]{\X} \centernot\implies B_{\varphi}(x') \subseteq \Rcontoper[\safeInit{\n}]{\pessiSet[]{\n}}$. 
Since the reachability returnability to the ball $(i.e., B_{\varphi}(x') \subseteq \Rcontoper[\safeInit{\n}]{\X})$ can be due to the paths from the regions which are not pessimistic safe. Hence, we specifically require local controllability to prove $B_{\varphi}(x') \subseteq \Rcontoper[\safeInit{\n}]{\pessiSet[]{\n}}$.

\looseness -1 The system may be locally controllable, but exploring even the connected safe region may not be possible due to state constraints. 
The term $\bigcup \Rcontoper[\safeInit{\n}]{\X}$ in \cref{assump: equal_boundary} allows us to consider the systems with state constraints as well. This guarantees the maximum domain exploration in cases where a large region is pessimistically safe; however, the agent cannot go beyond the reachable returnable set due to dynamic constraints.

\looseness -1 \cref{lem: local_control} assumes no input constraints. To consider both state and input constraints, we can define constraint sets with $(\state, \coninput)$ pair instead of only $\state$. Alternatively, wlog we can define a new dynamics $\tilde{f}$ with state $z = [\state^\top, \coninput^\top]^\top \in \R^{p+m}$ and consider the constraints for the state $z$. One way to define $\tilde{f}$ is, $\dot{z} = \tilde{\dyn}(z, \tilde{\coninput}),$  where $ \tilde{\dyn}(z, \tilde{\coninput}) \coloneqq [\dyn(\state,\coninput)^\top, \tilde{\coninput}^\top]^\top$ and $\tilde{\coninput} = \dot{\coninput} \in \R^m $.
\vspace{-1em}
\subsection{Auxiliary lemmas for proof of \texorpdfstring{\cref{thm:SE}}{Lg}} 
\label{apx:expander_proof}

\begin{lemma}\label{lem:setRSX}
    For any sets $\mathcal{R}, \mathcal{S}, \mathcal{X} \subseteq \R^\sdim$,  
    \begin{align*}
        \mathcal{R} \subseteq \mathcal{S} \subseteq \mathcal{X} \implies \partial \mathcal{R} \cap \partial \mathcal{X} \subseteq \partial \mathcal{S}.
    \end{align*}
\end{lemma}
\begin{proof}
    For any point $p \in \partial \X$, the lemma follows by showing $p \in \partial \mathcal{R} \implies p\in \partial \mathcal{S}$, under the assumption $\mathcal{R} \subseteq \mathcal{S} \subseteq \X$.

    Proof by contradiction. Let's assume $p \notin \partial \mathcal{S}$. Since $p \in \partial \mathcal{R}$ and $\mathcal{R} \subseteq \mathcal{S}$ implies $p \in \mathcal{S}$. Since $p \notin \partial \mathcal{S}$ there exists a sufficiently small $\epsilon>0$ such that $B_{\epsilon}(p) \subseteq \mathcal{S}$. Since $\mathcal{S}\subseteq\mathcal{X}$ implies $B_{\epsilon}(p) \subseteq \mathcal{X}$. This is a contradiction since $p \in \partial \X$.
\end{proof}


\begin{lemma}\label{lem:boundary_inclusion}
$\closure{\partial \Rcontoper[\safeInit{\n}]{\pessiSet[]{\n}} \backslash \partial \Rcontoper[\safeInit{\n}]{\X}} \subseteq \closure{\partial \pessiSet[]{\n} \backslash \partial \X}$, under \cref{assump: equal_boundary}. 
\end{lemma}
\begin{proof} 
\looseness -1 With \cref{assump: equal_boundary} on subtracting $\partial \X,\partial \Rcontoper[\safeInit{\n}]{\X}$,
\begin{multline}
    \left(\partial \Rcontoper[\safeInit{\n}]{\pessiSet[]{\n}} \backslash \partial \X \right) \backslash \partial \Rcontoper[\safeInit{\n}]{\X} \\
     \subseteq \left(\left(\partial \pessiSet[]{\n} \bigcup \partial \Rcontoper[\safeInit{\n}]{\X}\right) \backslash \partial \X \right) \backslash \partial \Rcontoper[\safeInit{\n}]{\X} \label{eq:boundary_inclusion}
\end{multline}
Since $\Rcontoper[\safeInit{\n}]{\pessiSet[]{\n}} \subseteq \Rcontoper[\safeInit{\n}]{\X} \subseteq \X$, by \cref{lem:setRSX}, $\partial \Rcontoper[\safeInit{\n}]{\pessiSet[]{\n}} \bigcap \partial \X \subseteq \partial \Rcontoper[\safeInit{\n}]{\X}$. Using this,  left hand side (LHS) of \cref{eq:boundary_inclusion} can be simplified as,
\begin{multline}
    \left(\partial \Rcontoper[\safeInit{\n}]{\pessiSet[]{\n}} \backslash \partial \X  \right) \backslash \partial \Rcontoper[\safeInit{\n}]{\X} \\
     = \partial \Rcontoper[\safeInit{\n}]{\pessiSet[]{\n}} \backslash \partial \Rcontoper[\safeInit{\n}]{\X} \label{eq:lhs}
\end{multline}
Consider right hand side (RHS) of \cref{eq:boundary_inclusion},
\begin{multline}
            \left(\left(\partial \pessiSet[]{\n} \bigcup \partial \Rcontoper[\safeInit{\n}]{\X}\right) \backslash \partial \X \right) \backslash \partial \Rcontoper[\safeInit{\n}]{\X} \\ 
            = \left(\partial \pessiSet[]{\n} \backslash \partial \X \right) \backslash \partial \Rcontoper[\safeInit{\n}]{\X} \subseteq \partial \pessiSet[]{\n} \backslash \partial \X \label{eq:rhs}
\end{multline}
Substituting \cref{eq:lhs} and \cref{eq:rhs} in \cref{eq:boundary_inclusion} and using closure operator, gives, 
$\closure{\partial \Rcontoper[\safeInit{\n}]{\pessiSet[]{\n}} \backslash \partial \Rcontoper[\safeInit{\n}]{\X}} \subseteq \closure{\partial \pessiSet[]{\n} \backslash \partial \X}$.\!\!\!\!
\end{proof}

\begin{lemma} \label{lem: boundaryDiff} Let the regular close sets $\calP = \{\mathcal{R}, \mathcal{S} \} \subseteq \R^\sdim$ additionally satisfy $\forall p \in \partial \calP,$ there exists an $\epsilon>0$ the set $B_{\epsilon}(p) \cap \interior{\calP}$ is path connected. 
Assume $\mathcal{R} \subseteq \mathcal{S} $ and
$\exists$ a continuous path $\zeta: [0,1] \to \mathcal{S}$ such that $\zeta(0) = \state_0 \in \mathcal{R}, \zeta(1) = \state_1 \in \mathcal{S}\backslash \mathcal{R}$. Then $\exists s'\in [0,1]: \zeta(s')\in \closure{\partial \mathcal{R}\backslash\partial\mathcal{S}}$.
\end{lemma}
\begin{proof}
Given $\exists s' \in [0,1]: \zeta(s') \in \partial \mathcal{R}$ and $\zeta(s)\notin \mathcal{R}, s\in (s', s'+\epsilon]$, where $\epsilon>0$ can be arbitrary small. We need to show $\zeta(s')\in \closure{\partial \mathcal{R}\backslash\partial\mathcal{S}}$.

\noindent Case 1: $\zeta(s') \notin \partial \mathcal{S}$ implies $\zeta(s') \in \partial \mathcal{R}\backslash\partial\mathcal{S} \subseteq \closure{\partial \mathcal{R}\backslash\partial\mathcal{S}}$.

\noindent Case 2: $\zeta(s') \in \partial \mathcal{S} \cap \partial \mathcal{R}$. Prove by contradiction. Let's assume $\zeta(s') \notin \closure{\partial \mathcal{R}\backslash\partial\mathcal{S}}$ and using closure definition, this implies $\exists$ sufficiently small $\epsilon>0$, such that 
\begin{align}
    B_{\epsilon}(\zeta(s')) \cap (\partial \mathcal{R}\backslash\partial\mathcal{S}) = \emptyset. \label{eq:assump}
\end{align}
For $\epsilon>0$, $s \in (s',s'+\epsilon], \zeta(s)\in \mathcal{S} \backslash\mathcal{R}$ and $\zeta(s') \in \partial \mathcal{S} \cap \partial \mathcal{R}$ implies $B_{\epsilon}(\zeta(s')) \cap \mathcal{R} \neq B_{\epsilon}(\zeta(s')) \cap \mathcal{S}$. Since the set $\mathcal{R}$ is regular close, a ball around $x \in \partial \mathcal{R}$ with sufficiently small radius $\epsilon>0$, intersected with $\mathcal{R}$ have a non-empty interior and preserves regular close. Additionally, for arbitrary small $\epsilon>0$ the set $B_{\epsilon}(x) \cap \interior{\mathcal{R}}, x \in \partial \mathcal{R}$ is path connected implies for even smaller $\delta>0, x' \in \partial (B_{\epsilon}(x) \cap \mathcal{R})$ the set $B_{\delta}(x') \cap \interior{(B_{\epsilon}(x) \cap \mathcal{R})}$ is also path connected. Same holds with set $\mathcal{S}$ and hence by using \cref{lem:subsetUnequalBound},
\begin{align*}
   \partial (B_{\epsilon}(\zeta(s')) \cap  \mathcal{R}) \neq \partial (B_{\epsilon}(\zeta(s')) \cap  \mathcal{S}),
\end{align*}
and since $\zeta(s') \in \partial \mathcal{S} \cap \partial \mathcal{R}$, this implies,
\begin{align}
    B_{\epsilon}(\zeta(s')) \cap \partial \mathcal{R} \neq B_{\epsilon}(\zeta(s')) \cap \partial \mathcal{S}. \label{eq:leaving}
\end{align}
\looseness -1 From assumption \cref{eq:assump} and \cref{eq:leaving}, we get
$$B_{\epsilon}(\zeta(s')) \cap (\partial \mathcal{S}\backslash\partial\mathcal{R}) \neq \emptyset.$$
\looseness -1 Hence, for a ball around $\zeta(s')$ of any $\epsilon>0$, $\partial \mathcal{S}$ is strict superset of $\partial \mathcal{R}$. Since $\mathcal{S}$ is a regular close-set, the above equation implies there exists a subset of $\mathcal{S}$ that is connected only through a point $\zeta(s')$ (see \Cref{fig: conti-locally-path-connected} left). However, this contradicts that $B_{\epsilon}(\zeta(s')) \cap \interior{\mathcal{S}}$ is path connected. 
\end{proof}
\begin{figure}
\setlength{\abovecaptionskip}{5pt}
    \centering
    \begin{subfigure}[b]{0.2\textwidth}
            \scalebox{0.5}{\tikzset{every picture/.style={line width=0.75pt}} 

\begin{tikzpicture}[x=0.75pt,y=0.75pt,yscale=-1,xscale=1]

\draw  [fill={rgb, 255:red, 177; green, 177; blue, 255 }  ,fill opacity=1 ] (457.96,261.1) -- (538.83,368.16) -- (374.75,366.35) -- cycle ;
\draw  [fill={rgb, 255:red, 177; green, 177; blue, 255 }  ,fill opacity=1 ] (290.5,314) .. controls (299,313) and (305,308) .. (329,330) .. controls (353,352) and (368.5,361.71) .. (374.75,366.35) .. controls (381,371) and (400,403) .. (380.78,446) .. controls (361.56,489) and (287.9,465.8) .. (270.2,452.4) .. controls (252.5,439) and (241.1,349.4) .. (257.5,326.5) .. controls (273.9,303.6) and (282,315) .. (290.5,314) -- cycle ;
\draw  [fill={rgb, 255:red, 108; green, 215; blue, 108 }  ,fill opacity=0.6 ] (374.75,366.35) .. controls (356.6,353) and (325.6,327.4) .. (312.4,328.2) .. controls (299.2,329) and (295,371.5) .. (290,391.5) .. controls (285,411.5) and (309.7,471.4) .. (354.2,457.4) .. controls (398.7,443.4) and (392.9,379.71) .. (374.75,366.35) -- cycle ;
\draw  [line width=3.75]  (312.5,415) .. controls (312.5,414.29) and (313.07,413.72) .. (313.78,413.72) .. controls (314.48,413.72) and (315.06,414.29) .. (315.06,415) .. controls (315.06,415.71) and (314.48,416.28) .. (313.78,416.28) .. controls (313.07,416.28) and (312.5,415.71) .. (312.5,415) -- cycle ;
\draw [line width=1.5]  [dash pattern={on 5.63pt off 4.5pt}]  (312.5,415) .. controls (334,420) and (355.5,382.71) .. (374.75,366.35) .. controls (393.62,350.33) and (423.45,335.56) .. (481.89,350.54) ;
\draw [shift={(485.5,351.5)}, rotate = 195.22] [fill={rgb, 255:red, 0; green, 0; blue, 0 }  ][line width=0.08]  [draw opacity=0] (13.4,-6.43) -- (0,0) -- (13.4,6.44) -- (8.9,0) -- cycle    ;
\draw [shift={(351.64,390.95)}, rotate = 132.29] [fill={rgb, 255:red, 0; green, 0; blue, 0 }  ][line width=0.08]  [draw opacity=0] (13.4,-6.43) -- (0,0) -- (13.4,6.44) -- (8.9,0) -- cycle    ;
\draw [shift={(433.37,344.38)}, rotate = 176.69] [fill={rgb, 255:red, 0; green, 0; blue, 0 }  ][line width=0.08]  [draw opacity=0] (13.4,-6.43) -- (0,0) -- (13.4,6.44) -- (8.9,0) -- cycle    ;
\draw  [line width=3.75]  (482.94,351.5) .. controls (482.94,350.79) and (483.52,350.22) .. (484.22,350.22) .. controls (484.93,350.22) and (485.5,350.79) .. (485.5,351.5) .. controls (485.5,352.21) and (484.93,352.78) .. (484.22,352.78) .. controls (483.52,352.78) and (482.94,352.21) .. (482.94,351.5) -- cycle ;
\draw   (340.05,366.35) .. controls (340.05,347.19) and (355.59,331.65) .. (374.75,331.65) .. controls (393.91,331.65) and (409.45,347.19) .. (409.45,366.35) .. controls (409.45,385.52) and (393.91,401.06) .. (374.75,401.06) .. controls (355.59,401.06) and (340.05,385.52) .. (340.05,366.35) -- cycle ;
\draw  [line width=3.75]  (374.75,366.35) .. controls (374.75,365.65) and (375.32,365.08) .. (376.03,365.08) .. controls (376.73,365.08) and (377.31,365.65) .. (377.31,366.35) .. controls (377.31,367.06) and (376.73,367.63) .. (376.03,367.63) .. controls (375.32,367.63) and (374.75,367.06) .. (374.75,366.35) -- cycle ;

\draw (314.7,362.23) node [anchor=north west][inner sep=0.75pt]  [font=\LARGE]  {$\mathcal{R}$};
\draw (322.12,419.9) node [anchor=north west][inner sep=0.75pt]  [font=\LARGE]  {$x_{0}$};
\draw (265.7,352.23) node [anchor=north west][inner sep=0.75pt]  [font=\LARGE]  {$\mathcal{S}$};
\draw (397.19,389.68) node [anchor=north west][inner sep=0.75pt]  [font=\LARGE]  {$B_{\epsilon }\left( \zeta \left( s^{'}\right)\right)$};
\draw (490.12,327.9) node [anchor=north west][inner sep=0.75pt]  [font=\LARGE]  {$x_{1}$};
\draw (361.33,333.13) node [anchor=north west][inner sep=0.75pt]  [font=\LARGE]  {$\zeta \left( s^{'}\right)$};

\end{tikzpicture}}
    \end{subfigure}\hfill
    \begin{subfigure}[b]{0.2\textwidth}
     \hspace{-2.5em}       \scalebox{0.5}{\tikzset{every picture/.style={line width=0.75pt}} 

\begin{tikzpicture}[x=0.75pt,y=0.75pt,yscale=-1,xscale=1]

\draw  [fill={rgb, 255:red, 177; green, 177; blue, 255 }  ,fill opacity=1 ] (319.33,66) .. controls (356.33,94) and (256.25,109.65) .. (264,148) .. controls (271.75,186.35) and (264,203) .. (253.78,239) .. controls (243.56,275) and (148.7,292.4) .. (131,279) .. controls (113.3,265.6) and (72.33,134.33) .. (102.33,98.33) .. controls (132.33,62.33) and (282.33,38) .. (319.33,66) -- cycle ;
\draw  [fill={rgb, 255:red, 108; green, 215; blue, 108 }  ,fill opacity=0.6 ] (136.67,109.67) .. controls (145.17,108.67) and (256.25,109.65) .. (264,148) .. controls (271.75,186.35) and (264,203) .. (253.78,239) .. controls (243.56,275) and (154.7,278.4) .. (137,265) .. controls (119.3,251.6) and (112.6,194.9) .. (129,172) .. controls (145.4,149.1) and (128.17,110.67) .. (136.67,109.67) -- cycle ;
\draw   (176.66,148) .. controls (176.66,99.77) and (215.77,60.66) .. (264,60.66) .. controls (312.23,60.66) and (351.34,99.77) .. (351.34,148) .. controls (351.34,196.23) and (312.23,235.34) .. (264,235.34) .. controls (215.77,235.34) and (176.66,196.23) .. (176.66,148) -- cycle ;
\draw   (227.95,100.29) .. controls (227.95,84.72) and (240.57,72.1) .. (256.14,72.1) .. controls (271.71,72.1) and (284.33,84.72) .. (284.33,100.29) .. controls (284.33,115.86) and (271.71,128.48) .. (256.14,128.48) .. controls (240.57,128.48) and (227.95,115.86) .. (227.95,100.29) -- cycle ;
\draw   (223.36,126.83) .. controls (223.36,114.98) and (232.97,105.37) .. (244.83,105.37) .. controls (256.68,105.37) and (266.29,114.98) .. (266.29,126.83) .. controls (266.29,138.69) and (256.68,148.3) .. (244.83,148.3) .. controls (232.97,148.3) and (223.36,138.69) .. (223.36,126.83) -- cycle ;
\draw  [line width=3.75]  (262.72,146.72) .. controls (262.72,146.02) and (263.29,145.44) .. (264,145.44) .. controls (264.71,145.44) and (265.28,146.02) .. (265.28,146.72) .. controls (265.28,147.43) and (264.71,148) .. (264,148) .. controls (263.29,148) and (262.72,147.43) .. (262.72,146.72) -- cycle ;
\draw  [line width=3.75]  (254.86,100.29) .. controls (254.86,99.58) and (255.43,99.01) .. (256.14,99.01) .. controls (256.85,99.01) and (257.42,99.58) .. (257.42,100.29) .. controls (257.42,101) and (256.85,101.57) .. (256.14,101.57) .. controls (255.43,101.57) and (254.86,101) .. (254.86,100.29) -- cycle ;
\draw  [line width=3.75]  (243.55,126.83) .. controls (243.55,126.13) and (244.12,125.55) .. (244.83,125.55) .. controls (245.53,125.55) and (246.11,126.13) .. (246.11,126.83) .. controls (246.11,127.54) and (245.53,128.11) .. (244.83,128.11) .. controls (244.12,128.11) and (243.55,127.54) .. (243.55,126.83) -- cycle ;

\draw (137.7,223.23) node [anchor=north west][inner sep=0.75pt]  [font=\LARGE]  {$\mathcal{R}$};
\draw (101.7,135.23) node [anchor=north west][inner sep=0.75pt]  [font=\LARGE]  {$\mathcal{S}$};
\draw (286.91,120.3) node [anchor=north west][inner sep=0.75pt]  [font=\LARGE]  {$B_{\epsilon }( x)$};
\draw (261.86,78.01) node [anchor=north west][inner sep=0.75pt]  [font=\LARGE]  {$x$};
\draw (184.57,154.3) node [anchor=north west][inner sep=0.75pt]  [font=\LARGE]  {$B_{\delta }( x')$};
\draw (235.19,131.68) node [anchor=north west][inner sep=0.75pt]  [font=\large]  {$P1$};
\draw (239.19,109.68) node [anchor=north west][inner sep=0.75pt]  [font=\large]  {$P2$};

\end{tikzpicture}}
    \end{subfigure}
    \caption{On the left, we present a contradictory case for \cref{lem: boundaryDiff}, i.e., the set $B_{\epsilon}(\zeta(s')) \cap \interior{\mathcal{S}}$ is not a path connected set. On the right, we illustrate proof steps for \cref{lem:subsetUnequalBound}. For regular close and locally path-connected interior sets $\mathcal{R}$ and $\mathcal{S}$, if $\mathcal{R}$ is a strict subset of $\mathcal{S}$, then their boundaries are unequal.}
    \label{fig: conti-locally-path-connected}
    \vspace{-1.5em}
\end{figure}
\begin{lemma} \label{lem:subsetUnequalBound} Let the regular close sets $\calP = \{\mathcal{R}, \mathcal{S} \} \subseteq \R^\sdim$ additionally satisfy $\forall p \in \partial \calP,$ there exists an $\epsilon>0$ the set $B_{\epsilon}(p) \cap \interior{\calP}$ is path connected. Then $\mathcal{S} \backslash \mathcal{R} \neq \emptyset \implies \partial \mathcal{R} \neq \partial \mathcal{S}$. 
\end{lemma}
\begin{proof} 
Proof by contradiction. Assume $\partial \mathcal{R} = \partial \mathcal{S}$. Since $\mathcal{S} \backslash \mathcal{R} \neq \emptyset$, $\exists x \in \mathcal{S} \backslash \mathcal{R} $. Consider the following two cases:\\
i) $x \in \partial \mathcal{S} \implies x \in \partial \mathcal{R}$, $x\notin \mathcal{S} \backslash \mathcal{R}$. (contradiction)\\
ii) $x \in \interior{\mathcal{S}} \implies \exists \epsilon: B_{\epsilon}(x) \cap \partial R \neq \emptyset$ and $\interior{B_{\epsilon}(x)} \subseteq \interior{\mathcal{S}} \backslash \interior{\mathcal{R}}$. 

Suppose the closed ball touches $\partial \mathcal{R}$ or $\partial \mathcal{S}$ at location $x'$ (see \Cref{fig: conti-locally-path-connected} right). With an arbitrary small $\delta>0$, consider a ball $B_{\delta}(x')$ at $x'$. The boundary $\partial \mathcal{R}$ divides the ball $B_{\delta}(x')$ in at least two parts.\\  $P_1: B_{\delta}(x') \cap \interior{\mathcal{R}} \subseteq \interior{\mathcal{S}}$ ($P_1 \neq \emptyset$ since $x'\in \partial \mathcal{R}$ and $\mathcal{R}$ is regular close. Subset hold since $\mathcal{R} \subseteq \mathcal{S}$).\\  $P_2: B_{\delta}(x') \cap \interior{B_{\epsilon}(x)} \subseteq \interior{\mathcal{S}}$ (since $\interior{B_{\epsilon}(x)} \subseteq \interior{\mathcal{S}}$). 

\looseness -1 Since $x' \in \partial \mathcal{S}$, we know $B_{\delta}(x') \cap \interior{\mathcal{S}}$ is path connected. This implies $\exists \zeta: [0,1] \to B_{\delta}(x') \cap \interior{\mathcal{S}}, \zeta(0)\in P_1 , \zeta(1) \in P_2$ and $\zeta(s)\!\notin \partial \mathcal{S}, s\in [0,1]$. However $P_1\!\subseteq\!\interior{R}$ but $P_2 \cap \interior{\mathcal{R}} = \emptyset$ (since $ \interior{B_{\epsilon}(x)} \cap \interior{\mathcal{R}} = \emptyset $). This implies $\exists s\in [0,1]: \zeta(s) \in 
\partial \mathcal{R}$. Since $\partial \mathcal{S} = \partial \mathcal{R}$, this contradicts that $\zeta(s)\notin \partial \mathcal{S}, s\in [0,1]$.
\end{proof}

\begin{lemma}
    \label{lem: subtract_domain_boundary} \looseness -1 Let \cref{assump: regularconnected} holds. Consider a path 
 $\zeta:[0,1] \to \Rcontoper[\safeInit{\n}]{\optiSet[]{\n}}$ such that $\zeta(0) \in \safeInit{\n}$, $\zeta(1) \in \safeInit{\n}$ and $\zeta(b^\star) \notin \Rcontoper[\safeInit{\n}]{\pessiSet[]{\n}}$ for some $b^\star \in [0,1]$. There exists $b' \in [0,b^{\star}): \zeta(b')\in \partial \Rcontoper[\safeInit{\n}]{\pessiSet[]{\n}}$. Then $\zeta(b') \in \closure{\partial \Rcontoper[\safeInit{\n}]{\pessiSet[]{\n}} \backslash \partial \Rcontoper[\safeInit{\n}]{\X}}$.
\end{lemma}
\begin{proof}
Note that $\Rcontoper[\safeInit{\n}]{\pessiSet[]{\n}} \subseteq \Rcontoper[\safeInit{\n}]{\X}$ and $\forall b \in [0,1], \zeta(b) \in \Rcontoper[\safeInit{\n}]{\X}$ with $\zeta(0) \in \Rcontoper[\safeInit{\n}]{\pessiSet[]{\n}}, \zeta(b^\star) \in \Rcontoper[\safeInit{\n}]{\X} \backslash \Rcontoper[\safeInit{\n}]{\pessiSet[]{\n}}$. 
Moreover due to \cref{assump: regularconnected}, $\Rcontoper[\safeInit{\n}]{\pessiSet[]{\n}}$ and $ \Rcontoper[\safeInit{\n}]{\X}$ are regular closed sets and satisfy the local path connected interior property, thus, by \cref{lem: boundaryDiff}, we get, 
$\zeta(b') \in \closure{\partial \Rcontoper[\safeInit{\n}]{\pessiSet[]{\n}} \backslash \partial \Rcontoper[\safeInit{\n}]{\X}}$.
Intuitively, this represents that the path $\zeta$ can cross the boundary of the reachable returnable pessimistic set only from the locations that lead it into the domain's reachable returnable set.
\end{proof}

\begin{lemma} \label{lem: boundary-pessi-zero}
$\forall x' \in ~ \closure{\partial \pessiSet[]{\n} \backslash \partial \Domain}, \lbconst[\n](x')=0$.
\end{lemma}
\begin{proof} By defi. $\pessiSet[]{\n} = \{\state \in \Domain | \lbconst[\n](\state) \geq 0 \}$. Hence $\partial \pessiSet[]{\n} \subseteq \{\state \in \Domain | \lbconst[\n](\state) = 0\} \cup \partial \Domain \!\implies \! \closure{\partial \pessiSet[]{\n} \backslash \partial \Domain} \subseteq \{\state \in \Domain | \lbconst[\n](\state) = 0\}$.\!\!\!\!\!\!\!
\end{proof}

\revdisp{
\begin{proof}[\textbf{Proof for \cref{thm:SE_exp}}]
\looseness -1 Analogous to \cref{coro:sample_complexity}, the sample complexity result from \cref{thm:sample_complexity} upper bounds the number of samples required until the infeasibility of the sampling rule \eqref{eqn:sampling_strategy_expander}. Thus $\exists \nfin \leq \n^\star : \ubconst[\nfin](\state) - \lbconst[\nfin](\state) < \epsconst \forall \state \in \expansion[\nfin](\safeInit{\nfin}, \LpessiSet[]{\nfin})$. Now it suffice to show if $\ubconst[\nfin](\state) - \lbconst[\nfin](\state) < \epsconst \forall \state \in \expansion[\nfin](\safeInit{\nfin}, \LpessiSet[]{\nfin}) $ then $\Rcontoper[\safeInit{\nfin}]{\constSet[,\epsconst]{}} \subseteq \Rcontoper[\safeInit{\nfin}]{\optiSet[]{\nfin}}  \subseteq \Rcontoper[\safeInit{\nfin}]{\LpessiSet[]{\nfin}} \subseteq \Rcontoper[\safeInit{\nfin}]{\constSet[]{}}$.

\looseness -1 We first prove $\Rcontoper[\safeInit{\nfin}]{\optiSet[]{\nfin}}\subseteq \Rcontoper[\safeInit{\nfin}]{\LpessiSet[]{\nfin}}$ by contradiction. Let's assume $\exists \ x^\star \in \Rcontoper[\safeInit{\nfin}]{\optiSet[]{\nfin}} \backslash \Rcontoper[\safeInit{\nfin}]{\LpessiSet[]{\nfin}}$. This implies there exists a continuously path $\zeta(b) \in \Rcontoper[\safeInit{\nfin}]{\optiSet[]{\nfin}} \forall b \in [0,1]$ such that $\zeta(0) \in \safeInit{\nfin}, \zeta(b^{\star}) =\state^\star$ and $\zeta(1) \in \safeInit{\nfin}$. Please see \cref{fig: conti-opti-in-pessi} for a visual description. 

\looseness -1 Since $\zeta(b^\star) \notin \Rcontoper[\safeInit{\nfin}]{\LpessiSet[]{\nfin}}$, $\exists~ b'<b^{\star}: \zeta(b')\in \partial \Rcontoper[\safeInit{\nfin}]{\LpessiSet[]{\nfin}}$ which using \cref{lem: subtract_domain_boundary} for $\LpessiSet[]{\nfin}$ analogous to $\pessiSet[]{\nfin}$ implies $\zeta(b') \in \closure{\partial \Rcontoper[\safeInit{\nfin}]{\LpessiSet[]{\nfin}} \backslash \partial \Rcontoper[\safeInit{\nfin}]{\X}}$. Furthermore, \cref{lem:boundary_inclusion} implies $\zeta(b') \in \closure{\partial \LpessiSet[]{\nfin} \backslash \partial \X}$ and finally using \cref{lem: boundary-pessi-zero} we get  $\max_{x\in\X} \lbconst[\nfin](x) - \LipConst \| x - \zeta(b') \|=0 \implies \lbconst[\nfin](\zeta(b')) \leq 0$. Notably \cref{lem:boundary_inclusion,lem: subtract_domain_boundary,lem: boundary-pessi-zero} considers $\pessiSet[]{\nfin}$ but analogously follows for $\LpessiSet[]{\nfin}$.

\looseness -1 Since, $\zeta(b') \in \optiSet[]{\nfin} \implies \ubconst[\nfin](\zeta(b')) - \epsconst \geq 0$. Hence from the above two equations, $\ubconst[\nfin](\zeta(b')) - \lbconst[\nfin](\zeta(b'))  \geq \epsconst$. Using definition \eqref{eq:expander},   
note that $\zeta(b') \in \expansion[\nfin](\safeInit{\nfin}, \LpessiSet[]{\nfin}) $. 
This implies $\ubconst[\nfin](\zeta(b')) - \lbconst[\nfin](\zeta(b')) < \epsconst$, which is a contradiction. Hence $\exists \nfin \leq \n^\star: \Rcontoper[\safeInit{\nfin}]{\optiSet[]{\nfin}}\subseteq \Rcontoper[\safeInit{\nfin}]{\LpessiSet[]{\nfin}}$, which yields the second set inclusion in \cref{eqn:objective} with enlarged pessimistic set $\LpessiSet[]{\nfin}$. Moreover, note that using \cref{coro:hp_bounds}, $\forall n \geq 0, \constSet[,\epsconst]{} \subseteq \optiSet[]{\n} ~\mathrm{and}~ \LpessiSet[]{\n} \subseteq \constSet[]{}$, which yields the other two set inclusions in \cref{eqn:objective} with enlarged pessimistic set $\LpessiSet[]{\nfin}$.
\end{proof}}

\subsection{Mutual Information}
\label{apx:mutual_info}

The mutual information between $Y_{X_\n}$ and $\constrain_{X_\n}$ as in \cite{beta-srinivas} is given in \cref{sec:backgroundGP}.
The Shannon entropy for a Gaussian, \mbox{$H(\N(\mu, \Sigma)) = \frac{1}{2} \log |2 \pi e \Sigma|$}. Using that \mbox{$x_1, \hdots, x_\n$} are deterministic conditioned on \mbox{$Y_{X_{\n-1}}$} and the predictive distribution for \mbox{$y_{x_\n}$} conditioned on \mbox{$Y_{X_{\n-1}}$} is a Gaussian \mbox{$\N(0, \noisevar^2 \! +\! \sigconst[\n-1]^2(\state_\n))$}, we get,
\begin{align*}
    H(Y_{X_\n}\!)\!&=\!H(y_{x_\n}|Y_{X_{\n-1}}) + H(Y_{X_{\n-1}}) \\
    &= \!\frac{1}{2}\! \log 2\pi e (\noisevar^2 \! +\! \sigconst[\n-1]^2(\state_\n)) \!+\! H(y_{x_{\n-1}}\!|Y_{X_{\n-2}}) +\! ... \\ 
    &= \frac{1}{2} \log(2\pi e \noisevar^2) +  \frac{1}{2} \log (1  + \noiseconst \sigconst[\n-1]^2(\state_\n)) \\ 
    &  \qquad \qquad + H(y_{x_{\n-1}}|Y_{X_{1:\n-2}}) + ... \numberthis \label{eqn: refactoring-det} \\
    &= \underbrace{\frac{1}{2} \sum_{i=1}^\n \log(2\pi e \noisevar^2)}_{H(Y_{X_\n}| \constrain_{X_\n})} +  \frac{1}{2} \sum_{i=1}^\n \log (1  + \noiseconst \sigconst[i-1]^2(\state_i)).
\end{align*}
$H(Y_{X_\n}|\constrain_{X_\n})$ is entropy due to Gaussian noise $\N(0,\noisevar^2)$ and finally using mutual information definition \eqref{eqn: mutual-info-definition} 
we get,
\begin{align*}
I(Y_{X_\n};\constrain_{X_\n}) &= \frac{1}{2} \sum_{i=1}^\n \log (1  + \noiseconst \sigconst[i-1]^2(\state_i)). \numberthis \label{eqn: mutual-info}
\end{align*}
\looseness -1 Notably, Mutual information $I(Y_{X_\n};\constrain_\n)$ is directly tied to the GP used to model the constraint $\constrain$ and depends only on the locations sampling location $\state$. It does not depend on the true constraint function $\constrain$ or the noisy realizations $y$.  Therefore, we derive mutual information for the Gaussian process model using i.i.d. Gaussian noise with a standard deviation of $\noisevar$ and use it to bound the GP's variance in \cref{eqn:non-dec-beta}.
\vspace{-2em}
\begin{IEEEbiography}[{\includegraphics[width=1in,height=1.25in,clip,keepaspectratio]{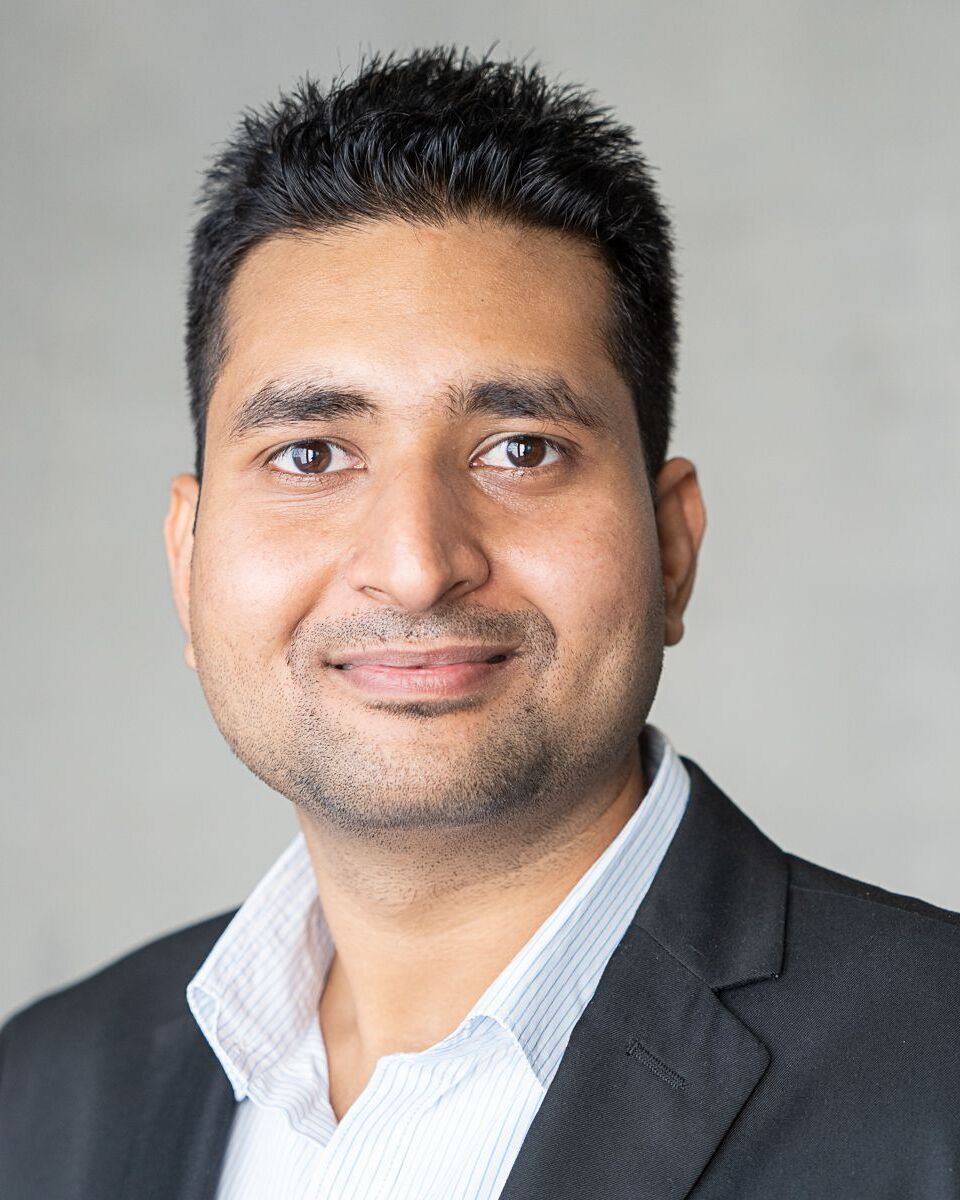}}]{Manish Prajapat}  \looseness -1 is a Doctoral Fellow at ETH AI Center. He earned his Master's degree in robotics, systems, and control from ETH Zurich in 2020, and his Bachelor's degree in Mechanical Engineering from the Indian Institute of Technology (IIT), Madras in 2017. At IIT Madras, he was recognized as the best graduating student co-curricular 2017 and also received the Sivasailam Merit Prize for the best bachelor thesis.
He was a visiting scholar at Caltech in 2020 and later was a research engineer at Fixposition AG, Zurich in 2021. His research interests are sequential decision-making under complex scenarios, e.g., non-Markovian objectives, unknown constraints or unknown dynamics of non-linear systems.
\end{IEEEbiography}\vspace{-3em}
\begin{IEEEbiography}[{\includegraphics[width=1in,height=1.25in,clip,keepaspectratio]{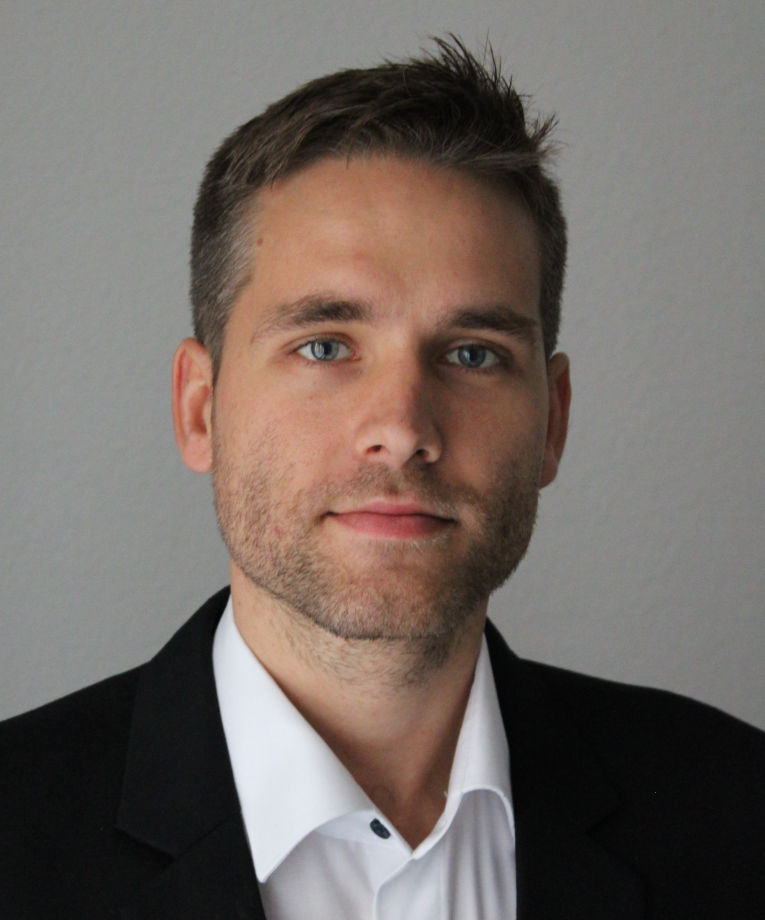}}]{Johannes K\"ohler} \looseness -1 received his Master degree in Engineering Cybernetics from the University of Stuttgart, Germany, in 2017. 
In 2021, he obtained a Ph.D. in mechanical engineering, also from the University of Stuttgart,
Germany, for which he received the 2021 European Systems \& Control PhD award.
He is currently a postdoctoral researcher at the Institute for Dynamic Systems and Control (IDSC) at ETH Zürich. 
His research interests are in the area of model predictive control, as well as control and estimation for nonlinear uncertain systems.
\end{IEEEbiography}\vspace{-3em}
\begin{IEEEbiography}[{\includegraphics[width=1in,height=1.25in,clip,keepaspectratio]{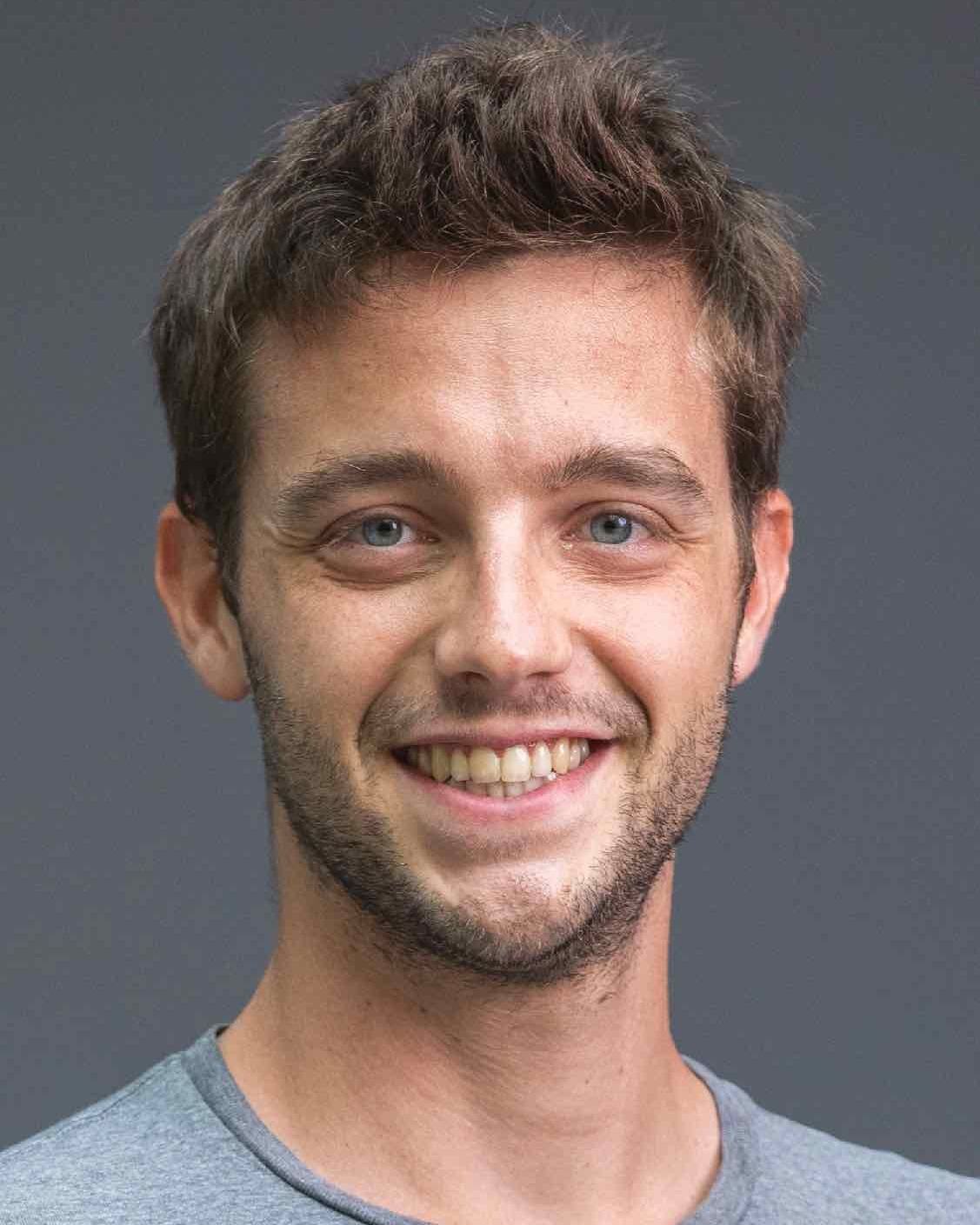}}]{Matteo Turchetta} \looseness -1 earned his Master's degree in Robotics Systems and Control from ETH Zürich in 2016. In 2021, he successfully completed his Ph.D. in Computer Science, also at ETH Zürich. Currently, Matteo serves as a postdoctoral researcher at the Learning and Adaptive Systems (LAS) group at ETH Zürich. His research focus lies in the domain of learning-based control and decision-making under uncertainty, with a particular emphasis on safety-constrained reinforcement learning.
\end{IEEEbiography}\vspace{-3em}
\begin{IEEEbiography}[{\includegraphics[width=1in,height=1.25in,clip,keepaspectratio]{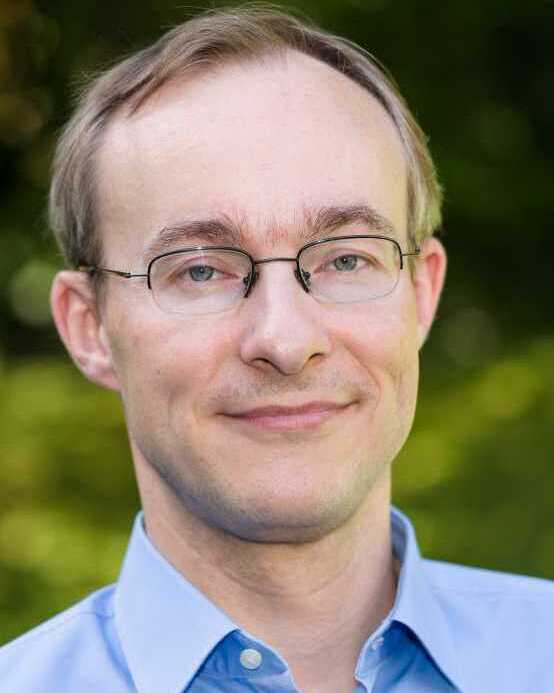}}]{Andreas Krause} 
\looseness -1 is a Professor of Computer Science at ETH Zurich, where he also serves as Chair of the ETH AI Center. Before that, he was an Assistant Professor of Computer Science at Caltech. He received his Ph.D. in Computer Science from Carnegie Mellon University (2008) and his Diploma in Computer Science and Mathematics from the Technical University of Munich, Germany (2004). 
He is an ACM Fellow, ELLIS Fellow, and his research on machine learning and adaptive systems has received multiple awards, including the ACM SIGKDD Test of Time award 2019 and the ICML Test of Time award 2020. He served as Program Co-Chair for ICML 2018, General Chair for ICML 2023 and Action Editor for the Journal of Machine Learning Research. In 2023, he was appointed to the United Nations’ High-level Advisory Body on AI.
\end{IEEEbiography}\vspace{-3em}
\begin{IEEEbiography}[{\includegraphics[width=1in,height=1.25in,clip,keepaspectratio]{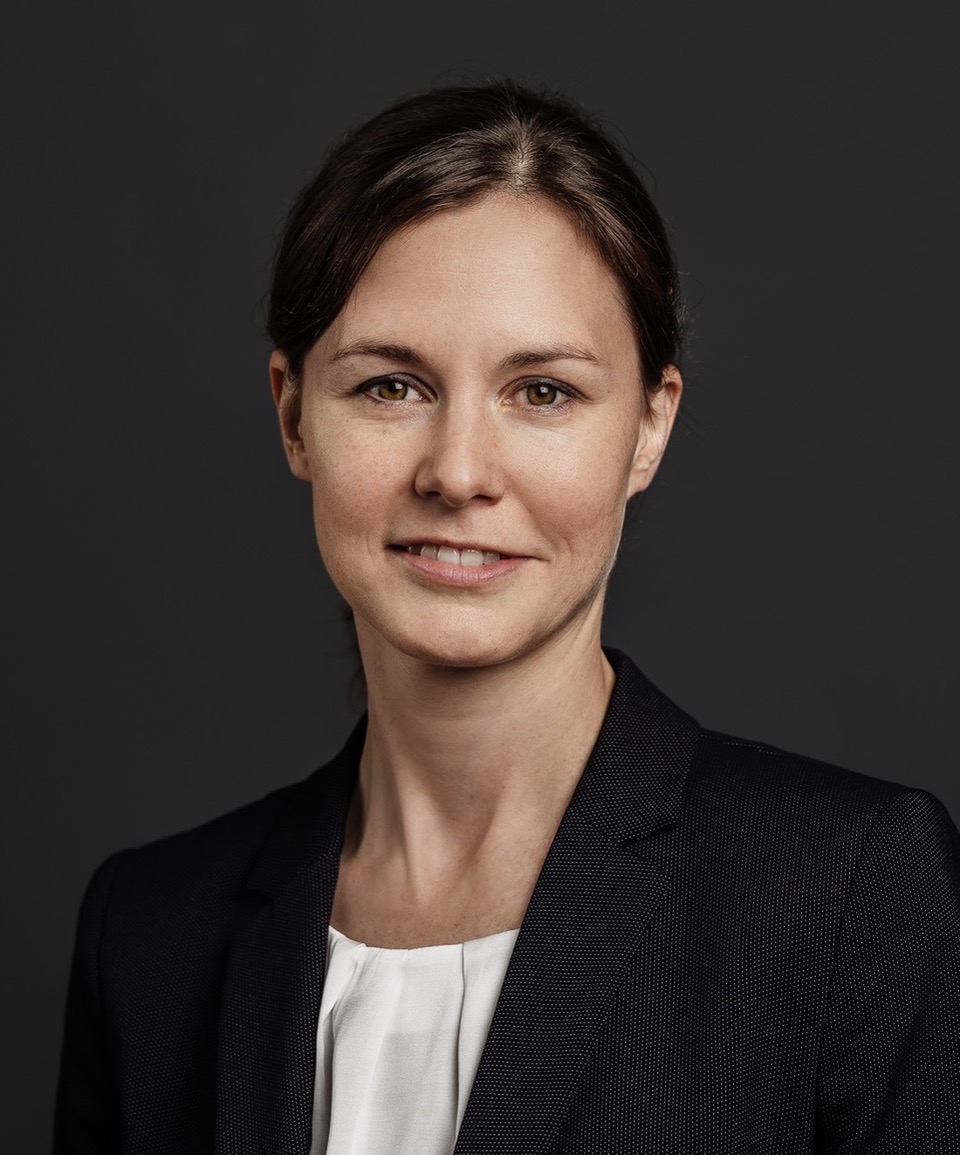}}]{Melanie N. Zeilinger}  \looseness -1 is an Associate Professor at ETH Zurich, Switzerland. She received the Diploma degree in engineering cybernetics from the University of Stuttgart, Germany, in 2006, and the Ph.D. degree with honors in electrical engineering from ETH Zurich, Switzerland, in 2011. From 2011 to 2012 she was a Postdoctoral Fellow with the Ecole Polytechnique Federale de Lausanne (EPFL), Switzerland. She was a Marie Curie Fellow and Postdoctoral Researcher with the Max Planck Institute for Intelligent Systems, Tübingen, Germany until 2015 and with the Department of Electrical Engineering and Computer Sciences at the University of California at Berkeley, CA, USA, from 2012 to 2014. From 2018 to 2019 she was a professor at the University of Freiburg, Germany. She was awarded the ETH medal for her PhD thesis, an SNF Professorship, the ETH Golden Owl for exceptional teaching in 2022 and the European Control Award in 2023. Her research interests include learning-based control with applications to robotics and human-in-the-loop control.
\end{IEEEbiography}
\end{document}